\documentclass[12pt,pdftex,noinfoline]{imsart}

\RequirePackage[OT1]{fontenc}
\usepackage{amsthm,amsmath,amsfonts,natbib,mathtools,amssymb}
\RequirePackage{hypernat}
\usepackage[ruled,section]{algorithm}
\usepackage{algorithmic}
\usepackage{graphicx}
\usepackage{verbatim}
\usepackage{times}
\usepackage[T1]{fontenc}
\usepackage[text={6.0in,8.6in},centering]{geometry}
\usepackage{subfigure}
\usepackage{graphicx}
\usepackage{hyperref}
\usepackage{yhmath}
\usepackage[usenames,dvipsnames]{xcolor}
\definecolor{shadecolor}{gray}{0.9}
\definecolor{shadecolor}{gray}{0.9}
\hypersetup{citecolor=MidnightBlue}
\hypersetup{linkcolor=MidnightBlue}
\hypersetup{urlcolor=MidnightBlue}
\usepackage{enumerate}
\usepackage{fullpage}
\usepackage{tikz}
\usetikzlibrary{decorations.pathreplacing}
\usetikzlibrary{arrows,positioning}

\let\hat\widehat
\let\tilde\widetilde
\def\Cov{\text{Cov}} 
\def\Var{\text{Var}} 
\def\Corr{\text{Corr}}

\def\text#1{\mbox{\rm #1}}

\newcommand{\argmin}{\mathop{\rm argmin}}

\newcommand{\wt}{\widetilde}

\newcommand{\fnorm}[1]{\|#1\|_{\rm F}}

\newcommand{\Tr}{\mathop{\text{Tr}}}

\newcommand{\supp}{{\rm supp}}
\newcommand{\iprod}[2]{\left \langle #1, #2 \right\rangle}

\newcommand{\TV}{{\sf TV}}

\theoremstyle{plain}
\newtheorem{theorem}{Theorem}[section]
\newtheorem{proposition}[theorem]{Proposition}
\newtheorem{lemma}[theorem]{Lemma}
\newtheorem{corollary}[theorem]{Corollary}
\theoremstyle{remark}

\theoremstyle{definition}

\def\P{{\mathbb P}}

\def\supp{\mathop{\text{supp}\kern.2ex}}
\def\argmin{\mathop{\text{arg\,min}\kern.2ex}}
\let\hat\widehat
\let\tilde\widetilde

\def\shape#1{
  \lower5pt\hbox{
  \hskip-7pt
  \tikzset{circ/.style={circle, draw, fill=black, scale=.2}}
  \begin{tikzpicture}[semithick,scale=.3]
  \node (l1) at (0,.866) [circ]{};
  \node (l2) at (1,.866) [circ]{};
  \node (l3) at (0.5,0) [circ]{};
  #1
  \end{tikzpicture}
  \hskip-8pt}
}

\def\Vshape{
  \draw[-,color=white] (l1) to node [auto] {} (l2);
  \draw[-] (l1) to node [auto] {} (l3);
  \draw[-] (l2) to node [auto] {} (l3);
}

\def\emptyshape{
  \draw[-,color=white] (l1) to node [auto] {} (l3);
  \draw[-,color=white] (l2) to node [auto] {} (l3);
  \draw[-,color=white] (l1) to node [auto] {} (l2);
}

\def\triangleshape{
  \draw[-] (l1) to node [auto] {} (l3);
  \draw[-] (l2) to node [auto] {} (l3);
  \draw[-] (l1) to node [auto] {} (l2);
}

\def\edgeshape{
  \draw[-] (l1) to node [auto] {} (l3);
  \draw[-,color=white] (l2) to node [auto] {} (l3);
  \draw[-,color=white] (l1) to node [auto] {} (l2);
}

\numberwithin{equation}{section}
\graphicspath{{./figs/}}

\begin{document}

\begin{frontmatter}
\title{\Large Testing Network Structure Using Relations Between \vskip-5pt Small Subgraph Probabilities}
\runauthor{Gao and Lafferty}

\begin{aug}
\vskip10pt
\author{\fnms{Chao}
  \snm{Gao${}^{*}$}\ead[label=e1]{chaogao@galton.uchicago.edu}}
 \and 
\author{\fnms{John}
  \snm{Lafferty${}^{*\dag}$}\ead[label=e4]{lafferty@galton.uchicago.edu}}
\vskip10pt
\address{
\begin{tabular}{c}
${}^*$Department of Statistics \\
${}^\dag$Department of Computer Science \\
University of Chicago 
\end{tabular}
\\[10pt]
\today\\[5pt]
\vskip10pt
}
\end{aug}

\begin{abstract}
We study the problem of testing for structure in networks using
relations between the observed frequencies of small subgraphs.
We consider the statistics
\begin{align*}
T_3 & =(\text{edge frequency})^3 - \text{triangle frequency}\\
T_2 & =3(\text{edge frequency})^2(1-\text{edge frequency}) - \text{V-shape frequency}
\end{align*}
and prove a central limit theorem for $(T_2, T_3)$ under an
Erd\H{o}s-R\'{e}nyi null model. We then 
analyze the power of the associated $\chi^2$ test statistic under a general class of
alternative models. In particular, when the alternative is a
$k$-community stochastic block model, with $k$ unknown, the power of
the test approaches one.  Moreover, the signal-to-noise ratio
required is strictly weaker than that required for community
detection.  We also study the relation with other statistics over
three-node subgraphs, and analyze the error under two natural
algorithms for sampling small subgraphs.  Together, our results show
how global structural characteristics of networks can be inferred from
local subgraph frequencies, without requiring the global community
structure to be explicitly estimated.
\end{abstract}
 
\vskip20pt 
\end{frontmatter}

\maketitle

\vskip10pt

\let\weaklyto\rightsquigarrow

\section{Introduction}

The statistical properties of graphs and networks have been
intensively studied in recent years, resulting in a rich and detailed
body of knowledge on stochastic graph models. Examples include graphons
and the stochastic block model \citep{holland83,lovasz12}, preferential
attachment models \citep{price76,barabasi99}, and
other generative network models \citep{bollobas01}. This work often seeks to model the
network structures observed in ``naturally occurring'' settings, such
as social media.  A focus has been to develop simple models that
can be rigorously studied, while still capturing some of the phenomena
observed in actual data. Related work has developed procedures to find
structure in networks, for example using spectral algorithms for finding
communities \citep{rohe2011spectral,jin2015,arias-castro2014}.  Another line of research has studied estimation and
detection of signals on graphs where
the structure of the signal is exploited to develop efficient
procedures \citep{dfs16,arias-castro2011}.

In this work our focus is different. We are interested in studying how
global structural properties of graphs and networks might be inferred
from purely local properties. In the absence of a probability model to
generate the graph, this is a classical mathematical topic. For
example, the Euler characteristic relates a graph invariant (genus) to
local graph counts (number of edges and faces). More general relations
between local structure and global invariants lie at the heart of
combinatorics and algebraic topology; topological data analysis is the
study of such relations under a data sampling model.  In this paper we
study how the presence of communities in a network might be detected
from deviations in small subgraph probabilities, focusing on
invariants between the probabilities of small subgraphs, such as edges
and triangles, that must hold for unstructured random graphs.

Our work is inspired by \cite{ugander2013subgraph}, who present striking data on the
empirical distributions of 3-node and 4-node subgraphs of Facebook friend
networks, comparing them to the distributions that would be obtained under an
Erd\H{o}s-R\'{e}nyi model. 
Notably, the subgraph frequencies of the Facebook subnetworks appear surprisingly close
to the corresponding probabilities under an Erd\H{o}s-R\'{e}nyi model, even though the
friend networks are expected to exhibit community structure.
However, the subgraph frequencies are not arbitrary. In fact, the global graph structure places
purely combinatorial restrictions on the subgraph probabilities,
sometimes called homomorphism constraints \citep{razborov2008}.
The interaction between these aspects are discussed by \cite{ugander2013subgraph}, who pose
the broad research question ``What properties of social graphs are `social'
properties and what properties are `graph' properties?''  Their work develops two
complementary methods to shed light on this question. First, they propose a
generative model that extends the Erd\H{o}s-R\'{e}nyi model and better matches the
empirical data. Second, they develop methods to bound the homomorphism
constraints that determine the feasible space of subgraph probabilities.

In the present work we take a complementary, statistical approach to
distinguishing graphs with social structure from random graphs using only small
subgraph frequencies, framing the problem in terms of
statistical testing.  We consider several questions, focusing on
the test statistics
\begin{align*}
T_3 & =(\text{edge frequency})^3 - \text{triangle frequency}\\
T_2 & =3(\text{edge frequency})^2(1-\text{edge frequency}) - \text{V-shape frequency}
\end{align*}
Noting, for example, the population invariant  $p^3 - \P(\text{triangle}) = 0$ under 
an Erd\H{o}s-R\'{e}nyi null model, we study these two statistics for testing 
against different alternative network models.  The
following is a high-level summary of our findings; precise statements are given
later.

\begin{itemize}

\item We show that, under appropriate normalization, $(T_2,T_3) \weaklyto  \text{multivariate normal}$.

\item Under a fairly general class of alternative models, the power of
the natural test using $T_2$ and $T_3$ approaches one.

\item When the alternative is a stochastic block model with two communities, the
power approaches one under the optimal scaling of the signal-to-noise
ratio for community detection.

\item When the alternative is a $k$-community stochastic block model, with $k$ unknown, the power of the
test again approaches one.  But the signal-to-noise ratio required is strictly
weaker than that required for community detection.

\item The sampling error is analyzed under two natural algorithms for sampling
small subgraphs.
\end{itemize}

Together, our results show how global structural characteristics of networks can be
inferred from local subgraph frequencies, without requiring the global community
structure to be explicitly estimated. Recent work of
\cite{maugis2017} considers a closely related problem of using small
subgraph counts to test whether a collection of networks is drawn from
a known graphon null model. \cite{bubeck2014} study tests based on 
signed triangles for distinguishing 
Erd\H{o}s-R\'{e}nyi graphs from random geometric graphs in the dense
regime. \cite{banerjee2016contiguity} and \cite{banerjee2017} study tests based on signed circles and establish the relation to likelihood ratio statistic for stochastic block models.

The remainder of the paper is organized as follows. In Section~\ref{sec:equations}
we introduce various subgraph
frequency equations and discuss how to use them to do network testing.
We state and explain our results that give a theoretical analysis of
the tests, for a range of alternative distributions, in
Section~\ref{sec:power}. Section~\ref{sec:sample}, we propose two network
sampling schemes to facilitate computation, and discuss tradeoffs
between computation and network sparsity. Section~\ref{sec:simulation} presents the results of numerical studies. Proofs
of all of the technical results are deferred to Section~\ref{sec:pf}.

We close this section by introducing the notation used in the
paper. For an integer $m$, we use $[m]$ to denote the set
$\{1,2,...,m\}$. Given a set $S$, the cardinality is denoted by $|S|$, and
$\mathbb{I}_S$ is the associated indicator function. When $A$ is a
square matrix, $\Tr(A)$ denotes its trace. For two matrices
$A,B\in\mathbb{R}^{d_1\times d_2}$, their trace inner product is
$\iprod{A}{B}=\Tr(AB^T)$. We use $\mathbb{P}$ and $\mathbb{E}$ to
denote generic probability and expectation, with respect to distributions
determined from the context.



\section{Relations Between Small Subgraph Frequencies}\label{sec:equations}

We now describe the basic relations between small subgraph
probabilities that we study.  There are four three-node subgraphs to consider: The empty graph (\shape{\emptyshape}), a
single edge (\shape{\edgeshape}), the V-shape
\hbox{(\shape{\Vshape})}, and the triangle (\shape{\triangleshape}).
Let $ A = (A_{ij})_{1\leq i<j\leq n}$ denote the adjacency matrix 
for a given undirected network on $n$ nodes.
We shall use subscripts that indicate the number of edges in the 
subgraph.  Thus, the relative frequency of triangles among 
all three-node subgraphs is 
$$\hat{F}_3=\frac{1}{{n\choose 3}}\sum_{1\leq i<j<k\leq
  n}A_{ij}A_{jk}A_{ki},$$
the $V$-shape frequency is 
$$\hat{F}_2=\frac{1}{{n\choose 3}}\sum_{1\leq i<j<k\leq n}\bigl\{(1-A_{ij})A_{jk}A_{ki}+A_{ij}(1-A_{jk})A_{ki}+A_{ij}A_{jk}(1-A_{ki})\bigr\},$$
the single edge frequency is 
$$\hat{F}_1=\frac{1}{{n\choose 3}}\sum_{1\leq i<j<k\leq n}\bigl\{(1-A_{ij})(1-A_{jk})A_{ki}+A_{ij}(1-A_{jk})(1-A_{ki})+(1-A_{ij})A_{jk}(1-A_{ki})\bigr\},$$
and the empty graph frequency is 
$$\hat{F}_0=\frac{1}{{n\choose 3}}\sum_{1\leq i<j<k\leq
  n}(1-A_{ij})(1-A_{jk})(1-A_{ki}).$$
Under an Erd\H{o}s-R\'{e}nyi model with edge probability $p$, the
population subgraph frequencies of the above four shapes are easily
seen to be
\begin{equation}
\mathbb{E}\hat{F}_3=p^3,\quad \mathbb{E}\hat{F}_2=3p^2(1-p),\quad\mathbb{E}\hat{F}_1=3p(1-p)^2,\quad\text{and}\quad \mathbb{E}\hat{F}_0=(1-p)^3.\label{eq:pop-freq}
\end{equation}
These expressions suggest that the subgraph frequencies can also be
estimated plugging in the empirical edge frequency
\begin{equation}
\hat{p}=\frac{1}{{n\choose 2}}\sum_{1\leq i<j\leq n}A_{ij}.\label{eq:p-hat}
\end{equation}
Under the Erd\H{o}s-R\'{e}nyi model, one can expect that the two ways
of estimating subgraph frequencies are close. This leads us to study the following four relations:
\begin{eqnarray*}
T_0 &=& (1-\hat{p})^3-\hat{F}_0,\\
T_1 &=& 3\hat{p}(1-\hat{p})^2-\hat{F}_1, \\
T_2 &=& 3\hat{p}^2(1-\hat{p})-\hat{F}_2, \\
T_3 &=& \hat{p}^3-\hat{F}_3.
\end{eqnarray*}
The quantities $T_0$, $T_1$, $T_2$ and $T_3$ are the differences
between the empirical subgraph frequencies and their plugin estimates.
It is easy to see that $T_0+T_1+T_2+T_3=0$; thus, there are at most
three degrees of freedom among the four equations. The following
proposition shows that for a sparse Erd\H{o}s-R\'{e}nyi graph, the
asymptotic number of degrees of freedom is actually only two.

\begin{proposition}\label{prop:corr}
Suppose that $p=o(1)$. Then
$$\Corr(T_2,T_3)=o(1),\quad \Corr(T_1,T_3)=o(1),\quad\text{and}\quad \Corr(T_1,T_2)=1+o(1),$$
where $o(1)$ indicates a quantity that converges to $0$ as $n\rightarrow\infty$.
\end{proposition}

The proposition shows that in order to understand the asymptotic
distribution of $T_0$, $T_1$, $T_2$, and $T_3$, it is sufficient to
study the asymptotic distribution of $(T_2,T_3)$.  However, this is
not a trivial task. Take $T_3$ for example. Although the leading term in
the expansion of $\hat{p}^3$ can be easily shown to be asymptotically
Gaussian using the delta method, and the leading term in the
expansion of $\hat{F}_3$ can also be shown to be asymptotically
Gaussian using the theory of incomplete U-statistics
\citep{nowicki1988subgraph,janson1991asymptotic}, the two leading
terms exactly cancel each other in $T_3$.  
It is therefore somewhat surprising that the joint
asymptotic distribution of $(T_2,T_3)$ is still Gaussian, as shown
in the next theorem.

\begin{theorem}\label{thm:vanilla}
Assume $(1-p)^{-1}=O(1)$ and $np\rightarrow\infty$. Then we have
$$\sqrt{{n\choose 3}}\Sigma_p^{-1/2}\begin{pmatrix}
T_2\\
T_3
\end{pmatrix}\leadsto N\left(\begin{pmatrix}
0 \\
0
\end{pmatrix},\begin{pmatrix}
1 & 0 \\
0 & 1
\end{pmatrix}\right),$$
where
\begin{equation} \label{eq:cov-def}
\Sigma_p=\begin{pmatrix}
3p^2(1-p)^2(1-3p)^2+9p^3(1-p)^3 & -6p^4(1-p)^2 \\
-6p^4(1-p)^2 & p^3(1-p)^3+3p^4(1-p)^2
\end{pmatrix}.
\end{equation}
\end{theorem}

The theorem gives a precise characterization of the asymptotic distribution of $(T_2,T_3)$. The condition $np\rightarrow\infty$ means that the network has a diverging degree. By the covariance structure (\ref{eq:cov-def}), it is easy to check that when $p=o(1)$, we have
$$\Corr(T_2,T_3)=\frac{\Cov(T_2,T_3)}{\sqrt{\Var(T_2)}\sqrt{\Var(T_3)}}=O(p^{3/2})=o(1),$$
which is consistent with Proposition \ref{prop:corr}.
The fact that $T_2$ and $T_3$ are asymptotically independent under
this scaling of $p$ suggests that one can build a chi-squared test to test whether a given network is generated by an Erd\H{o}s-R\'{e}nyi model. The proper normalization factor in the chi-squared test can be estimated using the empirical edge frequency (\ref{eq:p-hat}).

\begin{corollary}\label{cor:chi}
Assume $p=o(1)$ and $np\rightarrow\infty$. Then 
\begin{equation}
T^2=:{n\choose
  3}\left(\frac{T_2^2}{3\hat{p}^2(1-\hat{p})^2(1-3\hat{p})^2+9\hat{p}^3(1-\hat{p})^3}
+ \frac{T_3^2}{\hat{p}^3(1-\hat{p})^3+3\hat{p}^4(1-\hat{p})^2}\right)\leadsto \chi_2^2.\label{eq:chi^2}
\end{equation}
\end{corollary}

The natural $\alpha$-level test statistic is defined as
\begin{equation}
\phi_{\alpha}=\mathbb{I}\{T^2>C_{\alpha}\},\label{eq:test-alpha}
\end{equation}
where $C_{\alpha}$ is a constant satisfying
$\mathbb{P}(\chi_2^2>C_{\alpha})=\alpha$. Corollary \ref{cor:chi} then
implies that the Type-I error $\mathbb{P}_0\phi_{\alpha}$ is
asymptotically at the level $\alpha$ under an Erd\H{o}s-R\'{e}nyi
model.

\def\flac#1#2{{#1}/{#2}}
\def\given{\,|\,}

\section{Analysis of Statistical Power}\label{sec:power}

We now turn to an analysis of the power of the test $T^2$, adopting a fairly standard decision-theoretic framework for 
hypothesis
testing \citep{ingster2012nonparametric,baraud2002non}. Consider a
general simple versus composite hypothesis testing setting
$$H_0: \{A_{ij}\}_{1\leq i<j\leq n}\sim \mathbb{P}_0,\quad H_1: \{A_{ij}\}_{1\leq i<j\leq n}\sim \mathbb{P}_1\in\mathcal{P}(\delta),$$
where $\delta$ characterizes the deviation from the null model. Given a testing function $\phi$, its risk is defined as
$$\mathbb{P}_0\phi + \sup_{\mathbb{P}_1\in \mathcal{P}(\delta)}\mathbb{P}_1(1-\phi)$$
and the minimax risk is
$$\mathcal{R}(\delta)=\inf_{\phi}\left\{\mathbb{P}_0\phi +
\sup_{\mathbb{P}_1\in \mathcal{P}(\delta)}\mathbb{P}_1(1-\phi)\right\}.$$
Under an appropriate scaling condition 
$S(\delta)\rightarrow\infty$ on the deviation parameter $\delta$, a testing function $\phi$ is said to be consistent if
$$\mathbb{P}_0\phi + \sup_{\mathbb{P}_1\in
  \mathcal{P}(\delta)}\mathbb{P}_1(1-\phi)
\longrightarrow 0$$
as $S(\delta)\rightarrow\infty$.
It is said to be minimax optimal in case
\begin{equation*}
\limsup_{S(\delta)\to\infty} \frac{\mathbb{P}_0\phi +
  \sup_{\mathbb{P}_1\in
    \mathcal{P}(\delta)}\mathbb{P}_1(1-\phi)}{\mathcal{R}(\delta)} <
\infty.
\end{equation*}

The testing procedure we consider is the chi-squared statistic in
\eqref{eq:test-alpha}. Its Type-I error converges to $\alpha$, as
shown by Corollary \ref{cor:chi}. We slightly modify the test so that
it has a Type-I error that converges to zero, defining
\begin{equation}
\phi_n=\mathbb{I}\{T^2>C_n\},\label{eq:test-def}
\end{equation}
where $C_n$ is a sequence that diverges to infinity arbitrarily slowly.
\begin{theorem}\label{thm:type-1}
Assume an Erd\H{o}s-R\'{e}nyi model with $p$ such that $p=o(1)$ and
$n^2p\rightarrow 1$. Then 
$$\mathbb{P}_0\phi_n\rightarrow 0.$$
\end{theorem}

Note that consistency in terms of Type-I error is a weaker requirement
than the asymptotic distribution property in Corollary \ref{cor:chi}. 
We now study the power under various alternative models.

\subsection{Stochastic block model alternatives}

Stochastic block models serve as popular benchmarks in network
analysis problems such as community
detection \citep{bickel2009nonparametric,rohe2011spectral,mossel2012stochastic,lei2015consistency,abbe2016exact,zhang2016minimax,gao2015achieving}
and parameter
estimation \citep{bickel2011method,wolfe2013nonparametric,gao2015rate,borgs2015private,gao2016optimal,castillo2017}.
Hypothesis testing in the setting of stochastic block models have been
considered
by \cite{lei2016goodness,wang2015likelihood,bickel2016hypothesis}. However,
there is no procedure in the literature that can distinguish between a
stochastic block model and an Erd\H{o}s-R\'{e}nyi model under the
minimal signal-to-noise ratio requirement, with the only exception given by a forthcoming work by \cite{banerjee2017}. We tackle this problem 
using subgraph statistics.

We consider $A_{ij}\sim\text{Bernoulli}(p_{ij})$ independently for all
$1\leq i<j\leq n$ under the alternative distribution. As a special case,
the stochastic block model with two
communities is
$$\mathcal{S}(\pi,a,b)=\left\{\{p_{ij}\}_{1\leq i<j\leq n}: p_{ij}=a\mathbb{I}_{z(i)=z(j)}+b\mathbb{I}_{z(i)\neq z(j)}\text{ for some }a,b\in[0,1], z\in\mathcal{Z}_{n,2}(\pi)\right\},$$
where
$$\mathcal{Z}_{n,2}(\pi)=\left\{z\in\{1,2\}^n: n\pi\leq|\{i\in[n]:z(i)=1\}|\leq n(1-\pi)\right\}.$$
In other words,  the
parameters $a$ and $b$ are within-cluster and between-cluster
connectivity probabilities, and the number $\pi$ controls the size of the
two communities.

The following theorem shows that the power of the proposed
test \eqref{eq:test-def} converges to one under an appropriate
condition on the signal-to-noise ratio.

\begin{theorem}\label{thm:SBM}
Assume $n^2(a+b)\rightarrow\infty$, $a+b=O(n^{-1/2})$ and $\pi\in(0,1/2)$ is a constant that does not depend on $n$. Then, whenever
\begin{equation}
\frac{n(a-b)^2}{a+b}\longrightarrow\infty, \label{eq:SNR}
\end{equation}
it follows that
$$\inf_{\mathbb{P}_1\in \mathcal{S}(\pi,a,b)}\mathbb{P}_1\phi\longrightarrow 1.$$
\end{theorem}

Note that for a stochastic block model in $\mathcal{S}(\pi,a,b)$ with
a constant $\pi\in(0,1/2)$ that does not depend on $n$, it reduces to
an Erd\H{o}s-R\'{e}nyi model when $a-b=0$. Therefore, the quantity
$a-b \geq 0$ measures the deviation of a stochastic block model from an Erd\H{o}s-R\'{e}nyi model.
Theorem \ref{thm:SBM} says that for a sparse stochastic block model
that satisfies $a+b= O(1/\sqrt{n})$, the power of the proposed
test \eqref{eq:test-def} is asymptotically one as long as the
signal-to-noise ratio condition \eqref{eq:SNR} is satisfied.
Theorem \ref{thm:type-1} then implies that the proposed test \eqref{eq:test-def}
is consistent under the condition \eqref{eq:SNR}.
In fact, this scaling is also necessary for a consistent test.

\begin{theorem}\label{thm:SBM-lower}
Assume $a+b=o(1)$, $\pi\in(0,1/2)$ is a constant that does not depend on $n$ and
$$\frac{n(a-b)^2}{2(a+b)}<1-c,$$
for some arbitrarily small constant $c\in(0,1)$.
Then, for any testing procedure $\phi$ with Type-I error that
converges to zero 
under the Erd\H{o}s-R\'{e}nyi model with $p=(a+b)/2$, its power satisfies
$$\inf_{\mathbb{P}_1\in \mathcal{S}(\pi,a,b)}\mathbb{P}_1\phi\longrightarrow 0.$$
\end{theorem}

This lower bound is essentially a result in \cite{mossel2012stochastic}. We slightly modify their proof so that it is also applicable in our setting and holds for a wider range of model parameters. It is interesting to note that the simple test based on relations between subgraph frequencies can distinguish between an Erd\H{o}s-R\'{e}nyi model  and a stochastic block model under the optimal scaling condition. \cite {banerjee2016contiguity} and a forthcoming work by \cite{banerjee2017} justifies this approach by showing that the signed triangle frequency, which is asymptotically equivalent to $T_3$, is the leading term  in the likelihood ratio statistic

Note that Theorem~\ref{thm:SBM} requires a sparsity condition $a+b=O(\flac{1}{\sqrt{n}})$;
but this condition is weak. For a dense network such that $a+b = \Omega(\flac{1}{\sqrt{n}})$, the conclusion of Theorem \ref{thm:SBM} still holds if \eqref{eq:SNR} is replaced by the condition
$$
\frac{n^2(a-b)^6}{(a+b)^5}\longrightarrow\infty.
$$
We note that in this dense regime,
one can accurately estimate the community labels of a stochastic block
model, with high probability. In fact, efficient algorithms in the
literature for perfect community detection only require the sparsity
level $a+b = \Omega(\flac{\log
n}{n})$ \citep{mossel2014consistency,hajek2016achieving,abbe2016exact,gao2015achieving}. Therefore,
in this dense regime, the task essentially becomes a simple
vs.~simple hypothesis testing problem, because of the known (easily recoverable) community label, and a likelihood ratio test directly solves the problem.

For a stochastic block model with more than two communities, our
proposed test still exhibits good power when the sizes of the
communities are roughly equal. Define the space
$$\mathcal{S}(k,a,b)=\left\{\{p_{ij}\}_{1\leq i<j\leq n}: p_{ij}=a\mathbb{I}_{z(i)=z(j)}+b\mathbb{I}_{z(i)\neq z(j)}\text{ for some }a,b\in[0,1], z\in\mathcal{Z}_{n,k}\right\},$$
where
$$\mathcal{Z}_{n,k}=\left\{z\in[k]^n: \frac{n}{k}-1\leq|\{i\in[n]:z(i)=j\}|\leq \frac{n}{k}+1\text{ for all }j\in[k]\right\}.$$
Our next result shows that under an appropriate signal-to-noise ratio
condition, the proposed test \eqref{eq:test-def} again has asymptotic power one.

\begin{theorem}\label{thm:SBM-k}
Assume $n^2(a+b)\rightarrow\infty$ and $a+b=O(n^{-1/2})$. Then, whenever
\begin{equation}
\frac{n(a-b)^2}{k^{4/3}(a+b)}\longrightarrow\infty,\label{eq:SNR-k}
\end{equation}
we have
$$\inf_{\mathbb{P}_1\in \mathcal{S}(k,a,b)}\mathbb{P}_1\phi\longrightarrow 1.$$
\end{theorem}

Condition \eqref{eq:SNR-k} has an interesting dependence on the number
of communities $k$,
which reduces to \eqref{eq:SNR} when $k=2$. We remark that similar
conditions naturally arise in the context of community detection. For
example, in order to achieve weak consistency for community
detection, \cite{chin2015stochastic} require
$\flac{n(a-b)^2}{k^{2}(a+b)}\rightarrow\infty$
and \cite{gao2015achieving} require that
$\flac{n(a-b)^2}{k^{3}(a+b)}\rightarrow\infty$, among others in the
literature.  
Compared to these conditions, \eqref{eq:SNR-k} is the
weakest. However, 
testing and community detection are different network analysis
problems. It is thus an important problem, which we leave to future
investigation, whether or not the curious factor of $k^{4/3}$ in \eqref{eq:SNR-k} is optimal.

Unlike Theorem \ref{thm:SBM}, Theorem \ref{thm:SBM-k} requires the
communities to have equal size. However, our numerical study in
Section \ref{sec:simulation} shows that the proposed test still has
good power when the community sizes are different, suggesting that the condition might be weakened.

\subsection{Configuration models}

A configuration model is designed to model heterogeneity among the
node degrees; see, for example, \cite{van2009random}. It is a
special case of the degree-corrected stochastic block model
\citep{dasgupta2004spectral,karrer2011stochastic} with only one
community.  Under this model, every edge is independently sampled from
a Bernoulli distribution with mean $p_{ij}=\theta_i\theta_j$. A node with
a relatively large parameter $\theta$ tends to have more edges in the
network. Define the parameter space
$$\mathcal{C}(\rho,\delta)=\Bigl\{(\theta_i\theta_j)_{1\leq i<j\leq n}: \|\theta\|_{\infty}^2\leq\rho, V(\theta)\geq \delta\Bigr\},$$
where
$$V(\theta)=\left(\frac{1}{n}\sum_{i=1}^n\theta_i^2\right)^3-\left(\frac{1}{n}\sum_{i=1}^n\theta_i\right)^6.$$
The parameter $\rho$ indicates the sparsity of the network,
and the function $V(\cdot)$ characterizes the deviation of the
configuration model from Erd\H{o}s-R\'{e}nyi,
as shown by the following proposition.
\begin{proposition}\label{prop:delta-configuration}
For any $\theta\in\mathbb{R}^n$, $V(\theta)\geq 0$, and $V(\theta)=0$ if and only if $\theta$ is a constant vector.
\end{proposition}
Thus, $\delta$ is a parameter that can be interpreted as measuring
the deviation from Erd\H{o}s-R\'{e}nyi. The next theorem establishes a level of
$\delta$ that ensures asymptotically strong power.

\begin{theorem}\label{thm:configuration}
Assume $n^2\rho\rightarrow\infty$. Then, whenever
\begin{equation}
\frac{\delta^2}{\flac{\rho^3}{n^3}+\flac{\rho^5}{n^2}}\longrightarrow\infty,\label{eq:SNR-configuration}
\end{equation}
we have
$$\inf_{\mathbb{P}_1\in \mathcal{C}(\rho,\delta)}\mathbb{P}_1\phi\longrightarrow 1.$$
\end{theorem}

The signal-to-noise ratio condition \eqref{eq:SNR-configuration} has
two regimes. In the sparse regime where $\rho=O(1/\sqrt{n})$, the
scaling condition becomes $S_\rho(\delta) = \flac{n\delta^{2/3}}{\rho}\rightarrow\infty$. In
the dense regime where $\rho=\Omega(1/\sqrt{n})$, the condition becomes
$S_\rho(\delta) = \flac{n\delta}{\rho^{5/2}}\rightarrow\infty$. Since the denominator
$\flac{\rho^3}{n^3}+\flac{\rho^5}{n^2}$
in \eqref{eq:SNR-configuration} is the variance of the empirical
triangle frequency, the condition \eqref{eq:SNR-configuration} cannot
be improved with the proposed testing function \eqref{eq:test-def}.

\subsection{Low-rank latent variable models}

In a latent variable network model \citep{hoff2002latent}, each
network node is associated with an $r$-dimensional latent vector
$\xi_i=(\xi_{i1},...,\xi_{ir})$, and the adjacency matrix is sampled
from a graphon model
$A_{ij} \,|\,(\xi_i,\xi_j)\sim\text{Bernoulli}(f(\xi_i,\xi_j))$,
independently for all $1\leq i<j\leq n$. In this section, we consider
the case where the graphon $f$ admits a low-rank decomposition
$f(\xi_i,\xi_j)=\sum_{l=1}^rg_l(\xi_{il})g_l(\xi_{jl})$ 
in terms of $r$ feature functions
$g_1(\cdot),...,g_r(\cdot)$. Each network node is modeled with the $r$
features according to its latent vector
$\xi_i=(\xi_{i1},...,\xi_{ir})$, and connectivity between two
network nodes is determined by the inner product in this feature space.
For mathematical convenience, we let the latent variables
$\{\xi_{il}\}_{(i,l)\in[n]\times[r]}$ be specified as independent
$\text{Uniform}(0,1)$ random variables.
Define the graphon space
$$\mathcal{F}(\rho,\delta)=\left\{f(x,y)=\sum_{l=1}^rg_l(x_l)g_l(y_l): \sum_{l=1}^r\|g_l\|_{\infty}^2\leq\rho, V(g)\geq \delta\right\},$$
where
$$V(g)=\Tr\left[\left(\mathbb{E}g(\xi)g(\xi)^T\right)^3\right]-\|\mathbb{E}g(\xi)\|^6.$$
Here again, $\rho$ is the sparsity parameter, and $V(g)$ characterizes
the variability of $g$, according to the following proposition.
\begin{proposition}\label{prop:delta-latent}
For any $g$, $V(g)\geq 0$. Moreover, $V(g)=0$ if and only if $g$ is a constant function for each of the $r$ components.
\end{proposition}

Thus, $\delta$ is again a parameter that is interpreted as a deviation from Erd\H{o}s-R\'{e}nyi. 

\begin{theorem}\label{thm:latent}
Assume $n^2\rho\rightarrow\infty$. Then, whenever
\begin{equation}
\frac{\delta^2}{\flac{\rho^3}{n^3}+\flac{\rho^6}{n}}\longrightarrow\infty,\label{eq:SNR-f}
\end{equation}
we have
$$\inf_{\mathbb{P}_1\in \mathcal{F}(\rho,\delta)}\mathbb{P}_1\phi\longrightarrow 1.$$
\end{theorem}

Theorem \ref{thm:latent} shows that under the signal-to-noise ratio
condition \eqref{eq:SNR-f}, the asymptotic power of the proposed
test \eqref{eq:test-def} is one. Compared
with \eqref{eq:SNR-configuration}, the condition \eqref{eq:SNR-f}
involves an extra term $\flac{\rho^6}{n}$ in the denominator. This is
because in the graphon model
$A_{ij}\given (\xi_i,\xi_j)\sim\text{Bernoulli}(f(\xi_i,\xi_j))$, the
variance of the empirical triangle frequency is of order
$\flac{\rho^3}{n^3}+\flac{\rho^6}{n}$ instead of
$\flac{\rho^3}{n^3}+\flac{\rho^5}{n^2}$, as can be seen from the
variance decomposition 
$$\Var(\hat{F}_3)=\mathbb{E}\Var(\hat{F}_3\given\xi)+\Var(\mathbb{E}(\hat{F}_3\given\xi)).$$
The extra term $\flac{\rho^6}{n}$ is a result of the variability
$\Var(\mathbb{E}(\hat{F}_3\given\xi))$ of the latent vectors.
Under this latent variable model we again see two scaling regimes. The
sparse regime is now $\rho = O(\flac{1}{n^{2/3}})$ with the scaling is $S_\rho(\delta) =
n \delta^{2/3}/\rho \rightarrow\infty$, and the dense regime is $\rho
= \Omega(\flac{1}{n^{2/3}})$ with the scaling $S_\rho(\delta) =
n\delta^2 / \rho^6 \rightarrow \infty$.

\section{Subgraph Sampling Schemes}\label{sec:sample}

Computation of the testing procedure, though straightforward, requires summation over ${n\choose 3}$ triples. This may be computationally expensive for extremely large networks. This section proposes two sampling schemes that allow us to only go over a small subset of all the ${n\choose 3}$ triples while computing the testing statistic. Theoretical analysis suggests an interesting trade-off between computation and network sparsity.

\subsection{Sampling network nodes}

The first sampling method is conducted on the set of nodes. Let $\mathcal{M}\subset[n]$ be a set of uniform sampling without replacement of size $|\mathcal{M}|=m$. For each node $i\in\mathcal{M}$, we go over all pairs in $\{1\leq j<k\leq n:j,k\neq i\}$. For the triple $(i,j,k)$, it contributes to counts of edges, triangles and V-shapes by
$$\frac{A_{ij}+A_{ik}+A_{jk}}{3},\quad A_{ij}A_{ik}A_{jk},$$
and
$$A_{ij}A_{ik}(1-A_{jk})+A_{ij}A_{jk}(1-A_{ik})+A_{ik}A_{jk}(1-A_{ij}).$$
Taking averages, we obtain
\begin{eqnarray}
\nonumber \hat{p}^{\mathcal{M}} &=& \frac{1}{m{n-1\choose 2}}\sum_{i\in\mathcal{M}}\sum_{\{1\leq j<k\leq n:j,k\neq i\}}\frac{A_{ij}+A_{ik}+A_{jk}}{3}, \\
\nonumber \hat{F}_3^{\mathcal{M}} &=& \frac{1}{m{n-1\choose 2}}\sum_{i\in\mathcal{M}}\sum_{\{1\leq j<k\leq n:j,k\neq i\}}A_{ij}A_{ik}A_{jk}, \\
\nonumber \hat{F}_2^{\mathcal{M}} &=& \frac{1}{m{n-1\choose 2}}\sum_{i\in\mathcal{M}}\sum_{\{1\leq j<k\leq n:j,k\neq i\}}\Big[A_{ij}A_{ik}(1-A_{jk})+A_{ij}A_{jk}(1-A_{ik})+A_{ik}A_{jk}(1-A_{ij})\Big].
\end{eqnarray}
Rather than a summation over ${n\choose 3}$ triples, the above
quantities are computed by a summation over $m{n-1\choose 2}$ triples. Moreover, computations of $\hat{F}_2^{\mathcal{M}}$ and $\hat{F}_3^{\mathcal{M}}$ can be done by just looking at the local neighboring graph of each node in $\mathcal{M}$.
This makes the node sampling scheme very useful for extremely large networks such as Facebook data. For example, Given a person $i$ in the Facebook network, the empirical frequencies of triangles and V-shapes among the sub-network of all her friends including herself are
$$\frac{1}{{n-1\choose 2}}\sum_{\{1\leq j<k\leq n:j,k\neq i\}}A_{ij}A_{ik}A_{jk},$$
and
$$\frac{1}{{n-1\choose 2}}\sum_{\{1\leq j<k\leq n:j,k\neq i\}}\Big[A_{ij}A_{ik}(1-A_{jk})+A_{ij}A_{jk}(1-A_{ik})+A_{ik}A_{jk}(1-A_{ij})\Big].$$
Then, $\hat{F}_2^{\mathcal{M}}$ and $\hat{F}_3^{\mathcal{M}}$ are just averages of the above quantities over $i\in\mathcal{M}$. Similar sampling schemes of computing other subgraph frequencies have been considered in \cite{ugander2013subgraph}. Here we provide a theoretical justification of this heuristic approach.

With the help of $\hat{p}^{\mathcal{M}}$, $\hat{F}_2^{\mathcal{M}}$ and $\hat{F}_3^{\mathcal{M}}$, we define
\begin{eqnarray*}
T_2^{\mathcal{M}} &=& 3(\hat{p}^{\mathcal{M}})^2(1-\hat{p}^{\mathcal{M}})-\hat{F}_2^{\mathcal{M}},\\
T_3^{\mathcal{M}} &=& (\hat{p}^{\mathcal{M}})^3-\hat{F}_3^{\mathcal{M}}.
\end{eqnarray*}
The next theorem derives the joint asymptotic distribution of $(T_2^{\mathcal{M}}, T_3^{\mathcal{M}})$. We require that $m=m(n)$ is a nondecreasing function of $n$ that satisfies $m(1)=1$ and $m(l)\leq l$ for all $l\in[n]$.

\begin{theorem}\label{thm:sample-node}
Assume $(1-p)^{-1}=O(1)$, $p^3mn^2\rightarrow\infty$ and $m\rightarrow\infty$, and then we have,
$$\sqrt{f(n,m)}\Sigma_p^{-1/2}\begin{pmatrix}
T_2^{\mathcal{M}}\\
T_3^{\mathcal{M}}
\end{pmatrix}\leadsto N\left(\begin{pmatrix}
0 \\
0
\end{pmatrix},\begin{pmatrix}
1 & 0 \\
0 & 1
\end{pmatrix}\right),$$
where
$$f(n,m)=\frac{m^2(n-1)^2(n-2)^2}{36{m\choose 3}+16{m\choose 2}(n-m)+4m{n-m\choose 2}},$$
and $\Sigma$ is defined in (\ref{eq:cov-def}).
\end{theorem}

As a consequence, we obtain the following asymptotic distribution of the chi-squared statistic.

\begin{corollary}\label{cor:chi-node}
Assume $p=o(1)$ and $p^3mn^2\rightarrow\infty$ and $m\rightarrow\infty$, and then we have
\begin{equation}
(T^{\mathcal{M}})^2\leadsto \chi_2^2,\label{eq:chi^2-node}
\end{equation}
where $T^{\mathcal{M}}$ is defined in the same way as $T$ in (\ref{eq:chi^2}), except that $T_2,T_3,\hat{p}, {n\choose 3}$ are replaced by $T_2^{\mathcal{M}},T_3^{\mathcal{M}},\hat{p}^{\mathcal{M}}, f(n,m)$.
\end{corollary}

With (\ref{cor:chi-node}), one can use $(T^{\mathcal{M}})^2$ to obtain an $\alpha$-level test and a $p$-value. Next, we study the theoretical Type-I and Type-II errors of a procedure based on this testing statistic.
Consider the testing function
\begin{equation}
\phi^{\mathcal{M}}=\mathbb{I}\left\{(T^{\mathcal{M}})^2>C_n\right\},\label{eq:test-M}
\end{equation}
where $C_n$ is a sequence that varies to infinity arbitrarily slowly.

\begin{theorem}\label{thm:error-M}
Under an Erd\H{o}s-R\'{e}nyi model with edge probability $p$ that satisfies $p=o(1)$ and $p^3mn^2\rightarrow\infty$, then $\mathbb{P}\phi^{\mathcal{M}}\rightarrow 0$. For the power under alternatives, we have the following situations:
\begin{enumerate}
\item
Under the $2$-community stochastic block model $\mathcal{S}(\pi,a,b)$, where $n(a+b)\rightarrow\infty$, $\pi\in(0,1/2)$ is a constant that does not depend on $n$. Then, when $a+b=O(n^{-2/3})$ and
\begin{equation}
\frac{(mn^2)^{1/3}(a-b)^2}{a+b}\rightarrow\infty,\label{eq:SNR-M}
\end{equation}
we have $\mathbb{P}_1\in \inf_{\mathcal{S}(\pi,a,b)}\mathbb{P}_1\phi^{\mathcal{M}}\rightarrow 1$.
The conclusion still holds when $a+b=\Omega(n^{-2/3})$ if (\ref{eq:SNR-M}) is replaced by $\frac{m(a-b)^6}{(a+b)^6}\rightarrow\infty$.
\item Under the balanced $k$-community stochastic block model $\mathcal{S}(k,a,b)$, where $n(a+b)\rightarrow\infty$. Then, when $a+b=O(n^{-2/3})$ and
\begin{equation}
\frac{(mn^2)^{1/3}(a-b)^2}{k^{4/3}(a+b)}\rightarrow\infty,\label{eq:SNR-M-k}
\end{equation}
we have $\inf_{\mathbb{P}_1\in \mathcal{S}(k,a,b)}\mathbb{P}_1\phi^{\mathcal{M}}\rightarrow 1$.
The conclusion still holds when $a+b=\Omega(n^{-2/3})$ if (\ref{eq:SNR-M-k}) is replaced by $\frac{m(a-b)^6}{k^4(a+b)^6}\rightarrow\infty$.
\item Under the configuration model $\mathcal{C}(\rho,\delta)$, where
  $n\rho\rightarrow\infty$ and
\begin{equation}
\frac{\delta^2}{\flac{\rho^3}{mn^2}+\flac{\rho^6}{m}}\rightarrow\infty,\label{eq:SNR-M-C}
\end{equation}
we have $\inf_{\mathbb{P}_1\in \mathcal{C}(\rho,\delta)}\mathbb{P}_1\phi^{\mathcal{M}}\rightarrow 1$.
\item Under the low-rank latent variable model
  $\mathcal{F}(\rho,\delta)$, where $n\rho\rightarrow\infty$, when
\begin{equation}
\frac{\delta^2}{\flac{\rho^3}{mn^2}+\flac{\rho^6}{m}}\rightarrow\infty,\label{eq:SNR-M-L}
\end{equation}
we have $\inf_{\mathbb{P}_1\in \mathcal{F}(\rho,\delta)}\mathbb{P}_1\phi^{\mathcal{M}}\rightarrow 1$.
\end{enumerate}
\end{theorem}

The theorem gives performance guarantees of the testing function (\ref{eq:test-M}) constructed by sampling nodes. It covers both the null model and a list of alternative distributions. When the network is not too dense, namely, $a+b=O(n^{-2/3})$ or $\rho=O(n^{-2/3})$, the conditions (\ref{eq:SNR-M})-(\ref{eq:SNR-M-L}) correspond to (\ref{eq:SNR}), (\ref{eq:SNR-k}), (\ref{eq:SNR-configuration}) and (\ref{eq:SNR-f}). In other words, the role of $n$ is replaced by $(mn^2)^{1/3}$ in the node sampling scheme. This is because the testing statistic is constructed using $m{n-1\choose 2}\asymp mn^2$ triples instead of the original ${n\choose 3}\asymp n^3$ ones. As a consequence, we require more stringent signal-to-noise ratio conditions in Theorem \ref{thm:error-M}. Thus, one can use a smaller $m$ if a network is denser and contains more signal. Take the condition (\ref{eq:SNR-M}) for instance, the minimal order of $m$ that is required is at least  of order $\frac{(a+b)^3}{n^2(a-b)^6}$.

\subsection{Sampling network triples}

Another sampling approach is to directly sample from all ${n\choose 3}$ triples. We consider uniform sampling with replacement. Then, sampling a triple $(i,j,k)$ is equivalent to sampling the three indexes uniformly from $[n]$ without replacement, which can be done in a straightforward way.
Let $\Delta\subset\{(i,j,k):1\leq i<j<k\leq n\}$ be the set of triples that is sampled uniformly with replacement. Note that $\Delta$ is a multiset that possibly contains repeated elements. Define
\begin{eqnarray*}
\hat{p}^{\Delta} &=& \frac{1}{|\Delta|}\sum_{(i,j,k)\in\Delta}\frac{A_{ij}+A_{jk}+A_{ik}}{3}, \\
\hat{F}_2^{\Delta} &=& \frac{1}{|\Delta|}\sum_{(i,j,k)\in\Delta}\Big[A_{ij}A_{ik}(1-A_{jk})+A_{ij}A_{jk}(1-A_{ik})+A_{ik}A_{jk}(1-A_{ij})\Big]\\
\hat{F}_3^{\Delta} &=& \frac{1}{|\Delta|}\sum_{(i,j,k)\in\Delta}A_{ij}A_{ik}A_{jk}.
\end{eqnarray*}
Then, the two subgraph equations we need are
\begin{eqnarray*}
T_2^{\Delta} &=& 3(\hat{p}^{\Delta})^2(1-\hat{p}^{\Delta})-\hat{F}_2^{\Delta},\\
T_3^{\Delta} &=& (\hat{p}^{\Delta})^3 - \hat{F}_3^{\Delta}.
\end{eqnarray*}
The joint asymptotic distribution of $(T_2^{\Delta},T_3^{\Delta})$ is derived in the following theorem.

\begin{theorem}\label{thm:sample-triple}
Assume $(1-p)^{-1}=O(1)$, $|\Delta|p^3\rightarrow\infty$ and $np\rightarrow\infty$, and then we have
$$\sqrt{\frac{{n\choose 3}|\Delta|}{{n\choose 3}+|\Delta|}}\Sigma_p^{-1/2}\begin{pmatrix}
T_2^{\Delta}\\
T_3^{\Delta}
\end{pmatrix}\leadsto N\left(\begin{pmatrix}
0 \\
0
\end{pmatrix},\begin{pmatrix}
1 & 0 \\
0 & 1
\end{pmatrix}\right),$$
where $\Sigma$ is defined in (\ref{eq:cov-def}).
\end{theorem}

As a consequence, we obtain the following asymptotic distribution of the chi-squared statistic.

\begin{corollary}\label{cor:chi-triple}
Assume $p=o(1)$, $|\Delta|p^3\rightarrow\infty$ and $np\rightarrow\infty$, and then we have
\begin{equation}
(T^{\Delta})^2\leadsto \chi_2^2,\label{eq:chi^2-triple}
\end{equation}
where $T^{\Delta}$ is defined in the same way as $T$ in (\ref{eq:chi^2}), except that $T_2,T_3,\hat{p}, {n\choose 3}$ are replaced by $T_2^{\Delta},T_3^{\Delta},\hat{p}^{\Delta}, \frac{{n\choose 3}|\Delta|}{{n\choose 3}+|\Delta|}$.
\end{corollary}

The above asymptotic distribution results are analogous to Theorem \ref{thm:vanilla} and Corollary \ref{cor:chi}. The term ${n\choose 3}$ in Theorem \ref{thm:vanilla} and Corollary \ref{cor:chi} is replaced by $\frac{{n\choose 3}|\Delta|}{{n\choose 3}+|\Delta|}$ in Theorem \ref{thm:sample-triple} and Corollary \ref{cor:chi-triple}. Note that $\frac{{n\choose 3}|\Delta|}{{n\choose 3}+|\Delta|}$ and ${n\choose 3}$ are of the same order whenever $|\Delta|=\Omega(n^3)$. Therefore, in the following theorem that studies the testing errors, we only consider the regime $|\Delta|=o(n^3)$.

\begin{theorem}\label{thm:error-Delta}
Without loss of generality, we consider the situation $|\Delta|=o(n^3)$.
Under an Erd\H{o}s-R\'{e}nyi model with edge probability $p$ that satisfies $p=o(1)$ and $|\Delta|p^3\rightarrow\infty$, then
$\mathbb{P}\phi^{\Delta}\rightarrow 0$.
For the alternatives, we have the following situations:
\begin{enumerate}
\item
Under the $2$-community stochastic block model $\mathcal{S}(\pi,a,b)$, where $n^2(a+b)\rightarrow\infty$, $\pi\in(0,1/2)$ is a constant that does not depend on $n$. Then, when $a+b=O(n/\sqrt{|\Delta|})$ and
\begin{equation}
\frac{|\Delta|^{1/3}(a-b)^2}{a+b}\rightarrow\infty,\label{eq:SNR-Delta}
\end{equation}
we have $\inf_{\mathbb{P}_1\in \mathcal{S}(\pi,a,b)}\mathbb{P}_1\phi^{\Delta}\rightarrow 1$.
The conclusion still holds when $a+b=\Omega(n/\sqrt{|\Delta|})$ if (\ref{eq:SNR-Delta}) is replaced by $\frac{n^2(a-b)^6}{(a+b)^5}\rightarrow\infty$.
\item Under the balanced $k$-community stochastic block model $\mathcal{S}(k,a,b)$, where $n^2(a+b)\rightarrow\infty$. Then, when $a+b=O(n/\sqrt{|\Delta|})$ and
\begin{equation}
\frac{|\Delta|^{1/3}(a-b)^2}{k^{4/3}(a+b)}\rightarrow\infty,\label{eq:SNR-Delta-k}
\end{equation}
we have $\inf_{\mathbb{P}_1\in \mathcal{S}(k,a,b)}\mathbb{P}_1\phi^{\Delta}\rightarrow 1$.
The conclusion still holds when $a+b=\Omega(n/\sqrt{|\Delta|})$ if (\ref{eq:SNR-Delta-k}) is replaced by $\frac{n^2(a-b)^6}{k^4(a+b)^5}\rightarrow\infty$.
\item Under the configuration model $\mathcal{C}(\rho,\delta)$, where
  $n^2\rho\rightarrow\infty$, when
\begin{equation}
\frac{\delta^2}{\flac{\rho^3}{|\Delta|}+\flac{\rho^6}{n}}\rightarrow\infty,\label{eq:SNR-Delta-C}
\end{equation}
we have $\inf_{\mathbb{P}_1\in \mathcal{C}(\rho,\delta)}\mathbb{P}_1\phi^{\Delta}\rightarrow 1$.
\item Under the low-rank latent variable model
  $\mathcal{F}(\rho,\delta)$, where $n^2\rho\rightarrow\infty$, when
\begin{equation}
\frac{\delta^2}{\flac{\rho^3}{|\Delta|}+\flac{\rho^6}{n}}\rightarrow\infty,\label{eq:SNR-Delta-L}
\end{equation}
we have $\inf_{\mathbb{P}_1\in \mathcal{F}(\rho,\delta)}\mathbb{P}_1\phi^{\Delta}\rightarrow 1$.
\end{enumerate}
\end{theorem}
Note that if $|\Delta|$ is replaced by $n^3$, then Theorem \ref{thm:error-Delta} will recover the results in Section \ref{sec:power}. Stronger signal-to-noise ratio conditions (\ref{eq:SNR-Delta})-(\ref{eq:SNR-Delta-L}) are required for the efficiency of computation. This demonstrates an interesting trade-off between computational budget and network sparsity. For a network with denser edges and stronger signals, there is no need to use the full network to achieve powerful testing results. Consider the situation where $a\asymp b\asymp a-b$ in the 2-community stochastic block model. Then, the condition (\ref{eq:SNR-Delta}) becomes $|\Delta|^{1/3}a\rightarrow\infty$. In practice, $\hat{p}^{-3}$ will be a good suggestion of the order of $|\Delta|$ that is required.


\section{Numerical Studies}\label{sec:simulation}

This section is devoted to numerical studies of the proposed testing procedures in various settings. As a first task, we check whether the theoretical asymptotic distributions derived in Theorem \ref{thm:vanilla} and Corollary \ref{cor:chi} hold in simulations. We generate networks from Erd\H{o}s-R\'{e}nyi distributions. Four combinations of the network size $n$ and the edge probability $p$ are considered. For each scenario, we compute three statistics,
$$\frac{\sqrt{{n\choose 3}}T_3}{\sqrt{p^3(1-p)^3+3p^4(1-p)^2}},\quad \frac{\sqrt{{n\choose 3}}T_2}{\sqrt{3p^2(1-p)^2(1-3p)^2+9p^3(1-p)^3}},$$
and $T^2$ defined in (\ref{eq:chi^2}). Histograms of the three statistics are computed with $1000$ independent experiments. The results are plotted in Figure \ref{fig:1}. Each setting of $(n,p)$ occupies a column. The histograms of normalized $T_3$, normalized $T_2$ and $T^2$ are plotted in the three rows. Each histogram is superimposed with a theoretical density curve in red.
\begin{figure}[tbp]
\centering
\includegraphics[width=6in]{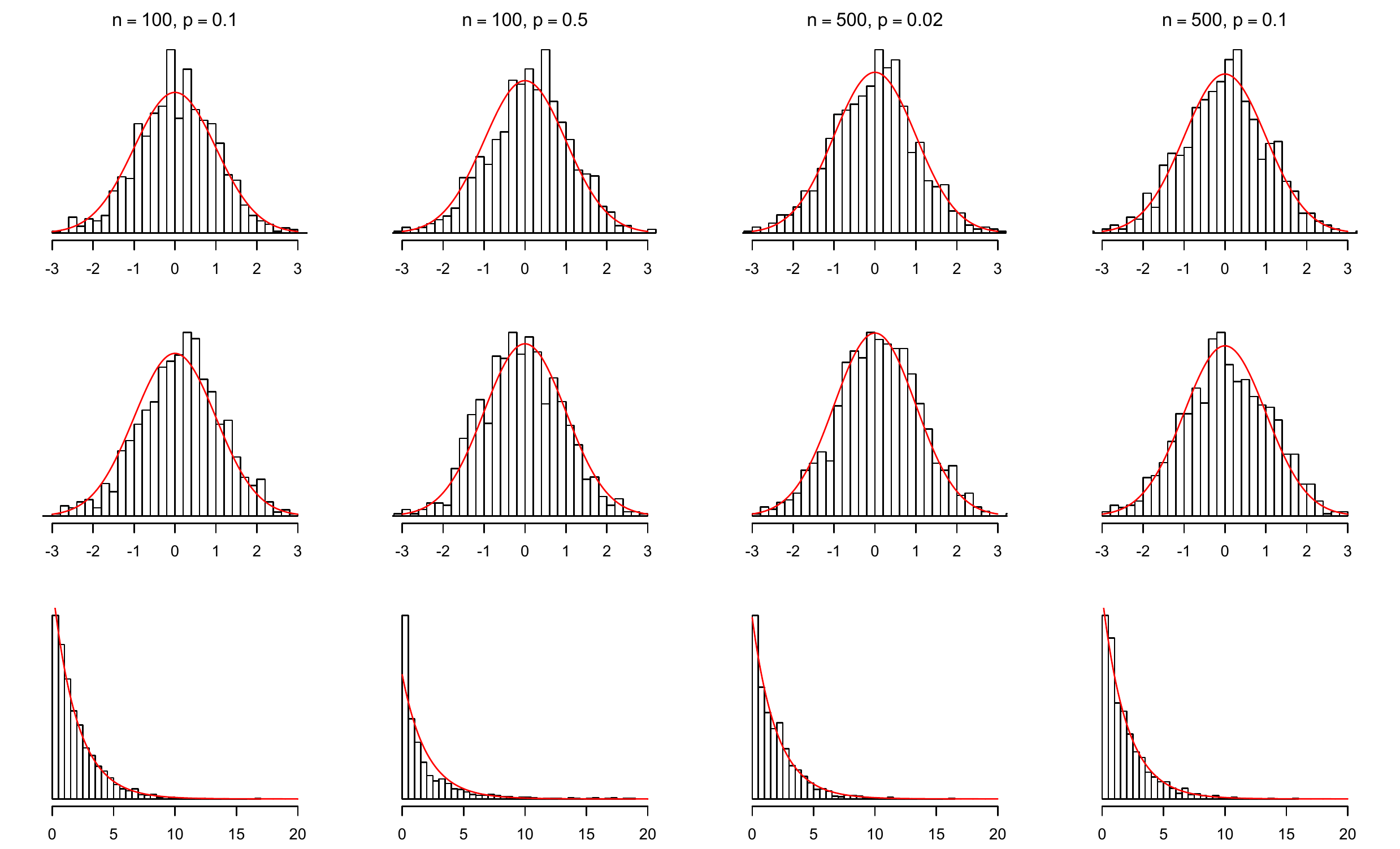}
\caption{Histograms of normalized $T_3$, normalized $T_2$ and $T^2$ for Erd\H{o}s-R\'{e}nyi graphs.}
\label{fig:1}
\end{figure}
The results show excellent matches between the theoretical asymptotic distributions and the empirical ones. The only situation that is slightly off is $T^2$ when $n=100$ and $p=0.5$ (3rd row and 2nd column in Figure \ref{fig:1}). This is because Corollary \ref{cor:chi} requires $p=o(1)$ to ensure asymptotic independence between $T_3$ and $T_2$. In this case, $p=0.5$ implies a dense network, and thus violates the assumption of Corollary \ref{cor:chi}.

Next, we study the power of the proposed chi-squared test under the 2-community stochastic block models.
The testing procedure is $\phi_{\alpha}$ defined in (\ref{eq:test-alpha}). We set $\alpha=0.05$ throughout this section.
The experiments consider SBM with size $n=100$. We study the power by varying both the signal-to-noise ratio $\frac{n(a-b)^2}{a+b}$ and the proportion of the smaller community $\gamma\in(0,1/2]$. In other words, we generate a stochastic block model with the sizes of the two communities roughly $n\gamma$ and $n(1-\gamma)$, and then connect edges between nodes according to their labels with parameters $a$ or $b$. Note that the number $\gamma$ is different from the $\pi$ in Theorem \ref{thm:SBM}, where the latter stands for the lower bound of $\gamma$.
\begin{figure}[tbp]
\centering
\includegraphics[width=5in]{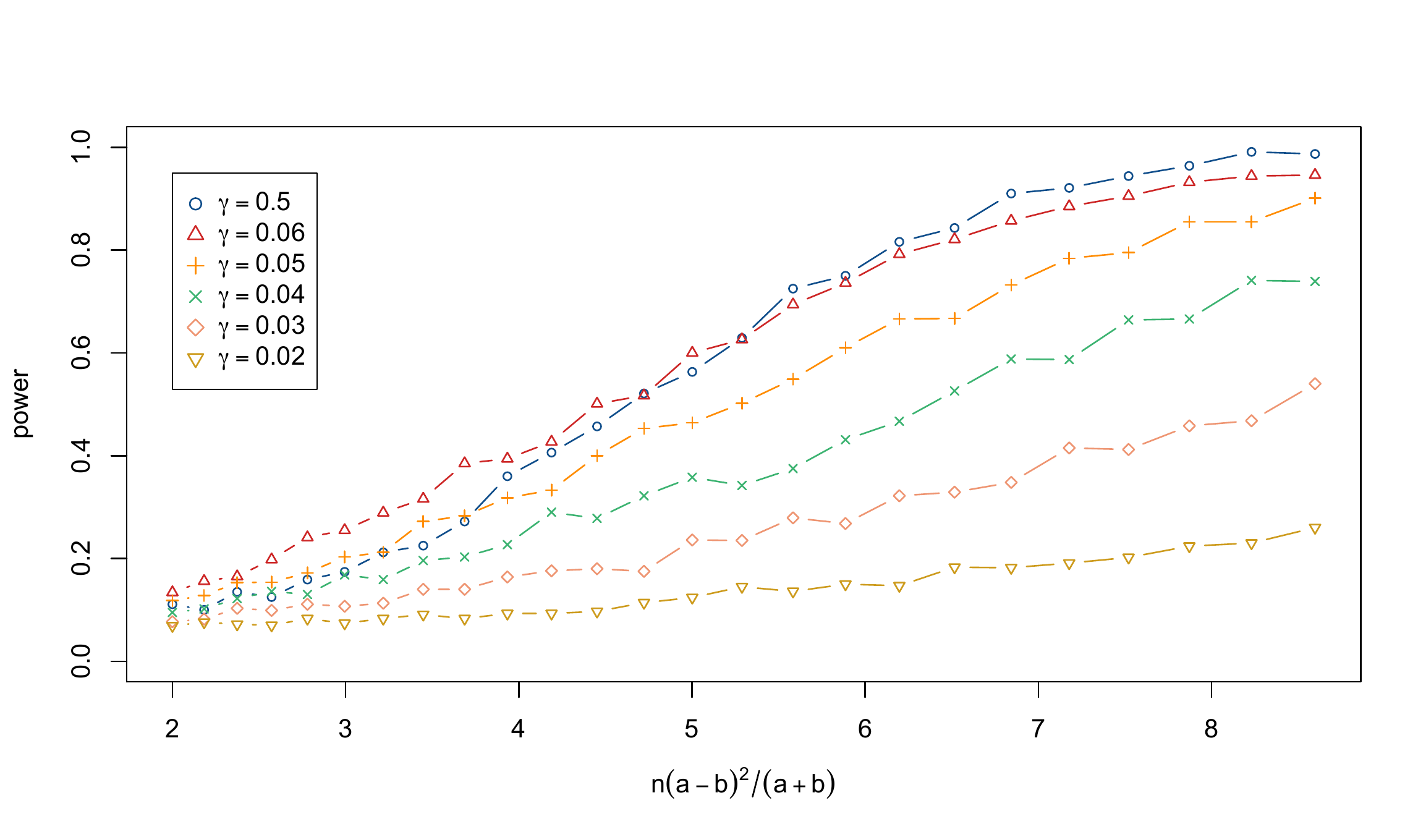}
\caption{Power of the proposed $0.05$-level test under the 2-community stochastic block models.}
\label{fig:2}
\end{figure}
The experiment results are shown in Figure \ref{fig:2}. Each curve is a power function against the signal-to-noise ratio $\frac{n(a-b)^2}{a+b}$ for a given $\gamma$. Each point on a curve is computed by averaging the results of $1000$ independent experiments. Starting from the left, we set $a=0.2+\sqrt{0.002}$ and $b=0.2\sqrt{0.002}$ so that $\frac{n(a-b)^2}{a+b}=2$, which corresponds to the theoretical information limit in Theorem \ref{thm:SBM-lower}. Then, we range $(a,b)$ in the set $\{(0.2+\sqrt{0.002}i,0.2-\sqrt{0.002}i):i\in\{1,2,...,25\}\}$. Figure \ref{fig:2} shows a power behavior that is well predicted by Theorem \ref{thm:SBM} and Theorem \ref{thm:SBM-lower}. As the number $\frac{n(a-b)^2}{a+b}$ becomes larger, the power increases to $1$. A surprising fact we learn from this simulation is the behavior of power against $\gamma$. As $\gamma$ decreases from $0.5$ to $0$, the $2$-community stochastic block model converges to an Erd\H{o}s-R\'{e}nyi model. Therefore, we expect a loss of power as $\gamma$ decreases. However, according to our experiments, the power barely decreases before $\gamma$ drops below $0.06$. This suggests that our test maintains the signal of the model for a very wide range of $\gamma$. After $\gamma$ drops below $0.06$, the power curves gradually become flat as $\gamma$ continues to decreases.

We then study the power of the proposed test under the $k$-community stochastic block models. The goal is to reveal the dependence on $k$ of the power curve.
\begin{figure}[tbp]
\centering
\includegraphics[width=5in]{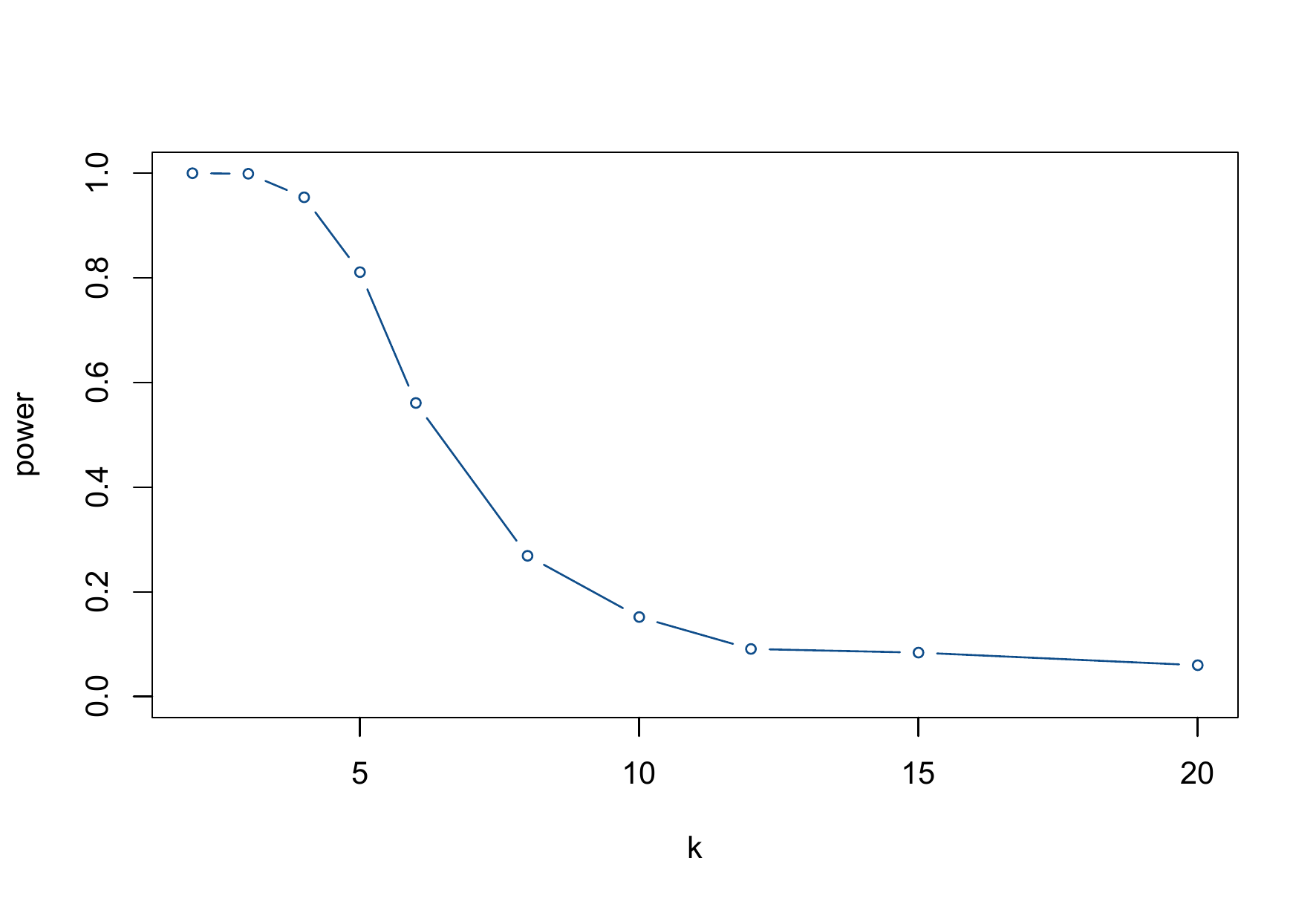}
\caption{Power of the proposed $0.05$-level test under the balanced $k$-community stochastic block models.}
\label{fig:3}
\end{figure}
Figure \ref{fig:3} shows an experiment by fixing $n=120$, $a=0.3$ and $b=0.1$. The number of communities $k$ varies in the set $\{2,3,4,5,6,8,10,12,15,20\}$. Again, the power at each point is calculated by averaging over $1000$ independent experiments. According to Figure \ref{fig:3}, the power is nearly $1$ when $k=2$ and $k=3$, and drops to a reasonable $0.8$ when $k=5$. After that, it decreases quickly to $0$. This dependence on $k$ is well predicated by Theorem \ref{thm:SBM-k}.

Besides the balanced $k$-community stochastic block models, it would be interesting to study an unbalanced SBM with $k>2$. Though this situation is not theoretically studied in the paper, numerical experiments may shed some light on the behavior of the power. We consider stochastic block models with $n=120$ and $k=3$. Among the three communities, the size of the third community is fixed as $40$. The number $\gamma$ indicates the proportion of the first community in the remaining $80$ nodes. In other words, the first community roughly has $80\gamma$ nodes and the second community roughly has $80(1-\gamma)$ nodes.
\begin{figure}[tbp]
\centering
\includegraphics[width=5in]{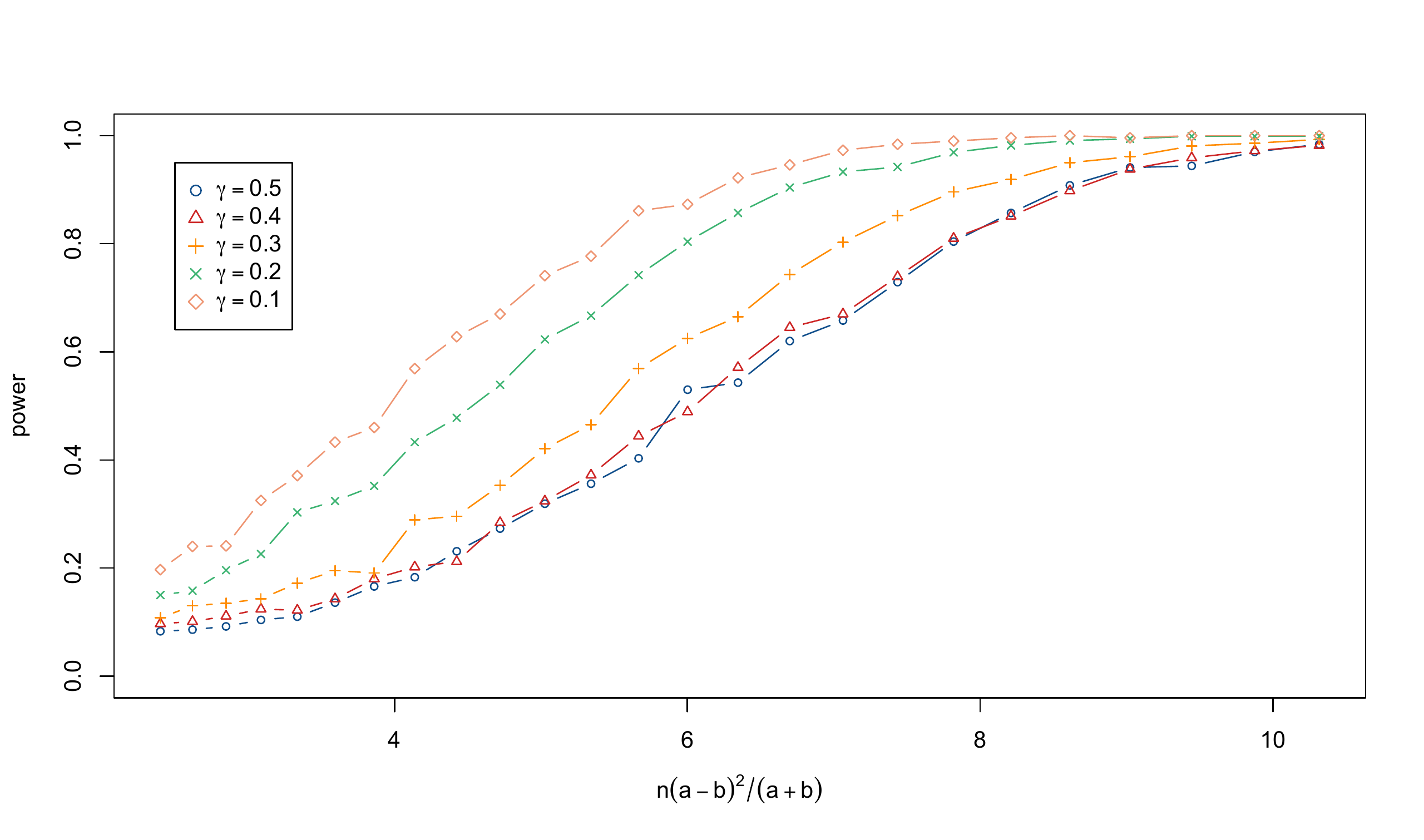}
\caption{Power of the proposed $0.05$-level test under the unbalanced 3-community stochastic block models.}
\label{fig:4}
\end{figure}
The experiment results are shown in Figure \ref{fig:4}. Each curve is a power function against the signal-to-noise ratio $\frac{n(a-b)^2}{a+b}$ for a given $\gamma$. Each point on a curve is computed by averaging the results of $1000$ independent experiments. The numbers $(a,b)$ range in the set $\{(0.2+\sqrt{0.002}i,0.2-\sqrt{0.002}i):i\in\{1,2,...,25\}\}$. The five curves correspond to five values of $\gamma$ in $\{0.5,0.4,0.3,0.2,0.1\}$. It is interesting to note that even when $\gamma\neq 0.5$, the test still has power. Moreover, the power actually increases as we decrease $\gamma$. It suggests that the balanced community size assumption in Theorem \ref{thm:SBM-k} can potentially be relaxed.

Now we study the power of our test under the configuration model. Consider a configuration model with size $n$ and each entry $\theta$ sampled from a beta distribution. Note that for $\theta_i\sim Beta(\alpha,\beta)$, the mean and variance are given by the formulas
$$\mathbb{E}\theta_i=\frac{\alpha}{\alpha+\beta},\quad\text{and}\quad \Var(\theta_i)=\frac{\alpha\beta}{(\alpha+\beta)^2(\alpha+\beta+1)}.$$
We consider the parametrization $\beta=\frac{1-h}{h}\alpha$. This leads to
$$\mathbb{E}\theta_i=h,\quad\text{and}\quad \Var(\theta_i)=\frac{h^2(1-h)}{\alpha+h}=\frac{h^2(1-h)}{\alpha}\left(1+O(h/\alpha)\right).$$
Therefore, $h$ can be roughly understood as the sparsity parameter of the network, and $\alpha$ can roughly be interpreted as the deviation from Erd\H{o}s-R\'{e}nyi models. Note that as $\alpha\rightarrow\infty$, the configuration model converges to an Erd\H{o}s-R\'{e}nyi model.
\begin{figure}[tbp]
\centering
\includegraphics[width=5in]{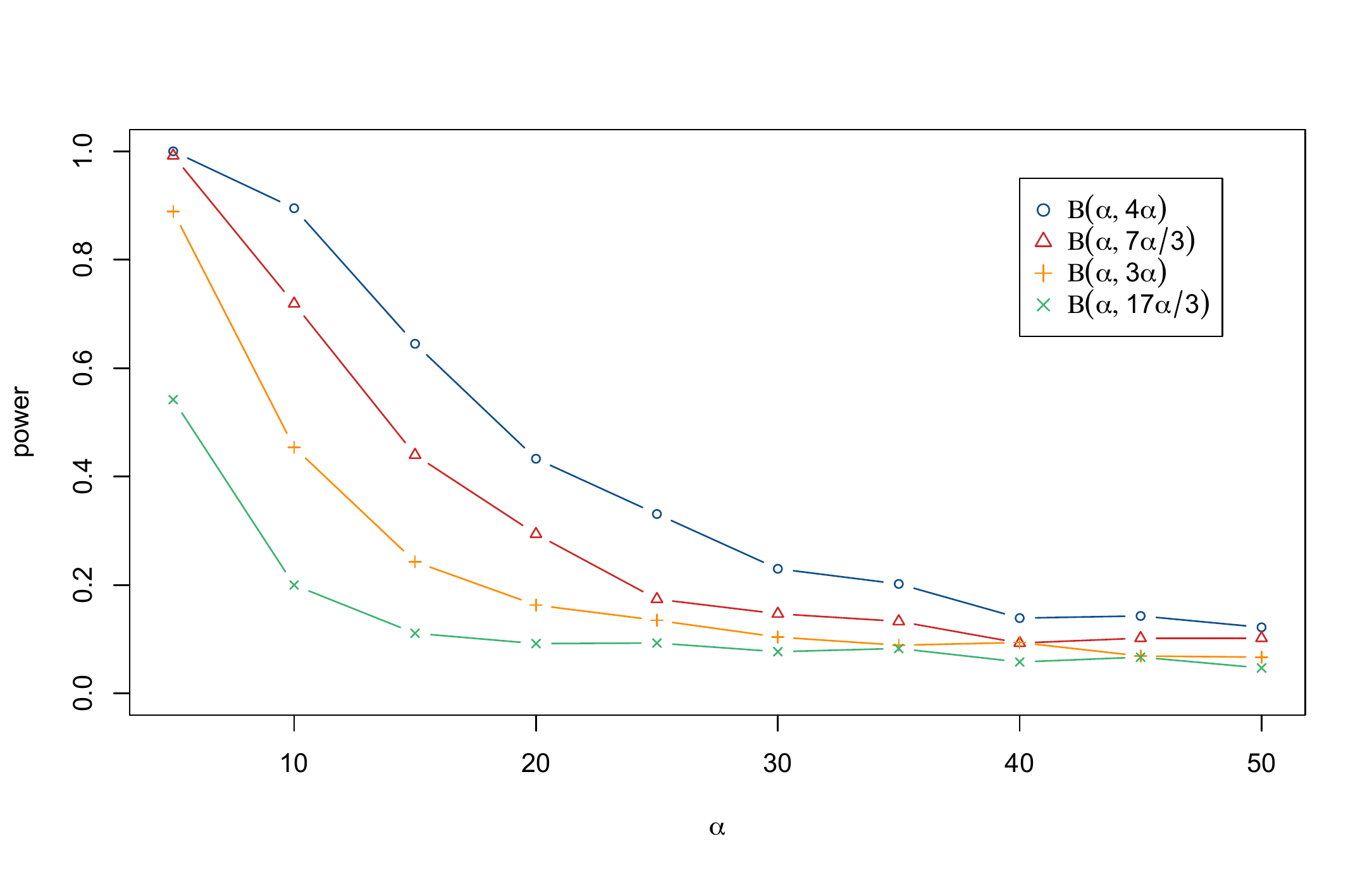}
\caption{Power of the proposed $0.05$-level test under the configuration models.}
\label{fig:5}
\end{figure}
Figure \ref{fig:5} shows four power curves with the values of $h$ in $\{0.3,0.25,0.2,0.15\}$, which corresponds to the distributions $\text{Beta}(\alpha,4\alpha)$, $\text{Beta}(\alpha,7\alpha/3)$, $\text{Beta}(\alpha,3\alpha)$ and $\text{Beta}(\alpha,17\alpha/3)$, respectively. Each point in Figure \ref{fig:5} is calculated by averaging $1000$ independent experiments, and for each experiment, we sample an independent $\theta$ from the beta distribution. This is to eliminate the variability of the experiments in sampling $\theta$. Thus, the points in Figure \ref{fig:5} are understood as the Bayes risk for the configuration model or the risk for the rank-one latent variable model. All the power curves decrease to $0$ as $\alpha$ increases. Moreover, the power also decreases with the sparsity level $h$. This is well suggested by the results in Theorem \ref{thm:configuration} and Theorem \ref{thm:latent} given that the signal-to-noise ratio in (\ref{eq:SNR-configuration}) and (\ref{eq:SNR-f}) is approximately $\frac{n^3\delta^2}{\rho^3}\asymp\frac{n^3h^6}{\alpha^2}$.

Finally, we study numerically the performance of the two network sampling procedures in Section \ref{sec:sample}. We first study sampling network nodes by checking the asymptotic distributions in Theorem \ref{thm:sample-node} and Corollary \ref{cor:chi-node}. We generate networks from Erd\H{o}s-R\'{e}nyi distributions. Four combinations of the network size $n$, the edge probability $p$, and sampling size $m$ are considered. For each scenario, we compute three statistics,
$$\frac{\sqrt{f(n,m)}T_3^{\mathcal{M}}}{\sqrt{p^3(1-p)^3+3p^4(1-p)^2}},\quad \frac{\sqrt{f(n,m)}T_2^{\mathcal{M}}}{\sqrt{3p^2(1-p)^2(1-3p)^2+9p^3(1-p)^3}},$$
and $(T^{\mathcal{M}})^2$ defined in (\ref{eq:chi^2-node}). Histograms of the three statistics are computed with $1000$ independent experiments. The results are plotted in Figure \ref{fig:6}. Each setting of $(n,p,m)$ occupies a column. The histograms of normalized $T_3^{\mathcal{M}}$, normalized $T_2^{\mathcal{M}}$ and $(T^{\mathcal{M}})^2$ are plotted in the three rows. Each histogram is superimposed with a theoretical density curve in red.
\begin{figure}[tbp]
\centering
\includegraphics[width=6in]{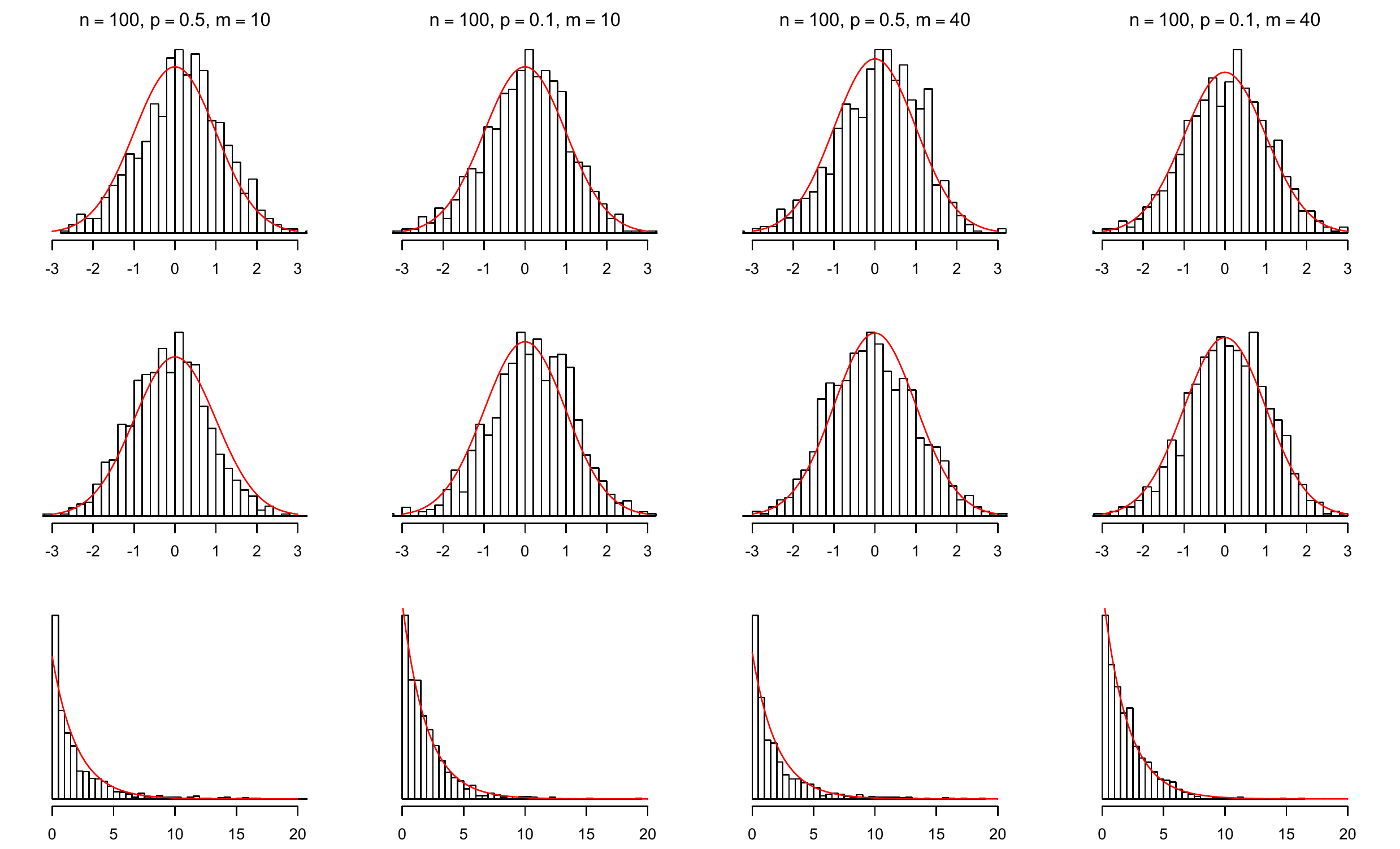}
\caption{Histograms of normalized $T_3^{\mathcal{M}}$, normalized $T_2^{\mathcal{M}}$ and $(T^{\mathcal{M}})^2$ for Erd\H{o}s-R\'{e}nyi graphs.}
\label{fig:6}
\end{figure}
The results show great matches between the theoretical asymptotic distributions and the empirical ones. This is true even for very small values of $m$.

We also study sampling network triples by checking the asymptotic distributions in Theorem \ref{thm:sample-triple} and Corollary \ref{cor:chi-triple}. We generate networks from Erd\H{o}s-R\'{e}nyi distributions. Four combinations of the network size $n$, the edge probability $p$, and sampling size $|\Delta|$ are considered. For each scenario, we compute three statistics,
$$\frac{\sqrt{\frac{{n\choose 3}|\Delta|}{{n\choose 3}+|\Delta|}}T_3^{\Delta}}{\sqrt{p^3(1-p)^3+3p^4(1-p)^2}},\quad \frac{\sqrt{\frac{{n\choose 3}|\Delta|}{{n\choose 3}+|\Delta|}}T_2^{\Delta}}{\sqrt{3p^2(1-p)^2(1-3p)^2+9p^3(1-p)^3}},$$
and $(T^{\Delta})^2$ defined in (\ref{eq:chi^2-node}). Histograms of the three statistics are computed with $1000$ independent experiments. The results are plotted in Figure \ref{fig:7}. Each setting of $(n,p,|\Delta|)$ occupies a column. The histograms of normalized $T_3^{\Delta}$, normalized $T_2^{\Delta}$ and $(T^{\Delta})^2$ are plotted in the three rows. Each histogram is superimposed with a theoretical density curve in red.
\begin{figure}[tbp]
\centering
\includegraphics[width=6in]{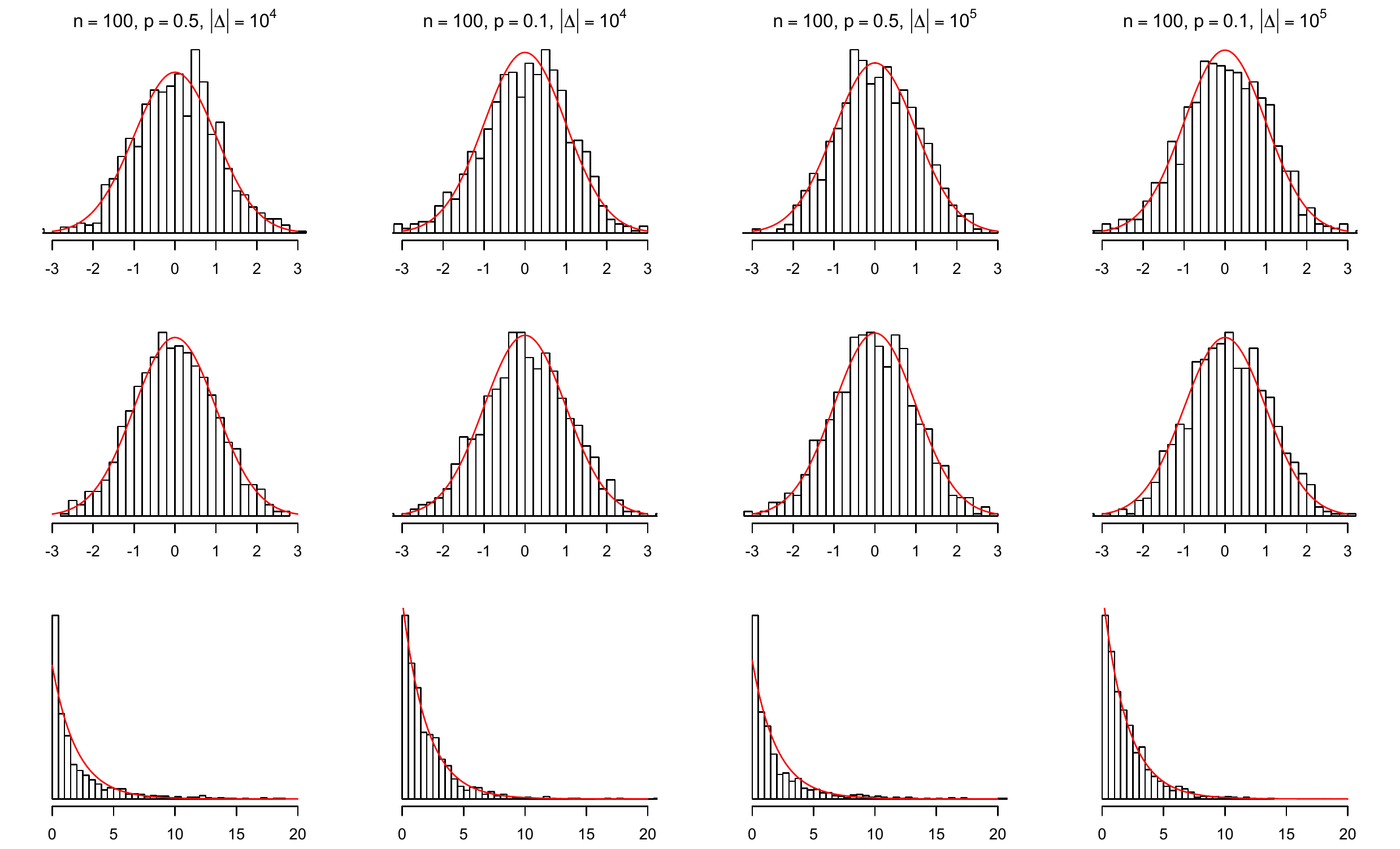}
\caption{Histograms of normalized $T_3^{\Delta}$, normalized $T_2^{\Delta}$ and $(T^{\Delta})^2$ for Erd\H{o}s-R\'{e}nyi graphs.}
\label{fig:7}
\end{figure}
Among the twelve plots in Figure \ref{fig:7}, there are slight mismatches between the histograms and the theoretical density curves for the chi-squared distributions on the first and the third columns. This is because $p=0.5$, and the assumption $p=o(1)$ in Corollary \ref{cor:chi-triple} is not satisfied. For the remaining ten plots, they all match the theoretical distributions quite well.

We close this section by studying the power of the test with the two sampling methods under the setting of the balanced 2-community stochastic block models. Figure \ref{fig:8} and Figure \ref{fig:9} show power curves for the balanced 2-community stochastic block model with $n=100$ and $(a,b)$ in the set $\{(0.3+0.01*i,0.2-0.01*i):i\in[10]\}$. Each point is an average of $1000$ independent experiments.
\begin{figure}[tbp]
\centering
\includegraphics[width=5in]{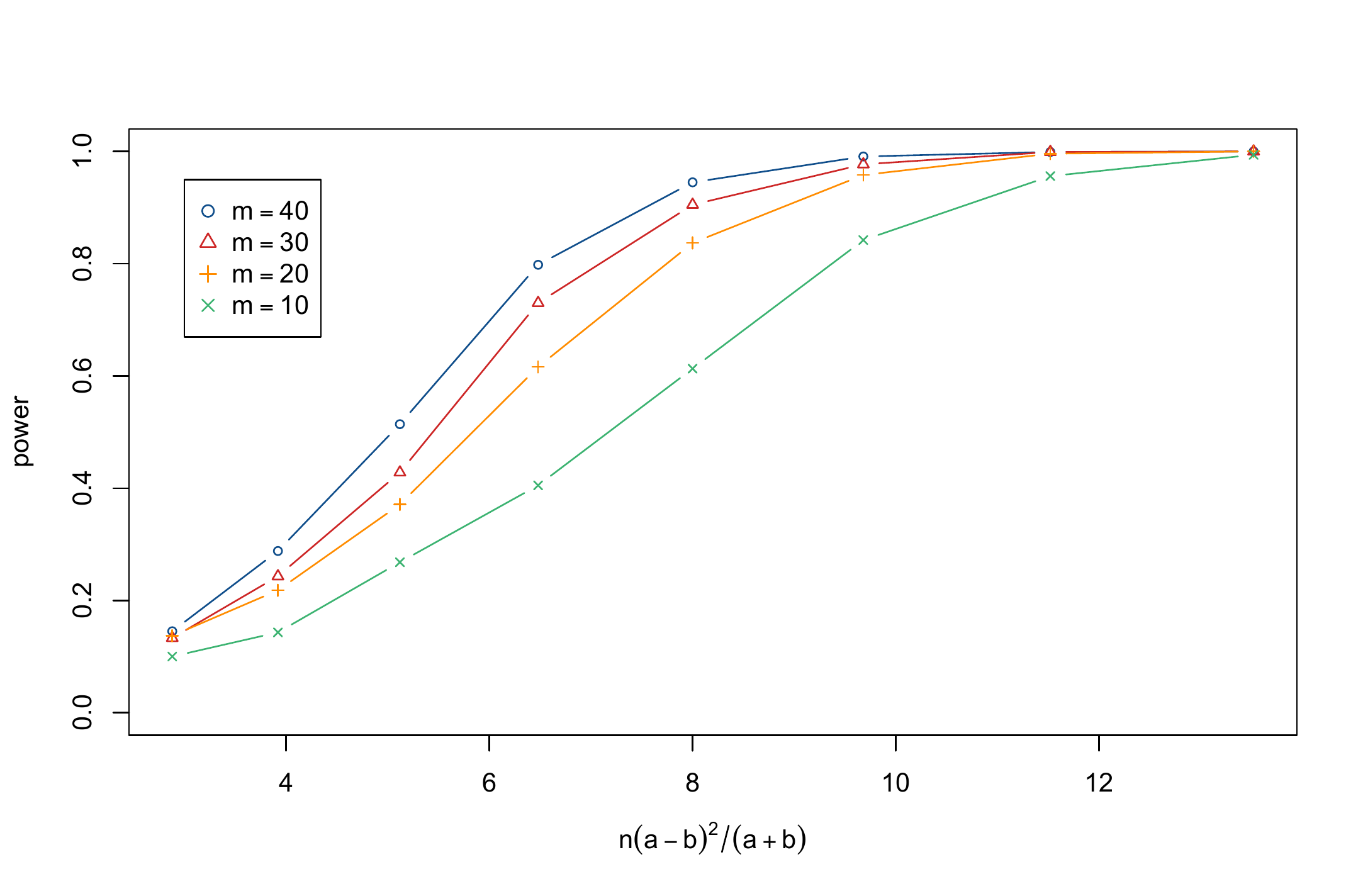}
\caption{Power of the proposed $0.05$-level test with node sampling under the balanced 2-community stochastic block models.}
\label{fig:8}
\end{figure}
For the testing procedure with node sampling, the cases $m=40, 30, 20, 10$ are considered. Figure \ref{fig:8} clearly shows a trade-off between the signal/sparsity and sampling budget. Namely, for the same value of power, it can be achieved with either a high signal-to-noise ratio and a low $m$ or a low signal-to-noise ratio but a high $m$.
\begin{figure}[tbp]
\centering
\includegraphics[width=5in]{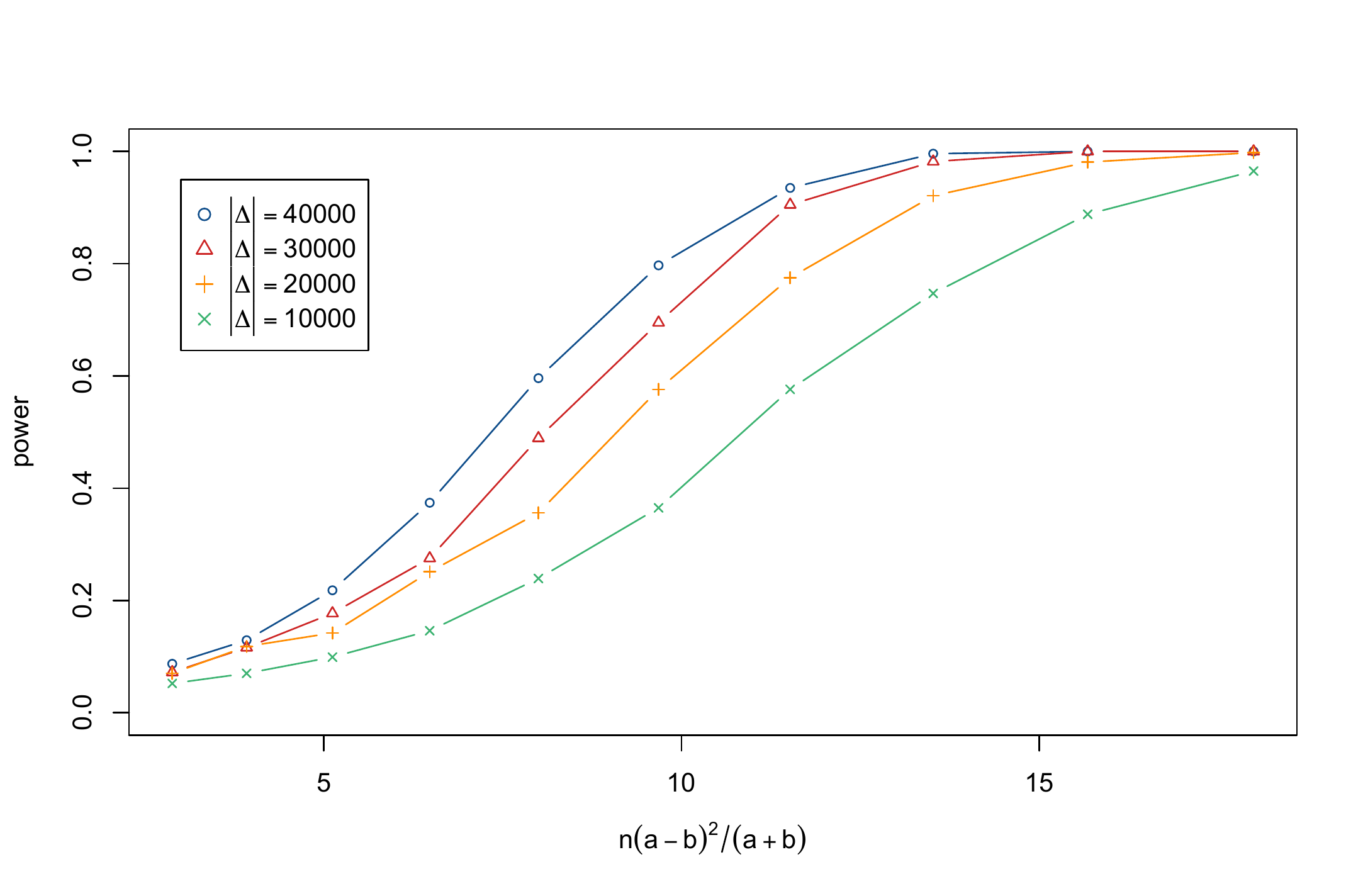}
\caption{Power of the proposed $0.05$-level test with triple sampling under the balanced 2-community stochastic block models.}
\label{fig:9}
\end{figure}
Figure \ref{fig:9} considers triple sampling with $|\Delta|=4*10^4, 3*10^4, 2*10^4, 10^4$. The same phenomenon is also found here. To better show the relation between sampling budget and network signal and sparsity, we plot the power against the adjusted signal-to-noise ratio suggested by Theorem \ref{thm:error-M} and Theorem \ref{thm:sample-triple}.
\begin{figure}[tbp]
\centering
\includegraphics[width=5in]{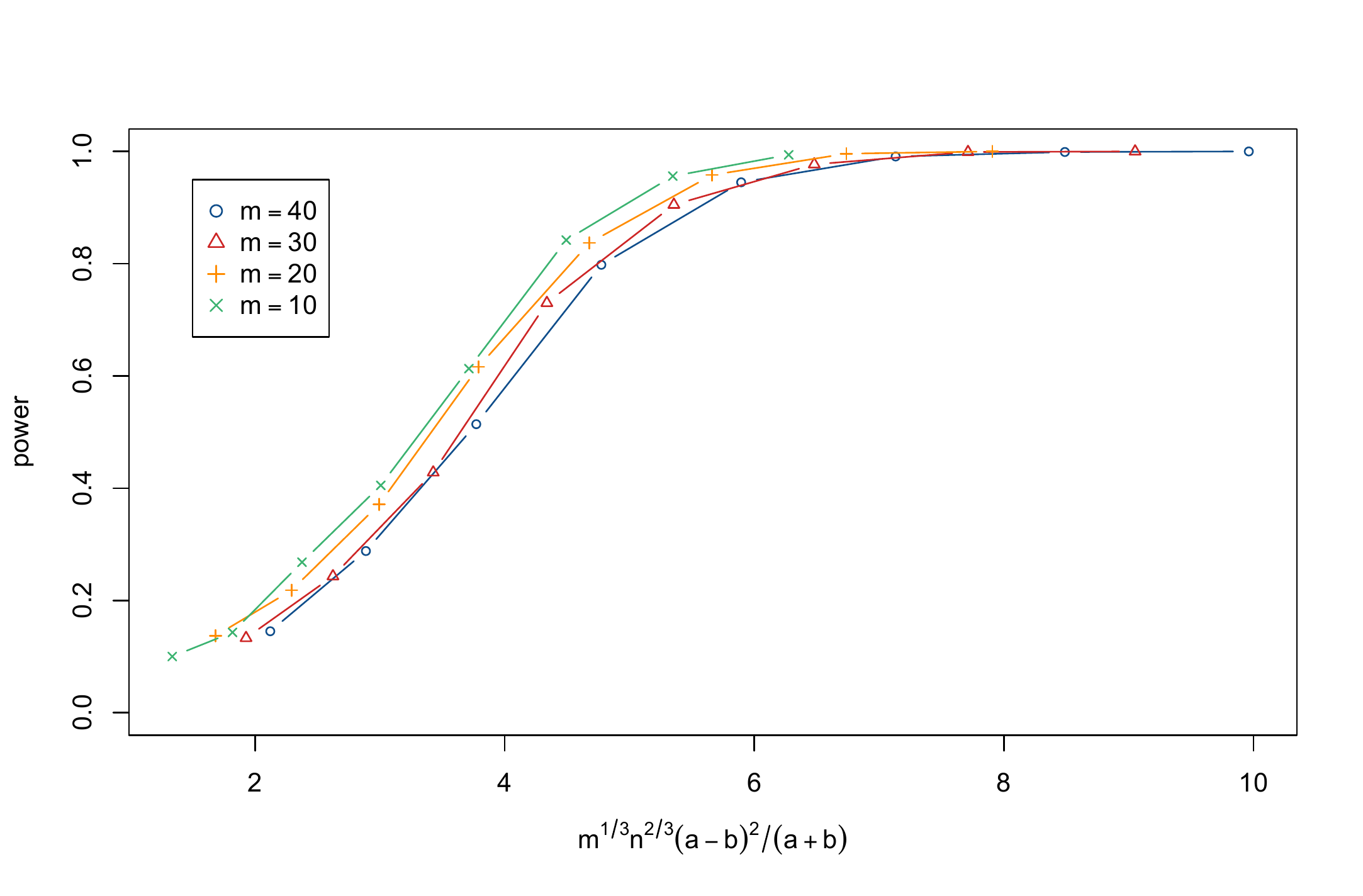}
\caption{Power of the proposed $0.05$-level test with node sampling under the balanced 2-community stochastic block models. The x-axis is the adjusted signal-to-noise ratio in Theorem \ref{thm:error-M}.}
\label{fig:10}
\end{figure}
Figure \ref{fig:10} shows the same set of power curves in Figure \ref{fig:8}. The only difference is that Figure \ref{fig:10} uses $\frac{m^{1/3}n^{2/3}(a-b)^2}{a+b}$ instead of $\frac{n(a-b)^2}{a+b}$ on the x-axis. Note that this is the adjusted signal-to-noise ratio given by (\ref{eq:SNR-M}) for node sampling. Interestingly, the four power curves in Figure \ref{fig:10} are well aligned, which validates the results in Theorem \ref{thm:error-M}.
\begin{figure}[tbp]
\centering
\includegraphics[width=5in]{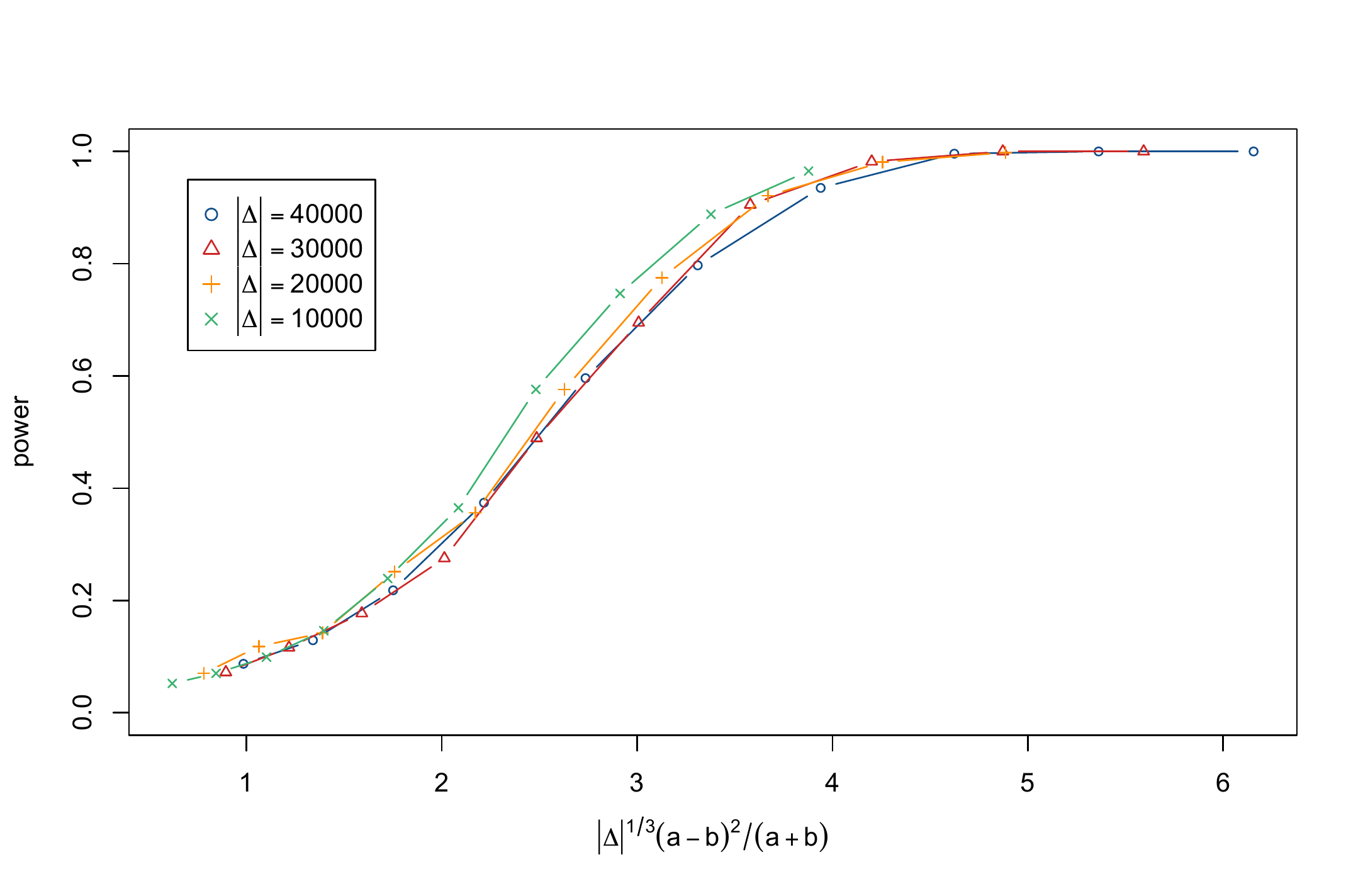}
\caption{Power of the proposed $0.05$-level test with triple sampling under the balanced 2-community stochastic block models. The x-axis is the adjusted signal-to-noise ratio in Theorem \ref{thm:sample-triple}.}
\label{fig:11}
\end{figure}
Similarly, Figure \ref{fig:11} is another version of Figure \ref{fig:9} by using $\frac{|\Delta|^{1/3}(a-b)^2}{a+b}$ on the x-axis. This number is given by (\ref{eq:SNR-Delta}) as the adjusted signal-to-noise ratio for triple sampling. Again, the four power curves are well aligned. This gives a numerical evidence that the results in Theorem \ref{thm:sample-triple} are sharp.


\section{Proofs}\label{sec:pf}

The section of proofs are organized as follows. Results related to asymptotic distributions are proved in Section \ref{sec:pf-ad}, which covers the proof of Theorem \ref{thm:vanilla}, Corollary \ref{cor:chi}, Theorem \ref{thm:type-1}, Theorem \ref{thm:sample-node} and Theorem \ref{thm:sample-triple}. Analyses of testing errors in Theorem \ref{thm:SBM}, Theorem \ref{thm:SBM-k}, Theorem \ref{thm:configuration}, Theorem \ref{thm:latent}, Theorem \ref{thm:error-M} and Theorem \ref{thm:error-Delta} are proved in Section \ref{sec:pf-power}. The lower bound result Theorem \ref{thm:SBM-lower} is proved in Section \ref{sec:pf-lb}. Finally, Propositions \ref{prop:corr}, \ref{prop:delta-configuration}, \ref{prop:delta-latent} and all technical lemmas are proved in Section \ref{sec:pf-tech}.

\subsection{Proofs of asymptotic distributions}\label{sec:pf-ad}

This section proves Theorem \ref{thm:vanilla}, Corollary \ref{cor:chi}, Theorem \ref{thm:type-1}, Theorem \ref{thm:sample-node} and Theorem \ref{thm:sample-triple}. Before stating the proofs of these results, we first state and prove another theorem. Define
\begin{eqnarray*}
R_3^{\mathcal{M}} &=& \frac{2}{m(n-1)(n-2)}\sum_{i=1}^m\sum_{\{1\leq j<k\leq n:j,k\neq i\}}(A_{ij}-p)(A_{ik}-p)(A_{jk}-p), \\
R_2^{\mathcal{M}} &=& \frac{2}{m(n-1)(n-2)}\sum_{i=1}^m\sum_{\{1\leq j<k\leq n:j,k\neq i\}}\Big[(A_{ij}-p)(A_{jk}-p) \\
&& +(A_{ik}-p)(A_{jk}-p)+(A_{ij}-p)(A_{ik}-p)\Big].
\end{eqnarray*}
The joint asymptotic distributions of $(R_3^{\mathcal{M}},R_2^{\mathcal{M}})$ are given by the following theorem. Recall that we require $m=m(n)$ is a nondecreasing function of $n$ that satisfies $m(1)=1$ and $m(l)\leq l$ for all $l\in[n]$.

\begin{theorem}\label{thm:main}
Assume $(1-p)^{-1}=O(1)$, $m\rightarrow\infty$ and $p^3mn^2\rightarrow\infty$, and then we have
$$\begin{pmatrix}
\sqrt{\frac{{f(n,m)}}{p^3(1-p)^3}}R_3^{\mathcal{M}} \\
\sqrt{\frac{{f(n,m)}}{3p^2(1-p)^2}} R_2^{\mathcal{M}}
\end{pmatrix}\leadsto N\left(\begin{pmatrix}
0 \\
0
\end{pmatrix},\begin{pmatrix}
1 & 0 \\
0 & 1
\end{pmatrix}\right),$$
under the null distribution.
\end{theorem}

In order to prove Theorem \ref{thm:main}, note that by Cram\'{e}r-Wold theorem, it is equivalent to prove
\begin{equation}
t_1 \sqrt{\frac{{f(n,m)}}{p^3(1-p)^3}}R_3^{\mathcal{M}} + t_2 \sqrt{\frac{{f(n,m)}}{3p^2(1-p)^2}} R_2^{\mathcal{M}} \leadsto N(0,1),\label{eq:to-prove}
\end{equation}
for any fixed $t_1$ and $t_2$ that satisfy $t_1^2+t_2^2=1$. We shall apply the following martingale central limit theorem to derive (\ref{eq:to-prove}). The version of martingale central limit theorem we use is taken from \cite{hall2014martingale}.

\begin{theorem}[\cite{hall2014martingale}]\label{thm:mds}
Suppose that for every $n\in\mathbb{N}$ and $k_n\rightarrow\infty$ the random variables $X_{n,1},...,X_{n,k_n}$ are a martingale difference sequence relative to an arbitrary filtration $\mathcal{F}_{n,1}\subset\mathcal{F}_{n,2}\subset\cdots \subset\mathcal{F}_{n,k_n}$. If
\begin{enumerate}
\item $\sum_{i=1}^{k_n}\mathbb{E}(X_{n,i}^2|\mathcal{F}_{n,i-1})\rightarrow 1$ in probability, and
\item $\sum_{i=1}^{k_n}\mathbb{E}\left(X_{n,i}^2\mathbb{I}_{\{|X_{n,i}|>\epsilon\}}|\mathcal{F}_{n,i-1}\right)\rightarrow 0$ in probability for every $\epsilon>0$,
\end{enumerate}
then $\sum_{i=1}^{k_n}X_{n,i}\leadsto N(0,1)$.
\end{theorem}

\begin{proof}[Proof of Theorem \ref{thm:main}]
Note that in (\ref{eq:to-prove}), $m=m(n)$ is a nondecreasing function of $n$. Recall the assumption that $m(1)=1$ and $m(l)\leq l$ for all $l\in[n]$. Without loss of generality, we take $\mathcal{M}=[m]$. We build a martingale sequence. Note that both $R_3^{\mathcal{M}}$ and $R_2^{\mathcal{M}}$ involve the summation $\sum_{i=1}^m\sum_{1\leq j<k\leq n:j,k,\neq i}(\cdot)$. It sums over all ${m\choose 3}$ triples from $[m]$ with multiplicity of $3$, $m{n-m \choose 2}$ triples with one node from $[m]$ and the other two from $[n]\backslash[m]$, and ${m\choose 2}(n-m)$ triples with two nodes from $[m]$ and the other one from $[n]\backslash[m]$ with multiplicity $2$. It is easy to check that
$$3{m\choose 3}+m{n-m\choose 2}+2{m\choose 2}(n-m)=\frac{m(n-1)(n-2)}{2}.$$
We collect all the $\frac{m(n-1)(n-2)}{2}$ triples in the multiset $\Delta_n$. Among all the triples in $\Delta_n$, the ones only use nodes from $[l]$ for some $l\in[n]$ are collected in the multiset $\Delta_{n,l}$. Given a triple $\lambda=(i,j,k)$, define
\begin{eqnarray}
\label{eq:def-Y} Y_{\lambda} &=& (A_{ij}-p)(A_{ik}-p)(A_{jk}-p), \\
\label{eq:def-W} W_{\lambda} &=& (A_{ij}-p)(A_{jk}-p)+(A_{ik}-p)(A_{jk}-p)+(A_{ij}-p)(A_{ik}-p).
\end{eqnarray}
For an $l\in[n]$, define
$$S_{n,l}=\frac{t_1\sum_{\lambda\in\Delta_{n,l}}Y_{\lambda}}{\frac{m(n-1)(n-2)}{2}\sqrt{\frac{p^3(1-p)^3}{f(n,m)}}}+\frac{t_2\sum_{\lambda\in\Delta_{n,l}}W_{\lambda}}{\frac{m(n-1)(n-2)}{2}\sqrt{\frac{3p^2(1-p)^2}{f(n,m)}}}.$$
Note that
$$\sum_{\lambda\in\Delta_{n,l}}Y_{\lambda}=\begin{cases}
3\sum_{1\leq i<j<k\leq l}Y_{(i,j,k)}, & l<m+1; \\
\sum_{i=1}^m\sum_{\{1\leq j<k\leq l:j,k\neq i\}}Y_{(i,j,k)}, & l\geq m+1.
\end{cases}$$
Similarly,
$$\sum_{\lambda\in\Delta_{n,l}}W_{\lambda}=\begin{cases}
3\sum_{1\leq i<j<k\leq l}W_{(i,j,k)}, & l<m+1; \\
\sum_{i=1}^m\sum_{1\leq j<k\leq l:j,k\neq i}W_{(i,j,k)}, & l\geq m+1.
\end{cases}$$
Therefore, we have
\begin{equation}
\sum_{\lambda\in\Delta_{n,l}}Y_{\lambda} - \sum_{\lambda\in\Delta_{n,l-1}}Y_{\lambda} = \begin{cases}
3\sum_{1\leq i<j\leq l-1}Y_{(i,j,l)}, & l<m+1; \\
2\sum_{1\leq i<j\leq m}Y_{(i,j,m+1)}, & l=m+1; \\
\sum_{i=1}^m\sum_{j=m+1}^{l-1}Y_{(i,j,l)} + 2\sum_{1\leq i<j\leq m}Y_{(i,j,l)}, & l>m+1.
\end{cases}\label{eq:Y-case}
\end{equation}
Similarly, we also have
\begin{equation}
\sum_{\lambda\in\Delta_{n,l}}W_{\lambda} - \sum_{\lambda\in\Delta_{n,l-1}}W_{\lambda} = \begin{cases}
3\sum_{1\leq i<j\leq l-1}W_{(i,j,l)}, & l<m+1; \\
2\sum_{1\leq i<j\leq m}W_{(i,j,m+1)}, & l=m+1; \\
\sum_{i=1}^m\sum_{j=m+1}^{l-1}W_{(i,j,l)} + 2\sum_{1\leq i<j\leq m}W_{(i,j,l)}, & l>m+1.
\end{cases}\label{eq:W-case}
\end{equation}
Define $\mathcal{F}_l$ to be the $\sigma$-algebra generated by the random variables $\{A_{ij}\}_{1\leq i<j\leq l}$. By (\ref{eq:Y-case}) and (\ref{eq:W-case}), it is not hard to check that
\begin{equation}
\mathbb{E}(S_{n,l}-S_{n,l-1}|\mathcal{F}_{l-1})=0.\label{eq:martingale}
\end{equation}
Therefore, for $X_{n,l}=S_{n,l}-S_{n,l-1}$, $\{X_{n,l}\}_l$ are a martingale difference sequence relative to the filtration $\{\mathcal{F}_l\}_l$. To apply Theorem \ref{thm:mds}, we need to calculate the asymptotic variance and check Lindeberg's condition, respectively.

\paragraph{Calculating asymptotic variance.}

Note that for any two triples $(i,j,k)$ and $(i',j',k')$,
\begin{equation}
\mathbb{E}Y_{(i,j,k)}W_{(i',j',k')}=0.\label{eq:inde-shapes}
\end{equation}
We first calculate the expectation of $\sum_{l=1}^n\mathbb{E}((S_{n,l}-S_{n,l-1})^2|\mathcal{F}_{l-1})$. Using the martingale property, we have
\begin{eqnarray}
\nonumber && \mathbb{E}\left[\sum_{l=1}^n\mathbb{E}((S_{n,l}-S_{n,l-1})^2|\mathcal{F}_{l-1})\right] \\
\nonumber &=& \mathbb{E}\left[\sum_{l=1}^n\mathbb{E}(S_{n,l}^2+S_{n,l-1}^2-2S_{n,l}S_{n,l-1}|\mathcal{F}_{l-1})\right] \\
\nonumber &=& \sum_{l=1}^n\mathbb{E}(S_{n,l}^2-S_{n,l-1}^2) \\
\nonumber &=& \mathbb{E}S_{n,n}^2 \\
\label{eq:decomp} &=& \mathbb{E}\left(\frac{t_1\sum_{\lambda\in\Delta_{n}}Y_{\lambda}}{\frac{m(n-1)(n-2)}{2}\sqrt{\frac{p^3(1-p)^3}{f(n,m)}}}\right)^2+\mathbb{E}\left(\frac{t_2\sum_{\lambda\in\Delta_{n}}W_{\lambda}}{\frac{m(n-1)(n-2)}{2}\sqrt{\frac{3p^2(1-p)^2}{f(n,m)}}}\right)^2,
\end{eqnarray}
where the last equality is due to (\ref{eq:inde-shapes}). Let $\Delta_n^{(1)}$ be the set of all triples from $[m]$, $\Delta_n^{(2)}$ be the set of all triples with one node from $[m]$ and the other two nodes from $[n]\backslash[m]$, and $\Delta_n^{(3)}$ be the set of all triples with two nodes from $[m]$ and the other one node from $[n]\backslash[m]$. It is easy to see that
$$|\Delta_n^{(1)}|={m\choose 3},\quad |\Delta_n^{(2)}|=m{n-m\choose 2},\quad\text{and}\quad |\Delta_n^{(3)}|={m\choose 2}(n-m).$$
With these notations, we have
$$\sum_{\lambda\in\Delta_n}Y_{\lambda}=3\sum_{\lambda\in\Delta_n^{(1)}}Y_{\lambda}+\sum_{\lambda\in\Delta_n^{(2)}}Y_{\lambda}+2\sum_{\lambda\in\Delta_n^{(3)}}Y_{\lambda}.$$
For any two triples $\lambda$ and $\lambda'$, direct calculation gives
\begin{equation}
\mathbb{E}Y_{\lambda}Y_{\lambda'}=\begin{cases}
p^3(1-p)^3, & \lambda=\lambda';\\
0, & \lambda\neq\lambda',
\end{cases}\label{eq:cross-m-Y}
\end{equation}
where we use $\lambda=\lambda'$ and $\lambda\neq \lambda'$ to denote the equality and inequality of two sets.
Therefore,
\begin{eqnarray*}
\mathbb{E}\left(\sum_{\lambda\in\Delta_{n}}Y_{\lambda}\right)^2 &=& 9\sum_{\lambda\in\Delta_n^{(1)}}\mathbb{E}Y_{\lambda}^2+\sum_{\lambda\in\Delta_n^{(2)}}\mathbb{E}Y_{\lambda}^2+4\sum_{\lambda\in\Delta_n^{(3)}}\mathbb{E}Y_{\lambda}^2 \\
&=& p^3(1-p)^3\left(9{m\choose 3} + m{n-m\choose 2} + 4{m\choose 2}(n-m)\right).
\end{eqnarray*}
Since for any two triples $\lambda$ and $\lambda'$, we also have
\begin{equation}
\mathbb{E}W_{\lambda}W_{\lambda'}=\begin{cases}
3p^2(1-p)^2, & \lambda=\lambda';\\
0, & \lambda\neq\lambda',
\end{cases}\label{eq:cross-m-W}
\end{equation}
similar calculation gives
$$\mathbb{E}\left(\sum_{\lambda\in\Delta_{n}}W_{\lambda}\right)^2=3p^2(1-p)^2\left(9{m\choose 3} + m{n-m\choose 2} + 4{m\choose 2}(n-m)\right).$$
Hence, by (\ref{eq:decomp}), we get
\begin{equation}
\mathbb{E}\left[\sum_{l=1}^n\mathbb{E}((S_{n,l}-S_{n,l-1})^2|\mathcal{F}_{l-1})\right]=1.\label{eq:expect}
\end{equation}

After calculating the expectation, we need to bound the variance $\Var{\left[\sum_{l=1}^n\mathbb{E}((S_{n,l}-S_{n,l-1})^2|\mathcal{F}_{l-1})\right]}$.
The first bound we use is
\begin{eqnarray}
\nonumber && \Var\left[\sum_{l=1}^n\mathbb{E}((S_{n,l}-S_{n,l-1})^2|\mathcal{F}_{l-1})\right] \\
\label{eq:var-var-Y} &\leq& 8t_1^4\frac{\Var\left[\sum_{l=1}^n\mathbb{E}\left(\left(\sum_{\lambda\in\Delta_{n,l}}Y_{\lambda} - \sum_{\lambda\in\Delta_{n,l-1}}Y_{\lambda}\right)^2\Big|\mathcal{F}_{l-1}\right)\right]}{\left(\frac{m(n-1)(n-2)}{2}\sqrt{\frac{p^3(1-p)^3}{f(n,m)}}\right)^4} \\
\label{eq:var-var-W} && + 8t_2^4\frac{\Var\left[\sum_{l=1}^n\mathbb{E}\left(\left(\sum_{\lambda\in\Delta_{n,l}}W_{\lambda} - \sum_{\lambda\in\Delta_{n,l-1}}W_{\lambda}\right)^2\Big|\mathcal{F}_{l-1}\right)\right]}{\left(\frac{m(n-1)(n-2)}{2}\sqrt{\frac{3p^2(1-p)^2}{f(n,m)}}\right)^4}.
\end{eqnarray}
We will give bounds for (\ref{eq:var-var-Y}) and (\ref{eq:var-var-W}), respectively. Note that for any $\{i,j\}$ and $\{i',j'\}$, we have
\begin{equation}
\mathbb{E}\left(Y_{(i,j,l)}Y_{(i',j',l)}|\mathcal{F}_{l-1}\right) = \begin{cases}
p^2(1-p)^2(A_{ij}-p)^2, & \{i,j\}=\{i',j'\}; \\
0, & \{i,j\}\neq\{i',j'\}.
\end{cases}\label{eq:Y-case-condition}
\end{equation}
By (\ref{eq:Y-case}) and (\ref{eq:Y-case-condition}), we get
\begin{eqnarray*}
&& \sum_{l=1}^n\mathbb{E}\left(\left(\sum_{\lambda\in\Delta_{n,l}}Y_{\lambda} - \sum_{\lambda\in\Delta_{n,l-1}}Y_{\lambda}\right)^2\Big|\mathcal{F}_{l-1}\right) \\
&=& \sum_{l=1}^m\mathbb{E}\left(\left(3\sum_{1\leq i<j\leq l-1}Y_{(i,j,l)}\right)^2\Big|\mathcal{F}_{l-1}\right) + \mathbb{E}\left(\left(2\sum_{1\leq i<j\leq m}Y_{(i,j,m+1)}\right)^2\Big|\mathcal{F}_{m}\right) \\
&& + \sum_{l=m+2}^n\mathbb{E}\left(\left(\sum_{i=1}^m\sum_{j=m+1}^{l-1}Y_{(i,j,l)} + 2\sum_{1\leq i<j\leq m}Y_{(i,j,l)} \right)^2\Big|\mathcal{F}_{l-1}\right) \\
&=& 9\sum_{l=1}^m\sum_{1\leq i<j\leq l-1}p^2(1-p)^2(A_{ij}-p)^2 + 4\sum_{1\leq i<j\leq m}p^2(1-p)^2(A_{ij}-p)^2 \\
&& + \sum_{l=m+2}^n\sum_{i=1}^m\sum_{j=m+1}^{l-1}p^2(1-p)^2(A_{ij}-p)^2 + 4\sum_{l=m+2}^n\sum_{1\leq i<j\leq m}p^2(1-p)^2(A_{ij}-p)^2.
\end{eqnarray*}
We will give bounds for the variance of the four terms above.
\begin{eqnarray*}
&& \Var\left(\sum_{l=1}^m\sum_{1\leq i<j\leq l-1}p^2(1-p)^2(A_{ij}-p)^2\right) \\
&\leq& p^4\Var\left(\sum_{j=2}^{m-1}(m-j)\sum_{i=1}^{j-1}(A_{ij}-p)^2\right) \\
&=& p^4\sum_{j=2}^{m-1}(m-j)\sum_{i=1}^{j-1}\Var((A_{ij}-p)^2) \\
&\leq& p^5\sum_{j=2}^{m-1}(m-j)^2(j-1) \\
&=& O(m^4p^5),
\end{eqnarray*}
\begin{eqnarray*}
&& \Var\left(\sum_{1\leq i<j\leq m}p^2(1-p)^2(A_{ij}-p)^2\right) \\
&=& \sum_{1\leq i<j\leq m}p^4(1-p)^4\Var((A_{ij}-p)^2) \\
&=& O(m^2p^5),
\end{eqnarray*}
\begin{eqnarray*}
&& \Var\left(\sum_{l=m+1}^n\sum_{i=1}^m\sum_{j=m+1}^{l-1}p^2(1-p)^2(A_{ij}-p)^2\right) \\
&=& \Var\left(\sum_{j=m+1}^{n-1}(j-m)\sum_{i=1}^m(A_{ij}-p)^2p^2(1-p)^2\right) \\
&\leq& p^4\sum_{j=m+1}^{n-1}(j-m)^2\sum_{i=1}^m\Var((A_{ij}-p)^2) \\
&=& O(p^5n^3m),
\end{eqnarray*}
\begin{eqnarray*}
&& \Var\left(\sum_{l=m+2}^n\sum_{1\leq i<j\leq m}p^2(1-p)^2(A_{ij}-p)^2\right) \\
&\leq& p^4n^2\Var\left(\sum_{1\leq i<j\leq m}(A_{ij}-p)^2\right) \\
&=& O(p^5n^2m^2).
\end{eqnarray*}
Hence,
$$\Var\left[\sum_{l=1}^n\mathbb{E}\left(\left(\sum_{\lambda\in\Delta_{n,l}}Y_{\lambda} - \sum_{\lambda\in\Delta_{n,l-1}}Y_{\lambda}\right)^2\Big|\mathcal{F}_{l-1}\right)\right]=O(p^5n^3m).$$
This leads to a bound for (\ref{eq:var-var-Y}), which is
$$8t_1^4\frac{\Var\left[\sum_{l=1}^n\mathbb{E}\left(\left(\sum_{\lambda\in\Delta_{n,l}}Y_{\lambda} - \sum_{\lambda\in\Delta_{n,l-1}}Y_{\lambda}\right)^2\Big|\mathcal{F}_{l-1}\right)\right]}{\left(\frac{m(n-1)(n-2)}{2}\sqrt{\frac{p^3(1-p)^3}{f(n,m)}}\right)^4}=O\left(\frac{1}{pmn}\right).$$
Now we give a bound for (\ref{eq:var-var-W}). Note that for any $\{i,j\}$ and $\{i',j'\}$, we have
\begin{equation}
\mathbb{E}\left(W_{(i,j,l)}W_{(i',j',l)}|\mathcal{F}_{l-1}\right) = \begin{cases}
p^2(1-p)^2+2p(1-p)(A_{ij}-p)^2, & \{i,j\}=\{i',j'\}; \\
p(1-p)(A_{ij}-p)(A_{i'j}-p), & i\neq i', j=j'; \\
0, & \{i,j\}\cap\{i',j'\}=\varnothing.
\end{cases}\label{eq:W-case-condition}
\end{equation}
By (\ref{eq:W-case}) and (\ref{eq:W-case-condition}), we get
\begin{eqnarray*}
&& \sum_{l=1}^n\mathbb{E}\left(\left(\sum_{\lambda\in\Delta_{n,l}}W_{\lambda} - \sum_{\lambda\in\Delta_{n,l-1}}W_{\lambda}\right)^2\Big|\mathcal{F}_{l-1}\right) \\
&=& \sum_{l=1}^m\mathbb{E}\left(\left(3\sum_{1\leq i<j\leq l-1}W_{(i,j,l)}\right)^2\Big|\mathcal{F}_{l-1}\right) + \mathbb{E}\left(\left(2\sum_{1\leq i<j\leq m}W_{(i,j,m+1)}\right)^2\Big|\mathcal{F}_{m}\right) \\
&& + \sum_{l=m+2}^n\mathbb{E}\left(\left(\sum_{i=1}^m\sum_{j=m+1}^{l-1}W_{(i,j,l)} + 2\sum_{1\leq i<j\leq m}W_{(i,j,l)} \right)^2\Big|\mathcal{F}_{l-1}\right) \\
&\leq& 9\sum_{l=1}^m\sum_{1\leq i<j\leq l-1}[p^2(1-p)^2+2p(1-p)(A_{ij}-p)^2] \\
&& + 4\sum_{1\leq i<j\leq m}[p^2(1-p)^2+2p(1-p)(A_{ij}-p)^2] \\
&& + \sum_{l=m+2}^n\sum_{i=1}^m\sum_{j=m+1}^{l-1}[p^2(1-p)^2+2p(1-p)(A_{ij}-p)^2] \\
&& + 4\sum_{l=m+2}^n\sum_{1\leq i<j\leq m}[p^2(1-p)^2+2p(1-p)(A_{ij}-p)^2] \\
&& + 18\sum_{l=1}^m\sum_{1\leq i<j\leq l-1}\sum_{\{1\leq i'\leq l-1: i'\neq i,j\}}p(1-p)(A_{ij}-p)(A_{i'j}-p) \\
&& + 8\sum_{1\leq i<j\leq m}\sum_{\{1\leq i'\leq m:i'\neq i,j\}}p(1-p)(A_{ij}-p)(A_{i'j}-p) \\
&& + 4\sum_{l=m+2}^n\sum_{i=1}^m\sum_{j=m+1}^{l-1}\sum_{\{1\leq i'\leq m:i'\neq i\}}p(1-p)(A_{ij}-p)(A_{i'j}-p) \\
&& + 4\sum_{l=m+2}^n\sum_{i=1}^m\sum_{j=m+1}^{l-1}\sum_{\{m+1\leq j'\leq l-1:j'\neq j\}}p(1-p)(A_{ij}-p)(A_{ij'}-p) \\
&& + 16\sum_{l=m+2}^n\sum_{1\leq i<j\leq m}\sum_{\{1\leq i'\leq m:i'\neq i,j\}}p(1-p)(A_{ij}-p)(A_{i'j}-p).
\end{eqnarray*}
We will give bounds for the variance of all the terms above.
\begin{eqnarray*}
&& \Var\left[\sum_{l=1}^m\sum_{1\leq i<j\leq l-1}[p^2(1-p)^2+2p(1-p)(A_{ij}-p)^2]\right] \\
&\leq& 4p^2\Var\left[\sum_{l=1}^m\sum_{1\leq i<j\leq l-1}(A_{ij}-p)^2\right] \\
&=& O(m^4p^3),
\end{eqnarray*}
\begin{eqnarray*}
&& \Var\left[\sum_{1\leq i<j\leq m}[p^2(1-p)^2+2p(1-p)(A_{ij}-p)^2]\right] \\
&\leq& 4p^2\Var\left[\sum_{1\leq i<j\leq m}(A_{ij}-p)^2\right] \\
&=& O(m^2p^3),
\end{eqnarray*}
\begin{eqnarray*}
&& \Var\left[\sum_{l=m+2}^n\sum_{i=1}^m\sum_{j=m+1}^{l-1}[p^2(1-p)^2+2p(1-p)(A_{ij}-p)^2]\right] \\
&\leq& 4p^2\Var\left[\sum_{l=m+2}^n\sum_{i=1}^m\sum_{j=m+1}^{l-1}(A_{ij}-p)^2\right] \\
&\leq& 4p^2\Var\left(\sum_{j=m+1}^{n-1}(j-m)\sum_{i=1}^m(A_{ij}-p)^2\right) \\
&\leq& 4p^2\sum_{j=m+1}^{n-1}(j-m)^2\sum_{i=1}^m\Var((A_{ij}-p)^2) \\
&=& O(p^3n^3m),
\end{eqnarray*}
\begin{eqnarray*}
&& \Var\left[\sum_{l=m+2}^n\sum_{1\leq i<j\leq m}[p^2(1-p)^2+2p(1-p)(A_{ij}-p)^2]\right] \\
&\leq& 4p^2\Var\left[\sum_{l=m+2}^n\sum_{1\leq i<j\leq m}(A_{ij}-p)^2\right] \\
&=& O(p^3n^2m^2),
\end{eqnarray*}

\begin{eqnarray*}
&& \Var\left[\sum_{l=1}^m\sum_{1\leq i<j\leq l-1}\sum_{\{1\leq i'\leq l-1: i'\neq i,j\}}p(1-p)(A_{ij}-p)(A_{i'j}-p)\right] \\
&\leq& p^2\mathbb{E}\left[\sum_{l=1}^m\sum_{1\leq i<j\leq l-1}\sum_{\{1\leq i'\leq l-1: i'\neq i,j\}}(A_{ij}-p)(A_{i'j}-p)\right]^2 \\
&\leq& p^2m\sum_{l=1}^m\mathbb{E}\left[\sum_{1\leq i<j\leq l-1}\sum_{\{1\leq i'\leq l-1: i'\neq i,j\}}(A_{ij}-p)(A_{i'j}-p)\right]^2 \\
&=& p^2m\sum_{l=1}^m\sum_{1\leq i<j\leq l-1}\sum_{\{1\leq i'\leq l-1: i'\neq i,j\}}\mathbb{E}(A_{ij}-p)^2(A_{i'j}-p)^2 \\
&=& O(p^4m^5),
\end{eqnarray*}
\begin{eqnarray*}
&& \Var\left[\sum_{1\leq i<j\leq m}\sum_{\{1\leq i'\leq m:i'\neq i,j\}}p(1-p)(A_{ij}-p)(A_{i'j}-p)\right] \\
&\leq& p^2\mathbb{E}\left[\sum_{1\leq i<j\leq m}\sum_{\{1\leq i'\leq m:i'\neq i,j\}}(A_{ij}-p)(A_{i'j}-p)\right]^2 \\
&=& O(p^4m^3),
\end{eqnarray*}
\begin{eqnarray*}
&& \Var\left[\sum_{l=m+2}^n\sum_{i=1}^m\sum_{j=m+1}^{l-1}\sum_{\{1\leq i'\leq m:i'\neq i\}}p(1-p)(A_{ij}-p)(A_{i'j}-p)\right] \\
&\leq& p^2\mathbb{E}\left[\sum_{l=m+2}^n\sum_{i=1}^m\sum_{j=m+1}^{l-1}\sum_{\{1\leq i'\leq m:i'\neq i\}}(A_{ij}-p)(A_{i'j}-p)\right]^2 \\
&\leq& p^2n\sum_{l=m+2}^n\mathbb{E}\left[\sum_{i=1}^m\sum_{j=m+1}^{l-1}\sum_{\{1\leq i'\leq m:i'\neq i\}}(A_{ij}-p)(A_{i'j}-p)\right]^2 \\
&=& p^2n\sum_{l=m+2}^n\sum_{i=1}^m\sum_{j=m+1}^{l-1}\sum_{\{1\leq i'\leq m:i'\neq i\}}\mathbb{E}(A_{ij}-p)^2(A_{i'j}-p)^2 \\
&=& O(p^4n^3m^2),
\end{eqnarray*}
\begin{eqnarray*}
&& \Var\left[\sum_{l=m+2}^n\sum_{i=1}^m\sum_{j=m+1}^{l-1}\sum_{\{m+1\leq j'\leq l-1:j'\neq j\}}p(1-p)(A_{ij}-p)(A_{ij'}-p)\right] \\
&\leq& p^2\mathbb{E}\left[\sum_{l=m+2}^n\sum_{i=1}^m\sum_{j=m+1}^{l-1}\sum_{\{m+1\leq j'\leq l-1:j'\neq j\}}(A_{ij}-p)(A_{ij'}-p)\right]^2 \\
&\leq& p^2n\sum_{l=m+2}^n\mathbb{E}\left[\sum_{i=1}^m\sum_{j=m+1}^{l-1}\sum_{\{m+1\leq j'\leq l-1:j'\neq j\}}(A_{ij}-p)(A_{ij'}-p)\right]^2 \\
&=& O(p^4n^4m),
\end{eqnarray*}
\begin{eqnarray*}
&& \Var\left[\sum_{l=m+2}^n\sum_{1\leq i<j\leq m}\sum_{\{1\leq i'\leq m:i'\neq i,j\}}p(1-p)(A_{ij}-p)(A_{i'j}-p)\right] \\
&\leq& p^2\mathbb{E}\left[\sum_{l=m+2}^n\sum_{1\leq i<j\leq m}\sum_{\{1\leq i'\leq m:i'\neq i,j\}}(A_{ij}-p)(A_{i'j}-p)\right]^2 \\
&\leq& p^2n\sum_{l=m+2}^n\mathbb{E}\left[\sum_{1\leq i<j\leq m}\sum_{\{1\leq i'\leq m:i'\neq i,j\}}(A_{ij}-p)(A_{i'j}-p)\right]^2 \\
&=& O(p^4n^2m^3).
\end{eqnarray*}
Combining all the variance bounds above, we obtain
$$\Var\left[\sum_{l=1}^n\mathbb{E}\left(\left(\sum_{\lambda\in\Delta_{n,l}}W_{\lambda} - \sum_{\lambda\in\Delta_{n,l-1}}W_{\lambda}\right)^2\Big|\mathcal{F}_{l-1}\right)\right]=O(p^3n^3m+p^4n^4m).$$
This leads to a bound for (\ref{eq:var-var-W}), which is
$$8t_2^4\frac{\Var\left[\sum_{l=1}^n\mathbb{E}\left(\left(\sum_{\lambda\in\Delta_{n,l}}W_{\lambda} - \sum_{\lambda\in\Delta_{n,l-1}}W_{\lambda}\right)^2\Big|\mathcal{F}_{l-1}\right)\right]}{\left(\frac{m(n-1)(n-2)}{2}\sqrt{\frac{3p^2(1-p)^2}{f(n,m)}}\right)^4}=O\left(\frac{1}{pmn}+\frac{1}{m}\right).$$
Hence, by combining the bounds for (\ref{eq:var-var-Y}) and (\ref{eq:var-var-W}), as long as $pmn\rightarrow\infty$ and $m\rightarrow\infty$, we have
$$\Var\left[\sum_{l=1}^n\mathbb{E}((S_{n,l}-S_{n,l-1})^2|\mathcal{F}_{l-1})\right]=o(1).$$
Together with (\ref{eq:expect}), we obtain that
$$\sum_{l=1}^n\mathbb{E}((S_{n,l}-S_{n,l-1})^2|\mathcal{F}_{l-1})=1+o_P(1).$$

\paragraph{Checking Lindeberg's condition.}

Now we start to establish Lindeberg's condition. We first reduce it to a fourth moment condition.
\begin{eqnarray}
\nonumber && \sum_{l=1}^n\mathbb{E}\left((S_{n,l}-S_{n,l-1})^2\mathbb{I}_{\{|S_{n,l}-S_{n,l-1}|>\epsilon\}}\Big|\mathcal{F}_{l-1}\right) \\
\nonumber &\leq& \sum_{l=1}^n\sqrt{\mathbb{E}\left((S_{n,l}-S_{n,l-1})^4\Big|\mathcal{F}_{l-1}\right)}\sqrt{\mathbb{P}\left(|S_{n,l}-S_{n,l-1}|>\epsilon\Big|\mathcal{F}_{l-1}\right)} \\
\nonumber &\leq& \epsilon^{-2}\sum_{l=1}^n\mathbb{E}\left((S_{n,l}-S_{n,l-1})^4\Big|\mathcal{F}_{l-1}\right) \\
\nonumber &\leq& 8\epsilon^{-2}t_1^4\sum_{l=1}^n\mathbb{E}\left(\left(\frac{\sum_{\lambda\in\Delta_{n,l}}Y_{\lambda}-\sum_{\lambda\in\Delta_{n,l-1}}Y_{\lambda}}{\frac{m(n-1)(n-2)}{2}\sqrt{\frac{p^3(1-p)^3}{f(n,m)}}}\right)^4\Big|\mathcal{F}_{l-1}\right) \\
\nonumber && + 8\epsilon^{-2}t_2^4\sum_{l=1}^n\mathbb{E}\left(\left(\frac{\sum_{\lambda\in\Delta_{n,l}}W_{\lambda}-\sum_{\lambda\in\Delta_{n,l-1}}W_{\lambda}}{\frac{m(n-1)(n-2)}{2}\sqrt{\frac{3p^2(1-p)^2}{f(n,m)}}}\right)^4\Big|\mathcal{F}_{l-1}\right).
\end{eqnarray}
Thus, it is sufficient to show that
\begin{equation}
\sum_{l=1}^n\mathbb{E}\left(\frac{\sum_{\lambda\in\Delta_{n,l}}Y_{\lambda}-\sum_{\lambda\in\Delta_{n,l-1}}Y_{\lambda}}{\frac{m(n-1)(n-2)}{2}\sqrt{\frac{p^3(1-p)^3}{f(n,m)}}}\right)^4=o(1),\label{eq:fourth-Y} 
\end{equation}
and
\begin{equation}
\sum_{l=1}^n\mathbb{E}\left(\frac{\sum_{\lambda\in\Delta_{n,l}}W_{\lambda}-\sum_{\lambda\in\Delta_{n,l-1}}W_{\lambda}}{\frac{m(n-1)(n-2)}{2}\sqrt{\frac{3p^2(1-p)^2}{f(n,m)}}}\right)^4=o(1).\label{eq:fourth-W} 
\end{equation}
We further decompose the left hand side of (\ref{eq:fourth-Y}).
By (\ref{eq:Y-case}), we get
\begin{eqnarray*}
&& \sum_{l=1}^n\mathbb{E}\left(\sum_{\lambda\in\Delta_{n,l}}Y_{\lambda} - \sum_{\lambda\in\Delta_{n,l-1}}Y_{\lambda}\right)^4 \\
&=& \sum_{l=1}^m\mathbb{E}\left(3\sum_{1\leq i<j\leq l-1}Y_{(i,j,l)}\right)^4 + \mathbb{E}\left(2\sum_{1\leq i<j\leq m}Y_{(i,j,m+1)}\right)^4 \\
&& + \sum_{l=m+2}^n\mathbb{E}\left(\sum_{i=1}^m\sum_{j=m+1}^{l-1}Y_{(i,j,l)} + 2\sum_{1\leq i<j\leq m}Y_{(i,j,l)} \right)^4 \\
&\leq&  81\sum_{l=1}^m\mathbb{E}\left(\sum_{1\leq i<j\leq l-1}Y_{(i,j,l)}\right)^4 + 16\mathbb{E}\left(\sum_{1\leq i<j\leq m}Y_{(i,j,m+1)}\right)^4 \\
&& + 8\sum_{l=m+2}^n\mathbb{E}\left(\sum_{i=1}^m\sum_{j=m+1}^{l-1}Y_{(i,j,l)}\right)^4 + 128\sum_{l=m+2}^n\mathbb{E}\left(\sum_{1\leq i<j\leq m}Y_{(i,j,l)} \right)^4.
\end{eqnarray*}
We will give bounds to the four terms above, respectively. For any pair $\tau=(i,j)$, we use the notation $(\tau,l)$ to denote the triple $(i,j,l)$. In order to deal with fourth moments, we need to study $M_{(\tau_1,\tau_2,\tau_3,\tau_4,l)}=\mathbb{E}Y_{(\tau_1,l)}Y_{(\tau_2,l)}Y_{(\tau_3,l)}Y_{(\tau_4,l)}$ for various situations of $\tau_1,\tau_2,\tau_3,\tau_4$. There are several different scenarios to consider. Note that the relations $=$ and $\neq$ are understood as relations for sets.
\begin{enumerate}
\item When $\tau_1=\tau_2=\tau_3=\tau_4$,
\begin{equation}
M_{(\tau_1,\tau_2,\tau_3,\tau_4,l)} = (p(1-p)^4+(1-p)p^4)^3.\label{eq:M1}
\end{equation}
\item When $\tau_1=\tau_2\neq \tau_3=\tau_4$,
\begin{equation}
M_{(\tau_1,\tau_2,\tau_3,\tau_4,l)} = \begin{cases}
p^6(1-p)^6, & \tau_1\cap\tau_3=\varnothing, \\
p^4(1-p)^4(p(1-p)^4+(1-p)p^4), & |\tau_1\cap\tau_3|=1.
\end{cases}\label{eq:M2}
\end{equation}
\item When $\tau_1\neq\tau_2=\tau_3=\tau_4$,
\begin{equation}
M_{(\tau_1,\tau_2,\tau_3,\tau_4,l)} = 0.\label{eq:M3}
\end{equation}
\item When $\tau_1\neq\tau_2\neq\tau_3\neq \tau_1=\tau_4$,
\begin{equation}
M_{(\tau_1,\tau_2,\tau_3,\tau_4,l)} = 0.\label{eq:M4}
\end{equation}
\item When $\tau_1$, $\tau_2$, $\tau_3$, $\tau_4$ are four different edges,
\begin{equation}
M_{(\tau_1,\tau_2,\tau_3,\tau_4,l)} = 0.\label{eq:M5}
\end{equation}
\end{enumerate}
Equipped with the moments calculations (\ref{eq:M1})-(\ref{eq:M5}), we are ready to proceed. Define the set
$$\mathcal{T}_l=\{(i,j):1\leq i<j\leq l\}.$$
For the first term,
\begin{eqnarray*}
&& \sum_{l=1}^m\mathbb{E}\left(\sum_{1\leq i<j\leq l-1}Y_{(i,j,l)}\right)^4 \\
&=& \sum_{l=1}^m\sum_{\tau_1\in\mathcal{T}_{l-1}}\sum_{\tau_2\in\mathcal{T}_{l-1}}\sum_{\tau_3\in\mathcal{T}_{l-1}}\sum_{\tau_4\in\mathcal{T}_{l-1}}M_{(\tau_1,\tau_2,\tau_3,\tau_4,l)} \\
&=& \sum_{l=1}^m\sum_{\tau\in\mathcal{T}_{l-1}}M_{(\tau,\tau,\tau,\tau,l)} + {4\choose 2}\sum_{l=1}^m\sum_{\{\tau_1,\tau_2\in\mathcal{T}_{l-1}: \tau_1\neq\tau_2\}}M_{(\tau_1,\tau_1,\tau_2,\tau_2,l)} \\
&\leq& \sum_{l=1}^m\sum_{\tau\in\mathcal{T}_{l-1}}(p(1-p)^4+(1-p)p^4)^3 + 6\sum_{l=1}^m\sum_{\tau_1\in\mathcal{T}_{l-1}}\sum_{\{\tau_2\in\mathcal{T}_{l-1}:\tau_2\cap\tau_1=\varnothing\}}p^6(1-p)^6 \\
&& + 6\sum_{l=1}^m\sum_{(i,j)\in\mathcal{T}_{l-1}}\sum_{\{i'\in[l-1]:i'\neq i,j\}}p^4(1-p)^4(p(1-p)^4+(1-p)p^4) \\
&& + 6\sum_{l=1}^m\sum_{(i,j)\in\mathcal{T}_{l-1}}\sum_{\{j'\in[l-1]:j'\neq i,j\}}p^4(1-p)^4(p(1-p)^4+(1-p)p^4) \\
&=& O\left(m^3p^3 + m^5p^6 + m^4p^5\right).
\end{eqnarray*}
The second term is obviously of the smaller order of the first term. For the third term, define
$$\mathcal{T}_{m,l}=\{(i,j):i\in[m], j\in[l]\backslash[m]\},$$
and then
\begin{eqnarray*}
&& \sum_{l=m+2}^n\mathbb{E}\left(\sum_{i=1}^m\sum_{j=m+1}^{l-1}Y_{(i,j,l)}\right)^4 \\
&=& \sum_{l=m+2}^n\sum_{\tau\in\mathcal{T}_{m,l-1}}M_{(\tau,\tau,\tau,\tau,l)} + {4\choose 2}\sum_{l=m+2}^n\sum_{\{\tau_1,\tau_2\in\mathcal{T}_{m,l-1}: \tau_1\neq\tau_2\}}M_{(\tau_1,\tau_1,\tau_2,\tau_2,l)} \\
&\leq& \sum_{l=m+2}^n\sum_{\tau\in\mathcal{T}_{m,l-1}}(p(1-p)^4+(1-p)p^4)^3 + 6\sum_{l=m+2}^n\sum_{\tau_1\in\mathcal{T}_{m,l-1}}\sum_{\{\tau_2\in\mathcal{T}_{m,l-1}:\tau_2\cap\tau_1=\varnothing\}}p^6(1-p)^6 \\
&& + 6\sum_{l=m+2}^n\sum_{(i,j)\in\mathcal{T}_{m,l-1}}\sum_{\{i'\in[m]:i'\neq i,j\}}p^4(1-p)^4(p(1-p)^4+(1-p)p^4) \\
&& + 6\sum_{l=m+2}^n\sum_{(i,j)\in\mathcal{T}_{m,l-1}}\sum_{\{j'\in[l-1]:j'\neq i,j\}}p^4(1-p)^4(p(1-p)^4+(1-p)p^4) \\
&=& O\left(n^2mp^3 + n^3m^2p^6 + n^3mp^5\right).
\end{eqnarray*}
For the fourth term, similar to the previous calculation, we get
\begin{eqnarray*}
&& \sum_{l=m+2}^n\mathbb{E}\left(\sum_{1\leq i<j\leq m}Y_{(i,j,l)} \right)^4 \\
&\leq&  n\sum_{\tau\in\mathcal{T}_{m}}(p(1-p)^4+(1-p)p^4)^3 + 6n\sum_{\tau_1\in\mathcal{T}_{m}}\sum_{\{\tau_2\in\mathcal{T}_{m}:\tau_2\cap\tau_1=\varnothing\}}p^6(1-p)^6 \\
&& + 6n\sum_{(i,j)\in\mathcal{T}_{m}}\sum_{\{i'\in[m]:i'\neq i,j\}}p^4(1-p)^4(p(1-p)^4+(1-p)p^4) \\
&& + 6n\sum_{(i,j)\in\mathcal{T}_{m}}\sum_{\{j'\in[m]:j'\neq i,j\}}p^4(1-p)^4(p(1-p)^4+(1-p)p^4) \\
&=& O\left(nm^2p^3+nm^4p^6+nm^3p^5\right).
\end{eqnarray*}
Combining the bounds above, we get
$$\sum_{l=1}^n\mathbb{E}\left(\sum_{\lambda\in\Delta_{n,l}}Y_{\lambda} - \sum_{\lambda\in\Delta_{n,l-1}}Y_{\lambda}\right)^4=O(n^2mp^3 + n^3m^2p^6 + n^3mp^5),$$
which directly implies
$$\sum_{l=1}^n\mathbb{E}\left(\frac{\sum_{\lambda\in\Delta_{n,l}}Y_{\lambda}-\sum_{\lambda\in\Delta_{n,l-1}}Y_{\lambda}}{\frac{m(n-1)(n-2)}{2}\sqrt{\frac{p^3(1-p)^3}{f(n,m)}}}\right)^4=O\left(\frac{1}{p^3mn^2}+\frac{1}{n}+\frac{1}{pmn}\right).$$
Therefore, (\ref{eq:fourth-Y}) holds as long as $p^3mn^2\rightarrow\infty$.

To show (\ref{eq:fourth-W}), we decompose the left hand side of (\ref{eq:fourth-W}). By (\ref{eq:W-case}), we get
\begin{eqnarray*}
&& \sum_{l=1}^n\mathbb{E}\left(\sum_{\lambda\in\Delta_{n,l}}W_{\lambda} - \sum_{\lambda\in\Delta_{n,l-1}}W_{\lambda}\right)^4 \\
&=& \sum_{l=1}^m\mathbb{E}\left(3\sum_{1\leq i<j\leq l-1}W_{(i,j,l)}\right)^4 + \mathbb{E}\left(2\sum_{1\leq i<j\leq m}W_{(i,j,m+1)}\right)^4 \\
&& + \sum_{l=m+2}^n\mathbb{E}\left(\sum_{i=1}^m\sum_{j=m+1}^{l-1}W_{(i,j,l)} + 2\sum_{1\leq i<j\leq m}W_{(i,j,l)} \right)^4 \\
&\leq&  81\sum_{l=1}^m\mathbb{E}\left(\sum_{1\leq i<j\leq l-1}W_{(i,j,l)}\right)^4 + 16\mathbb{E}\left(\sum_{1\leq i<j\leq m}W_{(i,j,m+1)}\right)^4 \\
&& + 8\sum_{l=m+2}^n\mathbb{E}\left(\sum_{i=1}^m\sum_{j=m+1}^{l-1}W_{(i,j,l)}\right)^4 + 128\sum_{l=m+2}^n\mathbb{E}\left(\sum_{1\leq i<j\leq m}W_{(i,j,l)} \right)^4.
\end{eqnarray*}
We will give bounds to the four terms above, respectively. In order to deal with fourth moments, we need to study $N_{(\tau_1,\tau_2,\tau_3,\tau_4,l)}=\mathbb{E}W_{(\tau_1,l)}W_{(\tau_2,l)}W_{(\tau_3,l)}W_{(\tau_4,l)}$ for various situations of $\tau_1,\tau_2,\tau_3,\tau_4$. There are several different scenarios to consider. 
\begin{enumerate}
\item When $\tau_1=\tau_2=\tau_3=\tau_4$,
\begin{equation}
N_{(\tau_1,\tau_2,\tau_3,\tau_4,l)} \leq 54(p(1-p)^4+(1-p)p^4)^2.\label{eq:N1}
\end{equation}
\item When $\tau_1=\tau_2\neq \tau_3=\tau_4$,
\begin{equation}
N_{(\tau_1,\tau_2,\tau_3,\tau_4,l)} \leq \begin{cases}
9p^4(1-p)^4, & \tau_1\cap\tau_3=\varnothing, \\
10[p(1-p)^4+(1-p)p^4]p^2(1-p)^2, & |\tau_1\cap\tau_3|=1.
\end{cases}\label{eq:N2}
\end{equation}
\item When $\tau_1\neq\tau_2=\tau_3=\tau_4$,
\begin{equation}
N_{(\tau_1,\tau_2,\tau_3,\tau_4,l)} = 0.\label{eq:N3}
\end{equation}
\item When $\tau_1\neq\tau_2\neq\tau_3\neq \tau_1=\tau_4$,
\begin{equation}
N_{(\tau_1,\tau_2,\tau_3,\tau_4,l)} \leq\begin{cases}
p^3(1-p)^3(1+2p), & (\tau_1,\tau_2,\tau_3)\text{ form a triangle},\\
0, & \text{otherwise}.
\end{cases}\label{eq:N4}
\end{equation}
\item When $\tau_1$, $\tau_2$, $\tau_3$, $\tau_4$ are four different edges,
\begin{equation}
N_{(\tau_1,\tau_2,\tau_3,\tau_4,l)} = \begin{cases}
p^4(1-p)^4, & (\tau_1,\tau_2,\tau_3,\tau_4)\text{ form a square}, \\
0, & \text{otherwise}.
\end{cases}\label{eq:N5}
\end{equation}
\end{enumerate}
Equipped with the moments calculations (\ref{eq:N1})-(\ref{eq:N5}), we are ready to proceed. Recall the definitions of $\mathcal{T}_l$ and $\mathcal{T}_{m,l}$. For the first term,
\begin{eqnarray*}
&& \sum_{l=1}^m\mathbb{E}\left(\sum_{1\leq i<j\leq l-1}W_{(i,j,l)}\right)^4 \\
&=& \sum_{l=1}^m\sum_{\tau_1\in\mathcal{T}_{l-1}}\sum_{\tau_2\in\mathcal{T}_{l-1}}\sum_{\tau_3\in\mathcal{T}_{l-1}}\sum_{\tau_4\in\mathcal{T}_{l-1}}N_{(\tau_1,\tau_2,\tau_3,\tau_4,l)} \\
&\leq& \sum_{l=1}^m\sum_{\tau\in\mathcal{T}_{l-1}}N_{(\tau,\tau,\tau,\tau,l)} + {4\choose 2}\sum_{l=1}^m\sum_{\{\tau_1,\tau_2\in\mathcal{T}_{l-1}: \tau_1\neq\tau_2\}}N_{(\tau_1,\tau_1,\tau_2,\tau_2,l)} \\
&& + 24\sum_{l=1}^m\sum_{1\leq i<j<k\leq l-1}N_{((i,j),(j,k),(i,k),(i,j),l)} + 24\sum_{l=1}^m\sum_{1\leq i<j<k<h\leq l-1}N_{((i,j),(j,k),(k,h),(i,h),l)}\\
&\leq& \sum_{l=1}^m\sum_{\tau\in\mathcal{T}_{l-1}}54(p(1-p)^4+(1-p)p^4)^2  + 6\sum_{l=1}^m\sum_{\tau_1\in\mathcal{T}_{l-1}}\sum_{\{\tau_2\in\mathcal{T}_{l-1}:\tau_2\cap\tau_1=\varnothing\}}9p^4(1-p)^4 \\
&& + 6\sum_{l=1}^m\sum_{(i,j)\in\mathcal{T}_{l-1}}\sum_{\{i'\in[l-1]:i'\neq i,j\}}10[p(1-p)^4+(1-p)p^4]p^2(1-p)^2 \\
&& + 6\sum_{l=1}^m\sum_{(i,j)\in\mathcal{T}_{l-1}}\sum_{\{j'\in[l-1]:j'\neq i,j\}}10[p(1-p)^4+(1-p)p^4]p^2(1-p)^2\\
&& + 24\sum_{l=1}^m\sum_{1\leq i<j<k\leq l-1}p^3(1-p)^3(1+2p) + 24\sum_{l=1}^m\sum_{1\leq i<j<k<h\leq l-1}p^4(1-p)^4 \\
&=& O\left(m^3p^2 + m^5p^4 + m^4p^3\right).
\end{eqnarray*}
The second term is obviously of the smaller order of the first term. For the third term,
\begin{eqnarray*}
&& \sum_{l=m+2}^n\mathbb{E}\left(\sum_{i=1}^m\sum_{j=m+1}^{l-1}W_{(i,j,l)}\right)^4 \\
&\leq& \sum_{l=m+2}^n\sum_{\tau\in\mathcal{T}_{m,l-1}}N_{(\tau,\tau,\tau,\tau,l)} + {4\choose 2}\sum_{l=m+2}^n\sum_{\{\tau_1,\tau_2\in\mathcal{T}_{m,l-1}: \tau_1\neq\tau_2\}}N_{(\tau_1,\tau_1,\tau_2,\tau_2,l)} \\
&\leq& \sum_{l=m+2}^n\sum_{\tau\in\mathcal{T}_{m,l-1}}54(p(1-p)^4+(1-p)p^4)^2 + 6\sum_{l=m+2}^n\sum_{\tau_1\in\mathcal{T}_{m,l-1}}\sum_{\{\tau_2\in\mathcal{T}_{m,l-1}:\tau_2\cap\tau_1=\varnothing\}}9p^4(1-p)^4 \\
&& + 6\sum_{l=m+2}^n\sum_{(i,j)\in\mathcal{T}_{m,l-1}}\sum_{\{i'\in[m]:i'\neq i,j\}}10[p(1-p)^4+(1-p)p^4]p^2(1-p)^2 \\
&& + 6\sum_{l=m+2}^n\sum_{(i,j)\in\mathcal{T}_{m,l-1}}\sum_{\{j'\in[l-1]:j'\neq i,j\}}10[p(1-p)^4+(1-p)p^4]p^2(1-p)^2 \\
&=& O(n^2mp^2+n^3m^2p^4+n^3mp^3).
\end{eqnarray*}
For the fourth term, similar to the previous calculation, we get
\begin{eqnarray*}
&& \sum_{l=m+2}^n\mathbb{E}\left(\sum_{1\leq i<j\leq m}W_{(i,j,l)} \right)^4 \\
&\leq&  n\sum_{\tau\in\mathcal{T}_{m}}54(p(1-p)^4+(1-p)p^4)^2  + 6n\sum_{\tau_1\in\mathcal{T}_{m}}\sum_{\{\tau_2\in\mathcal{T}_{m}:\tau_2\cap\tau_1=\varnothing\}}9p^4(1-p)^4 \\
&& + 6n\sum_{(i,j)\in\mathcal{T}_{m}}\sum_{\{i'\in[m]:i'\neq i,j\}}10[p(1-p)^4+(1-p)p^4]p^2(1-p)^2 \\
&& + 6n\sum_{(i,j)\in\mathcal{T}_{m}}\sum_{\{j'\in[m]:j'\neq i,j\}}10[p(1-p)^4+(1-p)p^4]p^2(1-p)^2\\
&& + 24n\sum_{1\leq i<j<k\leq m}p^3(1-p)^3(1+2p) + 24n\sum_{1\leq i<j<k<h\leq m}p^4(1-p)^4 \\
&=& O\left(nm^2p^2+nm^4p^4+nm^3p^3\right).
\end{eqnarray*}
Combining the bounds above, we get
$$\sum_{l=1}^n\mathbb{E}\left(\sum_{\lambda\in\Delta_{n,l}}W_{\lambda} - \sum_{\lambda\in\Delta_{n,l-1}}W_{\lambda}\right)^4=O(n^2mp^2+n^3m^2p^4+n^3mp^3),$$
which directly implies
$$\sum_{l=1}^n\mathbb{E}\left(\frac{\sum_{\lambda\in\Delta_{n,l}}W_{\lambda}-\sum_{\lambda\in\Delta_{n,l-1}}W_{\lambda}}{\frac{m(n-1)(n-2)}{2}\sqrt{\frac{3p^2(1-p)^2}{f(n,m)}}}\right)^4=O\left(\frac{1}{p^2mn^2}+\frac{1}{n}+\frac{1}{pmn}\right).$$
Therefore, (\ref{eq:fourth-W}) also holds when $p^3mn^2\rightarrow\infty$. The conclusions of (\ref{eq:fourth-Y}) and (\ref{eq:fourth-W}) imply that Lindeberg's condition holds. Thus, we have derived the desired asymptotic distribution conditioning on $\mathcal{M}=[m]$ by applying Theorem \ref{thm:mds}. It is easy to see that the conditional limit holds for any realization of $\mathcal{M}$. Therefore, we have proved Theorem \ref{thm:main}. 
\end{proof}

Now we are ready to give proofs of Theorem \ref{thm:vanilla}, Corollary \ref{cor:chi}, Theorem \ref{thm:sample-node} and Theorem \ref{thm:sample-triple}.

\begin{proof}[Proof of Theorem \ref{thm:vanilla}]
When $m=n$, Theorem \ref{thm:main} reduces to
$$\begin{pmatrix}
\sqrt{\frac{{n\choose 3}}{p^3(1-p)^3}}R_3 \\
\sqrt{\frac{{n\choose 3}}{3p^2(1-p)^2}} R_2
\end{pmatrix}\leadsto N\left(\begin{pmatrix}
0 \\
0
\end{pmatrix},\begin{pmatrix}
1 & 0 \\
0 & 1
\end{pmatrix}\right),$$
when $(1-p)^{-1}=O(1)$ and $np\rightarrow\infty$. Note the representations (\ref{eq:T1-rep}), (\ref{eq:T2-rep}) that will be stated in the proof of Proposition \ref{prop:corr}. Then, (\ref{eq:T1-rep}), (\ref{eq:T2-rep}) and Slutsky's theorem imply the desired result.
\end{proof}

\begin{proof}[Proof of Corollary \ref{cor:chi}]
This is a direct implication of Theorem \ref{thm:vanilla} by applying Slutsky's theorem.
\end{proof}

\begin{proof}[Proof of Theorem \ref{thm:type-1}]
It is sufficient to prove $T^2=O_P(1)$ under the null distribution. By the definition of $T^2$ and the fact that $(\hat{p}-p)^2=O_P(pn^{-2})$, it suffices to prove $n^3T_3^2/p^3=O_P(1)$ and $n^3T_2^2/p^2=O_P(1)$ under the condition $n^2p\rightarrow\infty$, respectively. In view of the representations (\ref{eq:T1-rep}) and (\ref{eq:T2-rep}), we only need to show
\begin{equation}
\frac{n^3R_3^2}{p^3}=O_P(1),\quad \frac{n^3R_2^2}{p^2}=O_P(1),\quad\text{and}\quad \frac{n^3(\hat{p}-p)^4}{p^2}=O_P(1).\label{eq:easyV}
\end{equation}
Using (\ref{eq:cross-m-Y}) and (\ref{eq:cross-m-W}), it is easy to check that (\ref{eq:easyV}) holds.
\end{proof}

\begin{proof}[Proof of Theorem \ref{thm:sample-node}]
Similar to (\ref{eq:T1-rep}) and (\ref{eq:T2-rep}), we have
\begin{eqnarray*}
T_3^{\mathcal{M}} &=& -R_3^{\mathcal{M}}-pR_2^{\mathcal{M}}+(\hat{p}^{\mathcal{M}}-p)^2(2p+\hat{p}^{\mathcal{M}}), \\
T_2^{\mathcal{M}} &=& 3R_3^{\mathcal{M}}-(1-3p)R_2^{\mathcal{M}}+3(1-2p-\hat{p}^{\mathcal{M}})(\hat{p}^{\mathcal{M}}-p)^2.
\end{eqnarray*}
Since $\mathbb{E}(\hat{p}^{\mathcal{M}}-p)^2=O\left(\frac{p}{mn}\right)$, the third terms in the expansions of both $T_3^{\mathcal{M}}$ and $T_2^{\mathcal{M}}$ are negligible when $p^3mn^2\rightarrow\infty$ and $m\rightarrow\infty$. Therefore, Slutsky's theorem implies the desired result.
\end{proof}

\begin{proof}[Proof of Theorem \ref{thm:sample-triple}]
We first observe the decomposition
\begin{eqnarray*}
T_3^{\Delta} &=& -R_3^{\Delta}-pR_2^{\Delta}+(\hat{p}^{\Delta}-p)^2(2p+\hat{p}^{\Delta}), \\
T_2^{\Delta} &=& 3R_3^{\Delta}-(1-3p)R_2^{\Delta}+3(1-2p-\hat{p}^{\Delta})(\hat{p}^{\Delta}-p)^2.
\end{eqnarray*}
The variance has decomposition
\begin{eqnarray*}
\mathbb{E}(\hat{p}^{\Delta}-p)^2 &=& \mathbb{E}\Var(\hat{p}^{\Delta}|A)+\Var(\mathbb{E}(\hat{p}^{\Delta}|A)), \\
\Var(R_3^{\Delta}) &=& \mathbb{E}\Var(R_3^{\Delta}|A)+\Var(\mathbb{E}(R_3^{\Delta}|A)), \\
\Var(R_2^{\Delta}) &=& \mathbb{E}\Var(R_2^{\Delta}|A)+\Var(\mathbb{E}(R_2^{\Delta}|A)).
\end{eqnarray*}
Note that $\mathbb{E}(\hat{p}^{\Delta}|A)=\hat{p}$, $\mathbb{E}(R_3^{\Delta}|A)=R_3$ and $\mathbb{E}(R_2^{\Delta}|A)=R_2$. Thus,
$$\Var(\mathbb{E}(\hat{p}^{\Delta}|A))=O\left(\frac{p}{n^2}\right),\quad \Var(\mathbb{E}(R_3^{\Delta}|A))=\frac{p^3(1-p)^3}{{n\choose 3}},\quad \Var(\mathbb{E}(R_2^{\Delta}|A))=\frac{3p^2(1-p)^2}{{n\choose 3}}.$$
Define i.i.d. random variables $\lambda_1,...,\lambda_{|\Delta|}$, each of which takes an element of $\{(i,j,k):1\leq i<j<k\leq n\}$ with probability $1/{n\choose 3}$. Then, we can write
$$\hat{p}^{\Delta}=\frac{1}{|\Delta|}\sum_{h=1}^{|\Delta|}E_{\lambda_h},\quad R_3^{\Delta}=\frac{1}{|\Delta|}\sum_{h=1}^{|\Delta|}Y_{\lambda_h},\quad R_2^{\Delta}=\frac{1}{|\Delta|}\sum_{h=1}^{|\Delta|}W_{\lambda_h},$$
where $E_{\lambda}=(A_{ij}+A_{ik}+A_{jk})/3$ for $\lambda=(i,j,k)$, and the definitions of $Y_{\lambda}$ and $W_{\lambda}$ are given by (\ref{eq:def-Y}) and (\ref{eq:def-W}). We calculate expected the conditional variance.
\begin{eqnarray*}
\mathbb{E}\Var(E_{\lambda_h}|A) &=& \mathbb{E}\left(\mathbb{E}(E_{\lambda_h}^2|A) - \mathbb{E}(E_{\lambda_h}|A)^2\right) \\
&=& O(p),
\end{eqnarray*}
\begin{eqnarray*}
\mathbb{E}\Var(Y_{\lambda_h}|A) &=& \mathbb{E}\left(\mathbb{E}(Y_{\lambda_h}^2|A) - \mathbb{E}(Y_{\lambda_h}|A)^2\right) \\
&=& \mathbb{E}\left(\frac{1}{{n\choose 3}}\sum_{1\leq i<j<k\leq n}Y_{(i,j,k)}^2-\left(\frac{1}{{n\choose 3}}\sum_{1\leq i<j<k\leq n}Y_{(i,j,k)}\right)^2\right) \\
&=& p^3(1-p)^3\left(1-\frac{1}{{n\choose 3}}\right),
\end{eqnarray*}
where the last equality is due to (\ref{eq:cross-m-Y}). The term $\mathbb{E}\Var(W_{\lambda_h}|A)$ can be calculated in a similar way suing (\ref{eq:cross-m-W}).
\begin{eqnarray*}
\mathbb{E}\Var(W_{\lambda_h}|A) &=& \mathbb{E}\left(\mathbb{E}(W_{\lambda_h}^2|A) - \mathbb{E}(W_{\lambda_h}|A)^2\right) \\
&=& \mathbb{E}\left(\frac{1}{{n\choose 3}}\sum_{1\leq i<j<k\leq n}W_{(i,j,k)}^2-\left(\frac{1}{{n\choose 3}}\sum_{1\leq i<j<k\leq n}W_{(i,j,k)}\right)^2\right) \\
&=& 3p^2(1-p)^2\left(1-\frac{1}{{n\choose 3}}\right).
\end{eqnarray*}
By the variance decomposition above, we have
\begin{eqnarray*}
\mathbb{E}(\hat{p}^{\Delta}-p)^2 &=& O\left(\frac{p}{n^2}+\frac{p}{|\Delta|}\right), \\
\Var(R_3^{\Delta}) &=& (1+o(1))\left(\frac{1}{{n\choose 3}}+\frac{1}{|\Delta|}\right)p^3(1-p)^3, \\
\Var(R_2^{\Delta}) &=& (1+o(1))\left(\frac{1}{{n\choose 3}}+\frac{1}{|\Delta|}\right)3p^2(1-p)^2.
\end{eqnarray*}
Therefore, as long as $|\Delta|^2p\rightarrow\infty$, we have
\begin{eqnarray}
\label{eq:mercedes} \frac{T_3^{\Delta}}{-R_3^{\Delta}-pR_2^{\Delta}} &=& 1+o_P(1), \\
\label{eq:benz} \frac{T_2^{\Delta}}{3R_3^{\Delta}-(1-3p)R_2^{\Delta}} &=& 1+o_P(1).
\end{eqnarray}
Thus, it is sufficient to derive the joint asymptotic distribution of $(R_3^{\Delta},R_2^{\Delta})$. We first study the conditional distribution given $A$. Note that
$$\Var(Y_{\lambda_h}|A)=\frac{1}{{n\choose 3}}\sum_{1\leq i<j<k\leq n}Y_{(i,j,k)}^2-\left(\frac{1}{{n\choose 3}}\sum_{1\leq i<j<k\leq n}Y_{(i,j,k)}\right)^2.$$
We need to show that $\Var(Y_{\lambda_h}|A)$ concentrates on its expected value. We will use the following property
\begin{equation}
\mathbb{E}(Y_{\lambda}^2-p^3(1-p)^3)(Y_{\lambda'}^2-p^3(1-p)^3) = \begin{cases}
O(p^3), & \lambda=\lambda', \\
O(p^5), & |\lambda\cap\lambda'|=1, \\
0, & \text{otherwise}.
\end{cases}\label{eq:4cross-Y}
\end{equation}
Note
$$\mathbb{E}\left(\frac{1}{{n\choose 3}}\sum_{1\leq i<j<k\leq n}Y_{(i,j,k)}\right)^2=O\left(\frac{p^3}{n^3}\right),$$
and
\begin{eqnarray*}
&& \Var\left(\frac{1}{{n\choose 3}}\sum_{1\leq i<j<k\leq n}Y_{(i,j,k)}^2\right) \\
&=& \frac{1}{{n\choose 3}^2}\sum_{1\leq i<j<k\leq n}\sum_{1\leq i'<j'<k'\leq n}\mathbb{E}\left(Y_{(i,j,k)}^2-p^3(1-p)^3\right)\left(Y_{(i',j',k')}^2-p^3(1-p)^3\right) \\
&\leq& \frac{1}{{n\choose 3}^2}\sum_{1\leq i<j<k\leq n}\mathbb{E}\left(Y_{(i,j,k)}^2-p^3(1-p)^3\right)^2 \\
&& + \frac{2}{{n\choose 3}^2}\sum_{1\leq i<j<k\leq n}\sum_{\{1\in[n]:l\neq i,j,k\}}\mathbb{E}\left(Y_{(i,j,k)}^2-p^3(1-p)^3\right)\left(Y_{(i,j,l)}^2-p^3(1-p)^3\right) \\
&& + \frac{2}{{n\choose 3}^2}\sum_{1\leq i<j<k\leq n}\sum_{\{1\in[n]:l\neq i,j,k\}}\mathbb{E}\left(Y_{(i,j,k)}^2-p^3(1-p)^3\right)\left(Y_{(l,j,k)}^2-p^3(1-p)^3\right) \\
&& + \frac{2}{{n\choose 3}^2}\sum_{1\leq i<j<k\leq n}\sum_{\{1\in[n]:l\neq i,j,k\}}\mathbb{E}\left(Y_{(i,j,k)}^2-p^3(1-p)^3\right)\left(Y_{(i,l,k)}^2-p^3(1-p)^3\right) \\
&=& O\left(\frac{p^3}{n^3}+\frac{p^5}{n^2}\right),
\end{eqnarray*}
where the last equality is due to (\ref{eq:4cross-Y}). Compared with the order of $\mathbb{E}\Var(Y_{\lambda_h}|A)$, when $np\rightarrow\infty$, we have
$$\Var(Y_{\lambda_h}|A)=(1+o_P(1))p^3(1-p)^3.$$
Now we study
$$\Var(W_{\lambda_h}|A)=\frac{1}{{n\choose 3}}\sum_{1\leq i<j<k\leq n}W_{(i,j,k)}^2-\left(\frac{1}{{n\choose 3}}\sum_{1\leq i<j<k\leq n}W_{(i,j,k)}\right)^2.$$
We have the following property
\begin{equation}
\mathbb{E}(W_{\lambda}^2-3p^2(1-p)^2)(W_{\lambda'}^2-3p^2(1-p)^2) = \begin{cases}
O(p^2), & \lambda=\lambda', \\
O(p^3), & |\lambda\cap\lambda'|=1, \\
0, & \text{otherwise}.
\end{cases}\label{eq:4cross-W}
\end{equation}
Since
$$\mathbb{E}\left(\frac{1}{{n\choose 3}}\sum_{1\leq i<j<k\leq n}W_{(i,j,k)}\right)^2=O\left(\frac{p^2}{n^3}\right),$$
and
\begin{eqnarray*}
&& \Var\left(\frac{1}{{n\choose 3}}\sum_{1\leq i<j<k\leq n}W_{(i,j,k)}^2\right) \\
&=& \frac{1}{{n\choose 3}^2}\sum_{1\leq i<j<k\leq n}\sum_{1\leq i'<j'<k'\leq n}\mathbb{E}\left(W_{(i,j,k)}^2-3p^2(1-p)^2\right)\left(W_{(i',j',k')}^2-3p^2(1-p)^2\right) \\
&\leq& \frac{1}{{n\choose 3}^2}\sum_{1\leq i<j<k\leq n}\mathbb{E}\left(W_{(i,j,k)}^2-3p^2(1-p)^2\right)^2 \\
&& + \frac{2}{{n\choose 3}^2}\sum_{1\leq i<j<k\leq n}\sum_{\{1\in[n]:l\neq i,j,k\}}\mathbb{E}\left(W_{(i,j,k)}^2-3p^2(1-p)^2\right)\left(W_{(i,j,l)}^2-3p^2(1-p)^2\right) \\
&& + \frac{2}{{n\choose 3}^2}\sum_{1\leq i<j<k\leq n}\sum_{\{1\in[n]:l\neq i,j,k\}}\mathbb{E}\left(W_{(i,j,k)}^2-3p^2(1-p)^2\right)\left(W_{(l,j,k)}^2-3p^2(1-p)^2\right) \\
&& + \frac{2}{{n\choose 3}^2}\sum_{1\leq i<j<k\leq n}\sum_{\{1\in[n]:l\neq i,j,k\}}\mathbb{E}\left(W_{(i,j,k)}^2-3p^2(1-p)^2\right)\left(W_{(i,l,k)}^2-3p^2(1-p)^2\right) \\
&=& O\left(\frac{p^2}{n^3}+\frac{p^3}{n^2}\right),
\end{eqnarray*}
where the last equality is due to (\ref{eq:4cross-W}). When $np\rightarrow\infty$, we have
$$\Var(W_{\lambda_h}|A)=(1+o_P(1))3p^2(1-p)^2.$$
We also need to study the conditional covariance between $R_3^{\Delta}$ and $R_2^{\Delta}$. We need the following property
\begin{equation}
\mathbb{E}Y_{\lambda}W_{\lambda}Y_{\lambda'}W_{\lambda'} = \begin{cases}
O(p^3), & \lambda=\lambda',\\
O(p^5), & |\lambda\cap\lambda'|=1, \\
0, & \text{otherwise}.
\end{cases}\label{eq:cross-cross}
\end{equation}
Note that
\begin{eqnarray*}
&& \Cov(Y_{\lambda_h},W_{\lambda_h}|A) \\
&=& \mathbb{E}(Y_{\lambda_h}W_{\lambda_h}|A)-\mathbb{E}(Y_{\lambda_h}|A)\mathbb{E}(W_{\lambda_h}|A) \\
&=& \frac{1}{{n\choose 3}}\sum_{1\leq i<j<k\leq n}Y_{(i,j,k)}W_{(i,j,k)}-\left(\frac{1}{{n\choose 3}}\sum_{1\leq i<j<k\leq n}Y_{(i,j,k)}\right)\left(\frac{1}{{n\choose 3}}\sum_{1\leq i<j<k\leq n}W_{(i,j,k)}\right).
\end{eqnarray*}
Thus, $\mathbb{E}\Cov(Y_{\lambda_h},W_{\lambda_h}|A)=0$ in view of (\ref{eq:inde-shapes}). Since
$$\left(\frac{1}{{n\choose 3}}\sum_{1\leq i<j<k\leq n}Y_{(i,j,k)}\right)\left(\frac{1}{{n\choose 3}}\sum_{1\leq i<j<k\leq n}W_{(i,j,k)}\right)=O_P\left(\frac{p^{5/2}}{n^3}\right),$$
and
\begin{eqnarray*}
&& \Var\left(\frac{1}{{n\choose 3}}\sum_{1\leq i<j<k\leq n}Y_{(i,j,k)}W_{(i,j,k)}\right) \\
&=& \frac{1}{{n\choose 3}^2}\sum_{1\leq i<j<k\leq n}\sum_{1\leq i'<j'<k'\leq n}\mathbb{E}Y_{(i,j,k)}W_{(i,j,k)}Y_{(i',j',k')}W_{(i',j',k')} \\
&\leq& \frac{1}{{n\choose 3}^2}\sum_{1\leq i<j<k\leq n}\mathbb{E}Y_{(i,j,k)}^2W_{(i,j,k)}^2 \\
&& + \frac{2}{{n\choose 3}^2}\sum_{1\leq i<j<k\leq n}\sum_{\{1\in[n]:l\neq i,j,k\}}\mathbb{E}Y_{(i,j,k)}W_{(i,j,k)}Y_{(i,j,l)}W_{(i,j,l)} \\
&& + \frac{2}{{n\choose 3}^2}\sum_{1\leq i<j<k\leq n}\sum_{\{1\in[n]:l\neq i,j,k\}}\mathbb{E}Y_{(i,j,k)}W_{(i,j,k)}Y_{(i,l,k)}W_{(i,l,k)} \\
&& + \frac{2}{{n\choose 3}^2}\sum_{1\leq i<j<k\leq n}\sum_{\{1\in[n]:l\neq i,j,k\}}\mathbb{E}Y_{(i,j,k)}W_{(i,j,k)}Y_{(l,j,k)}W_{(l,j,k)} \\
&=& O\left(\frac{p^3}{n^3}+\frac{p^5}{n^2}\right).
\end{eqnarray*}
Therefore, $\Cov(Y_{\lambda_h},W_{\lambda_h}|A)=O_P\left(\sqrt{\frac{p^3}{n^3}+\frac{p^5}{n^2}}\right)$. We are ready to show that
\begin{equation}
t_1 \sqrt{\frac{{|\Delta|}}{p^3(1-p)^3}}(R_3^{\Delta}-R_3) + t_2 \sqrt{\frac{{|\Delta|}}{3p^2(1-p)^2}} (R_2^{\Delta}-R_2)\Bigg|A \leadsto N(0,1), \label{eq:to-prove-Delta}
\end{equation}
for any fixed $t_1$ and $t_2$ that satisfy $t_1^2+t_2^2=1$. With the previous calculations for $\Var(Y_{\lambda_h}|A)$, $\Var(W_{\lambda_h}|A)$ and $\Cov(Y_{\lambda_h},W_{\lambda_h}|A)$, we have
$$\Var\left(t_1 \sqrt{\frac{{|\Delta|}}{p^3(1-p)^3}}R_3^{\Delta} + t_2 \sqrt{\frac{{|\Delta|}}{3p^2(1-p)^2}} R_2^{\Delta}\Bigg|A\right)=1+o_P(1),$$
as long as $np\rightarrow\infty$. Therefore, in order that (\ref{eq:to-prove-Delta}) holds, it is sufficient to establish Lyapunov's condition,
\begin{eqnarray*}
&& \frac{1}{|\Delta|^2}\sum_{h=1}^{|\Delta|}\mathbb{E}\left(\frac{Y_{\lambda_h}}{p^{3/2}}+\frac{W_{\lambda_h}}{p}\right)^4 \\
&\leq& \frac{8\mathbb{E}|Y_{\lambda_h}|^4}{p^6|\Delta|} + \frac{8|W_{\lambda_h}|^4}{p^4|\Delta|} \\
&=& O\left(\frac{1}{|\Delta|p^3} + \frac{1}{|\Delta|p^2}\right) \\
&=& o(1),
\end{eqnarray*}
as long as $|\Delta|p^3\rightarrow\infty$. Therefore, (\ref{eq:to-prove-Delta}) holds under the conditions $(n^3+|\Delta|)p^3\rightarrow\infty$. Combining with the joint asymptotic distribution of $(R_3,R_2)$, we obtain the unconditional asymptotic distribution of $(R_3^{\Delta},R_2^{\Delta})$. The final result is a consequence of (\ref{eq:mercedes}) and (\ref{eq:benz}).
\end{proof}

\subsection{Proofs of power analysis}\label{sec:pf-power}

This section gives proofs of Theorem \ref{thm:SBM}, Theorem \ref{thm:SBM-k}, Theorem \ref{thm:configuration}, Theorem \ref{thm:latent}, Theorem \ref{thm:error-M} and Theorem \ref{thm:error-Delta}. We first state an auxiliary lemma. Recall the definition of $\hat{p}$ in (\ref{eq:p-hat}). We also define 
$$\bar{p}=\frac{1}{{n\choose 2}}\sum_{1\leq i<j\leq n}p_{ij},\quad \hat{q}=\frac{1}{{n\choose 3}}\sum_{1\leq i<j<k\leq n}A_{ij}A_{ik}A_{jk},\quad\bar{q}=\frac{1}{{n\choose 3}}\sum_{1\leq i<j<k\leq n}p_{ij}p_{ik}p_{jk}.$$
Consider the model where $A_{ij}\sim \text{Bernoulli}(p_{ij})$ independently for all $1\leq i<j\leq n$. The following lemma establishes error bounds for $\hat{p}^3$ and $\hat{q}$.

\begin{lemma}\label{lem:var-p-q}
We have
$$|\hat{p}^3-\bar{p}^3|=O_P\left(\frac{\rho^{5/2}}{n}+\frac{\rho^2}{n^2}+\frac{\rho^{3/2}}{n^3}\right),$$
and
$$|\hat{q}-\bar{q}|=O_P\left(\frac{\rho^{5/2}}{n}+\frac{\rho^{3/2}}{n^{3/2}}\right),$$
where $\rho=\max_{1\leq i<j\leq n}p_{ij}$.
\end{lemma}

Lemma \ref{lem:var-p-q} will be proved in Section \ref{sec:pf-tech}.
Now we are ready to state the proofs of Theorem \ref{thm:SBM}, Theorem \ref{thm:SBM-k} and Theorem \ref{thm:configuration}.

\begin{proof}[Proof of Theorem \ref{thm:SBM}]
In order to show that power converges to $1$, it is sufficient to prove
\begin{equation}
\frac{n^3T_3^2}{(a+b)^3}\rightarrow\infty,\label{eq:suff-SBM}
\end{equation}
in probability under the alternative distribution. This is because
$$\mathbb{P}(T^2> C_{\alpha})\geq\mathbb{P}\left({n\choose 3}\frac{T_3^2}{\hat{p}^3(1-\hat{p})^3+3\hat{p}^4(1-\hat{p})^2}>C_{\alpha}\right),$$
which converges to $1$ if (\ref{eq:suff-SBM}) holds when $n^2(a+b)\rightarrow\infty$.
Note that
$$T_3^2\geq \frac{1}{2}(\bar{p}^3-\bar{q})^2-2(\hat{p}^3-\bar{p}^3)^2-2(\hat{q}-\bar{q})^2.$$
The orders of $(\hat{p}^3-\bar{p}^3)^2$ and $(\hat{q}-\bar{q})^2$ are controlled by $O_P\left(\frac{(a+b)^5}{n^2}+\frac{(a+b)^3}{n^3}\right)$ according to Lemma \ref{lem:var-p-q}. Now we study $(\bar{p}^3-\bar{q})^2$. Let $n_1$ and $n_2$ be the sizes of the two communities, respectively. Then, define $\gamma=n_1/n$ and $1-\gamma=n_2/n$. Define
$$\bar{q}(a,b)=[\gamma^3+(1-\gamma)^3]a^3+[3\gamma^2(1-\gamma)+3(1-\gamma)^2\gamma]ab^2,$$
and
$$\bar{p}(a,b)=a+2\gamma(1-\gamma)(b-a).$$
Then, we have the bounds
$$|\bar{p}(a,b)^3-\bar{p}^3|=O\left(\frac{(a+b)^3}{n}\right),\quad\text{and}\quad |\bar{q}(a,b)-\bar{q}|=O\left(\frac{(a+b)^3}{n}\right).$$
Finally, we study the difference $\bar{q}(a,b)-\bar{p}(a,b)^3$. Define $\delta(a,b)=\bar{q}(a,b)-\bar{p}^3(a,b)$.
Similarly, we can define $\bar{q}(b,a)$, $\bar{p}^3(b,a)$ and $\delta(b,a)$. Note that
$$\bar{q}(a,b)+\bar{q}(b,a)=[b+\gamma(a-b)]^3+[a+\gamma(b-a)]^3,$$
$$\bar{q}(a,b)-\bar{q}(b,a)=[\gamma(a+b)-b]^3+[a-\gamma(a+b)]^3,$$
$$\bar{p}^3(a,b)+\bar{p}^3(b,a)=[a+2\gamma(1-\gamma)(b-a)]^3+[b+2\gamma(1-\gamma)(a-b)]^3,$$
$$\bar{p}^3(a,b)-\bar{p}^3(b,a)=[a+2\gamma(1-\gamma)(b-a)]^3-[b+2\gamma(1-\gamma)(a-b)]^3.$$
Therefore,
\begin{eqnarray*}
&& \delta(a,b) + \delta(b,a) \\
&=& \bar{q}(a,b)+\bar{q}(b,a) - [\bar{p}^3(a,b)+\bar{p}^3(b,a)] \\
&=&  \gamma^3[1-8(1-\gamma)^3](b-a)^3 + 3\gamma[1-2(1-\gamma)]a^2(b-a) + 3\gamma^2[1-4(1-\gamma)^2]a(b-a)^2 \\
&& + \gamma^3[1-8(1-\gamma)^3](a-b)^3 + 3\gamma[1-2(1-\gamma)]b^2(a-b) + 3\gamma^2[1-4(1-\gamma)^2]b(a-b)^2 \\
&=& 3\gamma(1-\gamma)[1-4\gamma(1-\gamma)](a+b)(a-b)^2,
\end{eqnarray*}
and
\begin{eqnarray*}
&& \delta(a,b) - \delta(b,a) \\
&=& \bar{q}(a,b)-\bar{q}(b,a) - [\bar{p}^3(a,b)-\bar{p}^3(b,a)] \\
&=& [3\gamma(1-\gamma)-12\gamma^2(1-\gamma)^2+16\gamma^3(1-\gamma)^3](a-b)^3.
\end{eqnarray*}
This gives the equation
$$\delta(a,b)=\frac{3\gamma(1-\gamma)[1-4\gamma(1-\gamma)](a+b)(a-b)^2}{2}+\frac{[3\gamma(1-\gamma)-12\gamma^2(1-\gamma)^2+16\gamma^3(1-\gamma)^3](a-b)^3}{2}.$$
Therefore, under the assumption that $a+b=O(n^{-1/2})$, a sufficient condition for (\ref{eq:suff-SBM}) is $\frac{n(a-b)^2}{a+b}\rightarrow\infty$.
\end{proof}

\begin{proof}[Proof of Theorem \ref{thm:SBM-k}]
The proof is similar to that of Theorem \ref{thm:SBM}. The only difference is the analysis of $(\bar{p}^3-\bar{q})^2$. Define
$$\bar{p}^*=\frac{1}{k}a+\frac{k-1}{k}b,\quad \bar{q}^*=\frac{1}{k^2}a^3+\frac{3(k-1)}{k^2}ab^2+\frac{(k-1)(k-2)}{k^2}b^3.$$
Then, it is easy to see that
$$\left|(\bar{p}^*)^3-\bar{p}^3\right|=O\left(\frac{(a+b)^3}{n}\right),\quad |\bar{q}^*-\bar{q}|=O\left(\frac{(a+b)^3}{n}\right).$$
Since
$$\bar{q}^*-(\bar{p}^*)^3=\frac{k-1}{k^3}(a-b)^3.$$
A sufficient condition is $a+b=O(n^{-1/2})$ and $\frac{n(a-b)^2}{k^{4/3}(a+b)}\rightarrow\infty$.
\end{proof}

\begin{proof}[Proof of Theorem \ref{thm:configuration}]
Similar to the proof of Theorem \ref{thm:SBM}, we need to show
$\frac{n^3T_3^2}{\hat{p}^3}\rightarrow\infty$ in probability under the alternative distribution.
Under the assumption $n^2\rho\rightarrow\infty$, it is sufficient to show $\frac{n^3T_3^2}{\rho^3}\rightarrow\infty$. Note that
$$T_3^2\geq \frac{1}{2}(\bar{p}^3-\bar{q})^2-2(\hat{p}^3-\bar{p}^3)^2-2(\hat{q}-\bar{q})^2.$$
The orders of $(\hat{p}^3-\bar{p}^3)^2$ and $(\hat{q}-\bar{q})^2$ are controlled by $O_P\left(\frac{\rho^5}{n^2}+\frac{\rho^3}{n^3}\right)$ according to Lemma \ref{lem:var-p-q}. Now we study $(\bar{p}^3-\bar{q})^2$. Note that
$$\left|\bar{p}^3-\left(\frac{1}{n}\sum_{i=1}^n\theta_i\right)^6\right|=O\left(\frac{\rho^3}{n}\right),\quad\text{and}\quad\left|\bar{q}-\left(\frac{1}{n}\sum_{i=1}^n\theta_i^2\right)^3\right|=O\left(\frac{\rho^3}{n}\right).$$
Hence, a sufficient condition is $\frac{\delta^2}{\frac{\rho^3}{n^3}+\frac{\rho^5}{n^2}}\rightarrow\infty$.
\end{proof}

Before proving Theorem \ref{thm:latent}, we need another lemma that bounds the error of $\hat{p}^3$ and $\hat{q}$ under the graphon model. Define
$$\tilde{p}=\frac{1}{{n\choose 2}}\sum_{1\leq i<j\leq n}\mathbb{E}\left[g(\xi_i)^Tg(\xi_j)\right],\quad \tilde{q}=\frac{1}{{n\choose 3}}\sum_{1\leq i<j<k\leq n}\mathbb{E}\left[g(\xi_i)^Tg(\xi_j)g(\xi_i)^Tg(\xi_k)g(\xi_j)^Tg(\xi_k)\right],$$
where $g(\xi_i)$ is understood as the vector $(g_1(\xi_{i1}),...,g_r(\xi_{ir}))^T$.

\begin{lemma}\label{lem:var-p-q-latent}
We have
$$|\hat{p}^3-\tilde{p}^3|=O_P\left(\frac{\|g\|_{\infty}^6}{\sqrt{n}}+\frac{\|g\|_{\infty}^5}{n}+\frac{\|g\|_{\infty}^4}{n^2}+\frac{\|g\|_{\infty}^3}{n^3}\right),$$
and
$$|\hat{q}-\tilde{q}|=O_P\left(\frac{\|g\|_{\infty}^{6}}{\sqrt{n}}+\frac{\|g\|_{\infty}^5}{n}+\frac{\|g\|_{\infty}^3}{n^{3/2}}\right),$$
where $\|g\|_{\infty}=\sqrt{\sum_{l=1}^r\|g_l\|_{\infty}^2}$.
\end{lemma}

The proof of Lemma \ref{lem:var-p-q-latent} will be given in Section \ref{sec:pf-tech}. Now we are ready to state the proof of Theorem \ref{thm:latent}.

\begin{proof}[Proof of Theorem \ref{thm:latent}]
Again, under the condition $n^2\rho\rightarrow\infty$, it is sufficient to show that $n^3T_3^2/\rho^3\rightarrow\infty$. Note that
$$T_3^2\geq \frac{1}{2}(\bar{p}^3-\bar{q})^2-2(\hat{p}^3-\bar{p}^3)^2-2(\hat{q}-\bar{q})^2.$$
The orders of $(\hat{p}^3-\tilde{p}^3)^2$ and $(\hat{q}-\tilde{q})^2$ are controlled by $O_P\left(\frac{\rho^6}{n}+\frac{\rho^5}{n^2}+\frac{\rho^3}{n^3}\right)$ according to Lemma \ref{lem:var-p-q-latent}. Now we study $(\tilde{p}^3-\tilde{q})^2$. Note that
$$\left|\tilde{p}^3-\left\|\mathbb{E}g(\xi)\right\|^6\right|=O\left(\frac{\rho^3}{n}\right),\quad\text{and}\quad\left|\tilde{q}-\Tr\left[\left(\mathbb{E}g(\xi)g(\xi)^T\right)^3\right]\right|=O\left(\frac{\rho^3}{n}\right).$$
Hence, a sufficient condition is $\frac{\delta^2}{\frac{\rho^6}{n}+\frac{\rho^5}{n^2}+\frac{\rho^3}{n^3}}\rightarrow\infty$.
\end{proof}

The proof of Theorem \ref{thm:error-M} requires the following two lemmas.
\begin{lemma}\label{lem:var-p-q-M}
We have
$$|(\hat{p}^{\mathcal{M}})^3-\bar{p}^3|=O_P\left(\frac{\rho^{3/2}}{(mn)^{3/2}}+\frac{\rho^2}{mn}+\frac{\rho^{5/2}}{(mn)^{1/2}}+\frac{\rho^3}{m^{1/2}}\right),$$
and
$$|\hat{q}^{\mathcal{M}}-\bar{q}|=O_P\left(\frac{\rho^3}{m^{1/2}}+\frac{\rho^{5/2}}{(mn)^{1/2}}+\frac{\rho^{3/2}}{m^{1/2}n}\right),$$
where $\rho=\max_{1\leq i<j\leq n}p_{ij}$.
\end{lemma}

\begin{lemma}\label{lem:var-p-q-latent-M}
We have
$$|(\hat{p}^{\mathcal{M}})^3-\tilde{p}^3|=O_P\left(\frac{\|g\|_{\infty}^{3}}{(mn)^{3/2}}+\frac{\|g\|_{\infty}^4}{mn}+\frac{\|g\|_{\infty}^{5}}{(mn)^{1/2}}+\frac{\|g\|_{\infty}^6}{m^{1/2}}\right),$$
and
$$|\hat{q}^{\mathcal{M}}-\tilde{q}|=O_P\left(\frac{\|g\|_{\infty}^6}{m^{1/2}}+\frac{\|g\|_{\infty}^{5}}{(mn)^{1/2}}+\frac{\|g\|_{\infty}^{3}}{m^{1/2}n}\right),$$
where $\|g\|_{\infty}=\sqrt{\sum_{l=1}^r\|g_l\|_{\infty}^2}$.
\end{lemma}
The proofs of the two lemmas will be proved in Section \ref{sec:pf-tech}.

\begin{proof}[Proof of Theorem \ref{thm:error-M}]
The consistency of the Type-I error is an immediate consequence of Corollary \ref{cor:chi-node}. To see the power under a stochastic block model or a configuration model, we use the same argument in the proofs of Theorem \ref{thm:SBM}, Theorem \ref{thm:SBM-k} and Theorem \ref{thm:configuration}. The only difference is that by Lemma \ref{lem:var-p-q-M}, the error $|(\hat{p}^{\mathcal{M}})^3-\bar{p}^3|^2+|\hat{q}^{\mathcal{M}}-\bar{q}|^2$ is bounded by
$O\left(\frac{\rho^6}{m}+{\frac{\rho^3}{mn^2}}\right)$ under the assumption that $n\rho\rightarrow\infty$.
Depending on whether $\rho=O(n^{-2/3})$ or not, either ${\frac{\rho^3}{mn^2}}$ or $\frac{\rho^6}{m}$ dominates the error. For the power under a low-rank latent variable model, we use the same argument in the proof of Theorem \ref{thm:latent}. Note that when $n\rho\rightarrow\infty$, the error $|(\hat{p}^{\mathcal{M}})^3-\tilde{p}^3|^2+|\hat{q}^{\mathcal{M}}-\tilde{q}|^2$ is bounded by $O\left(\frac{\rho^6}{m}+{\frac{\rho^3}{mn^2}}\right)$, with $\|g\|_{\infty}^2\leq\rho$. This completes the proof.
\end{proof}

The proof of of Theorem \ref{thm:error-Delta} requires the following two lemmas.
\begin{lemma}\label{lem:var-p-q-Delta}
We have
$$|(\hat{p}^{\Delta})^3-\bar{p}^3|=O_P\left(\frac{\rho^{5/2}}{n}+\frac{\rho^2}{n^2}+\frac{\rho^{3/2}}{n^3}+\frac{\rho^{5/2}}{|\Delta|^{1/2}}+\frac{\rho^2}{|\Delta|}+\frac{\rho^{3/2}}{|\Delta|^{3/2}}\right),$$
and
$$|\hat{q}^{\Delta}-\bar{q}|=O_P\left(\frac{\rho^{5/2}}{n}+\frac{\rho^{3/2}}{n^{3/2}}+\frac{\rho^{3/2}}{|\Delta|^{1/2}}\right),$$
where $\rho=\max_{1\leq i<j\leq n}p_{ij}$.
\end{lemma}

\begin{lemma}\label{lem:var-p-q-latent-Delta}
We have
$$|\hat{p}^{\Delta}-\tilde{p}^3|=O_P\left(\frac{\|g\|_{\infty}^6}{\sqrt{n}}+\frac{\|g\|_{\infty}^5}{n}+\frac{\|g\|_{\infty}^4}{n^2}+\frac{\|g\|_{\infty}^3}{n^3}+\frac{\|g\|_{\infty}^{5}}{|\Delta|^{1/2}}+\frac{\|g\|_{\infty}^4}{|\Delta|}+\frac{\|g\|_{\infty}^{3}}{|\Delta|^{3/2}}\right),$$
$$|\hat{q}^{\Delta}-\tilde{q}|=O_P\left(\frac{\|g\|_{\infty}^{6}}{\sqrt{n}}+\frac{\|g\|_{\infty}^5}{n}+\frac{\|g\|_{\infty}^3}{n^{3/2}}+\frac{\|g\|_{\infty}^3}{|\Delta|^{1/2}}\right),$$
where $\|g\|_{\infty}=\sqrt{\sum_{l=1}^r\|g_l\|_{\infty}^2}$.
\end{lemma}
The proofs of the two lemmas will be proved in Section \ref{sec:pf-tech}.

\begin{proof}[Proof of Theorem \ref{thm:error-Delta}]
The consistency of the Type-I error is an immediate consequence of Corollary \ref{cor:chi-triple}. To see the power under a stochastic block model or a configuration model, we use the same argument in the proofs of Theorem \ref{thm:SBM}, Theorem \ref{thm:SBM-k} and Theorem \ref{thm:configuration}. The only difference is that by Lemma \ref{lem:var-p-q-M}, the error $|(\hat{p}^{\Delta})^3-\bar{p}^3|^2+|\hat{q}^{\Delta}-\bar{q}|^2$ is bounded by
$O\left(\frac{\rho^5}{n^2}+\frac{\rho^3}{|\Delta|}\right)$ under the assumption that $|\Delta|=o(n^3)$.
This leads to the desired conclusion of power for the first three cases. For the power under a low-rank latent variable model, we use the same argument in the proof of Theorem \ref{thm:latent}. Note that when $|\Delta|=o(n^3)$, the error $|(\hat{p}^{\Delta})^3-\tilde{p}^3|^2+|\hat{q}^{\Delta}-\tilde{q}|^2$ is bounded by $O\left(\frac{\rho^6}{n}+\frac{\rho^5}{n^2}+\frac{\rho^3}{|\Delta|}\right)$, with $\|g\|_{\infty}^2\leq\rho$. This completes the proof.
\end{proof}

\subsection{Proof of the lower bound}\label{sec:pf-lb}

\begin{proof}[Proof of Theorem \ref{thm:SBM-lower}]
Given $a$ and $b$, let $P$ be the probability distribution induced by the Erd\H{o}s-R\'{e}nyi graph with edge probability $\frac{a+b}{2}$. Let $Q$ be the probability distribution of an SBM with uniformly assigned community labels. That is, for each $i\in[n]$, $z(i)=1$ or $z(i)=-1$ with probability $1/2$ independently. Then, for each $i<j$, $Q(A_{ij}=1|z(i)=z(j))=a$ and $Q(A_{ij}=1|z(i)\neq z(j))=b$ independently. For any testing function $\phi$, we use Cauchy-Schwarz inequality to derive
$$Q\phi\leq \sqrt{\int \left(\frac{dQ}{dP}\right)^2dP}\sqrt{P\phi}.$$
The chi-squared divergence is bounded by the following result.
\begin{proposition}\label{prop:affinity}
When $a+b=o(1)$ and $\frac{n(a-b)^2}{2(a+b)}<1-c$ for some arbitrarily small constant $c\in(0,1)$, there exists a constant $C>0$ such that $\sqrt{\int \left(\frac{dQ}{dP}\right)^2dP}<C$.
\end{proposition}
The proof of this proposition is given in the end of the section.
This conclusion implies that $Q\phi\leq C\sqrt{P\phi}$. We need to give a bound for $\inf_{\mathcal{S}(\pi,a,b)}\mathbb{P}\phi$. Note that the distribution $Q$ has decomposition $Q=\sum_z w(z)Q_z$, where $Q_z$ is the conditional distribution of $\{A_{ij}\}_{1\leq i<j\leq n}$ given the community label $z$, and $w$ is the uniform distribution on $z$. Define
$$\tilde{w}(z)=\frac{w(z)\mathbb{I}_{z\in\mathcal{Z}_{n,2}(\pi)}}{\sum_{z\in\mathcal{Z}_{n,2}(\pi)}w(z)},\quad\text{and}\quad \tilde{Q}=\sum_z\tilde{w}(z)Q_z.$$
By data processing inequality,
\begin{eqnarray*}
\TV(Q,\tilde{Q}) &\leq& \TV(w,\tilde{w}) \\
&\leq& 2\sum_{z\notin\mathcal{Z}_{n,2}(\pi)}w(z) \\
&\leq& 4\exp\left(-2n(\pi-1/2)^2\right),
\end{eqnarray*}
where the last inequality is Hoeffding's inequality. Finally,
$$\inf_{\mathcal{S}(\pi,a,b)}\mathbb{P}\phi\leq \tilde{Q}\phi\leq Q\phi + \TV(Q,\tilde{Q})\leq C\sqrt{P\phi}+4\exp\left(-2n(\pi-1/2)^2\right).$$
Therefore, the proof is complete.
\end{proof}

\begin{proof}[Proof of Proposition \ref{prop:affinity}]
The proof is basically rewriting the argument in \cite{mossel2012stochastic} to make sure their result still holds for a wider range of parameters and also simplify some of their steps. Consider $p_{ij}$ that depends on $z$ such that $p_{ij}=a\mathbb{I}_{\{z(i)=z(j)\}}+b\mathbb{I}_{\{z(i)\neq z(j)\}}$. Similarly, define $p_{ij}'=a\mathbb{I}_{\{\sigma(i)=\sigma(j)\}}+b\mathbb{I}_{\{\sigma(i)\neq \sigma(j)\}}$. Both $\sigma$ and $z$ are uniformly sampled form $\{-1,1\}$ independently. Direct calculation gives
$$\int \left(\frac{dQ}{dP}\right)^2dP=\mathbb{E}_{\sigma,z}\prod_{1\leq i<j\leq n}\left(\frac{p_{ij}p_{ij}'}{(a+b)/2}+\frac{(1-p_{ij})(1-p_{ij}')}{1-(a+b)/2}\right).$$
Note that when $\sigma(i)\sigma(j)=z(i)z(j)$,
$$\frac{p_{ij}p_{ij}'}{(a+b)/2}+\frac{(1-p_{ij})(1-p_{ij}')}{1-(a+b)/2}=1+t,$$
and when $\sigma(i)\sigma(j)\neq z(i)z(j)$,
$$\frac{p_{ij}p_{ij}'}{(a+b)/2}+\frac{(1-p_{ij})(1-p_{ij}')}{1-(a+b)/2}=1-t,$$
where
$$t=\frac{\left(\frac{a-b}{2}\right)^2}{[(a+b)/2][1-(a+b)/2]}.$$
Thus,
$$\int \left(\frac{dQ}{dP}\right)^2dP=\mathbb{E}_{\sigma,z}(1+t)^{s_+}(1-t)^{s_-},$$
where $s_+$ and $s_-$ are the numbers of pairs $(i,j)$ such that $\sigma(i)\sigma(j)=z(i)z(j)$ and $\sigma(i)\sigma(j)\neq z(i)z(j)$, respectively. Following \cite{mossel2012stochastic}, we define $\rho=\frac{1}{n}\sum_{i=1}^n\sigma(i)z(i)$. Then,
\begin{eqnarray*}
&& \mathbb{E}_{\sigma,z}(1+t)^{s_+}(1-t)^{s_-} \\
&=& \mathbb{E}_{\sigma,z}(1+t)^{(1+\rho^2)\frac{n^2}{4}-\frac{n}{2}}(1-t)^{(1-\rho^2)\frac{n^2}{4}} \\
&=& (1-t^2)^{\frac{n^2}{4}}(1+t)^{-\frac{n}{2}}\mathbb{E}_{\sigma,z}(1+t)^{\frac{n^2\rho^2}{4}}(1-t)^{-\frac{n^2\rho^2}{4}} \\
&\leq& \mathbb{E}_{\sigma,z}(1+t)^{\frac{n^2\rho^2}{4}}(1-t)^{-\frac{n^2\rho^2}{4}} \\
&\leq& \mathbb{E}_{\sigma,z}\exp\left(\frac{n^2t}{2}\rho^2\right).
\end{eqnarray*}
By Lemma 5.5 in \cite{mossel2012stochastic}, $\mathbb{E}_{\sigma,z}\exp\left(\frac{n^2t}{2}\rho^2\right)<C$ when $nt<1-c$ for some constant $c\in(0,1)$. This leads to the desired condition in the result.
\end{proof}

\subsection{Proofs of technical lemmas}\label{sec:pf-tech}

\begin{proof}[Proof of Proposition \ref{prop:corr}]
For $T_3,T_2,T_1$ defined in Section \ref{sec:equations}, we have
\begin{eqnarray}
\label{eq:T1-rep} T_3 &=& -R_3-pR_2+(\hat{p}-p)^2(2p+\hat{p}), \\
\label{eq:T2-rep} T_2 &=& 3R_3-(1-3p)R_2+3(1-2p-\hat{p})(\hat{p}-p)^2, \\
\nonumber T_1 &=& -3R_3-(1-3p)R_2+3(\hat{p}-2p)(\hat{p}-p)^2,
\end{eqnarray}
where $R_3$ and $R_2$ are defined as $R_3^{\mathcal{M}}$ and $R_2^{\mathcal{M}}$ with $\mathcal{M}=[n]$.
By (\ref{eq:cross-m-Y}) and (\ref{eq:cross-m-W}), it is not hard to calculate that
$$\Var(R_3)=\frac{1}{{n\choose 3}}p^3(1-p^3),\quad \Var(R_2)=\frac{1}{{n\choose 3}}3p^2(1-p)^2.$$
Moreover,
$$\mathbb{E}(\hat{p}-p)^2=\frac{1}{{n\choose 2}}p(1-p).$$
Therefore, when $p=o(1)$, we have
$$T_3=-(1+o_P(1))R_3,\quad T_2=-(1+o_P(1))R_2,\quad\text{and}\quad T_1=-(1+o_P(1))R_2.$$
The relation (\ref{eq:inde-shapes}) implies that $\Corr(R_3,R_2)=0$, which directly leads to the desired conclusion.
\end{proof}

\begin{proof}[Proof of Proposition \ref{prop:delta-configuration}]
Since $\frac{1}{n}\sum_{i=1}^n\theta_i^2\geq \left(\frac{1}{n}\sum_{i=1}^n\theta_i\right)^2$, the result follows by taking cube on both sides.
\end{proof}

\begin{proof}[Proof of Proposition \ref{prop:delta-latent}]
Note that
\begin{eqnarray*}
\|\mathbb{E}g(\xi)\|^6 &=& \left(\sum_{l=1}^r\left[\mathbb{E}g_l(\xi_{l1})\right]^2\right)^3 \\
&=& \sum_{a=1}^r\sum_{b=1}^r\sum_{c=1}^r\left[\mathbb{E}g_a(\xi_{a1})\right]^2\left[\mathbb{E}g_b(\xi_{b1})\right]^2\left[\mathbb{E}g_c(\xi_{c1})\right]^2 
\end{eqnarray*}
On the other hand,
\begin{eqnarray*}
\Tr\left[\left(\mathbb{E}g(\xi)g(\xi)^T\right)^3\right] &=& \sum_{a=1}^r\sum_{b=1}^r\sum_{c=1}^r\mathbb{E}\left(g_a(\xi_{a1})g_a(\xi_{a2})g_b(\xi_{b2})g_b(\xi_{b3})g_c(\xi_{c3})g_c(\xi_{c1})\right).
\end{eqnarray*}
For each triple $(a,b,c)$, we discuss three cases. In the first case, $a=b=c$, and then
\begin{eqnarray*}
&& \mathbb{E}\left(g_a(\xi_{a1})g_a(\xi_{a2})g_b(\xi_{b2})g_b(\xi_{b3})g_c(\xi_{c3})g_c(\xi_{c1})\right) \\
&=& \left(\mathbb{E}g_a(\xi_{a1})^2\right)^3 \\
&\geq& \left[\mathbb{E}g_a(\xi_{a1})\right]^2\left[\mathbb{E}g_b(\xi_{b1})\right]^2\left[\mathbb{E}g_c(\xi_{c1})\right]^2.
\end{eqnarray*}
In the second case, $a=b\neq c$, and then
\begin{eqnarray*}
&& \mathbb{E}\left(g_a(\xi_{a1})g_a(\xi_{a2})g_b(\xi_{b2})g_b(\xi_{b3})g_c(\xi_{c3})g_c(\xi_{c1})\right) \\
&=& \left(\mathbb{E}g_a(\xi_{a1})^2\right)^2\mathbb{E}\left(g_c(\xi_{c3})g_c(\xi_{c1})\right) \\
&\geq& \left[\mathbb{E}g_a(\xi_{a1})\right]^2\left[\mathbb{E}g_b(\xi_{b1})\right]^2\left[\mathbb{E}g_c(\xi_{c1})\right]^2.
\end{eqnarray*}
Finally, when $a\neq b\neq c\neq a$,  we have
$$\mathbb{E}\left(g_a(\xi_{a1})g_a(\xi_{a2})g_b(\xi_{b2})g_b(\xi_{b3})g_c(\xi_{c3})g_c(\xi_{c1})\right)=\left[\mathbb{E}g_a(\xi_{a1})\right]^2\left[\mathbb{E}g_b(\xi_{b1})\right]^2\left[\mathbb{E}g_c(\xi_{c1})\right]^2.$$
Therefore, $\Tr\left[\left(\mathbb{E}g(\xi)g(\xi)^T\right)^3\right]\geq \|\mathbb{E}g(\xi)\|^6$. The above derivation shows that the necessary and sufficient condition for the equality to hold is $\Var(g_a(\xi_{a1}))=0$ for all $a\in[r]$, which implies a constant function.
\end{proof}

\begin{proof}[Proof of Lemma \ref{lem:var-p-q}]
We first analyze $|\hat{p}^3-\bar{p}^3|$. Direct calculation gives
$$|\hat{p}^3-\bar{p}^3|=|(\hat{p}-\bar{p})^3+3(\hat{p}-\bar{p})^2\bar{p}+3(\hat{p}-\bar{p})\bar{p}^2|.$$
Since
$$\mathbb{E}(\hat{p}-\bar{p})^2\leq \frac{1}{{n\choose 2}}\bar{p}=O\left(\frac{\bar{p}}{n^2}\right),$$
then $|\hat{p}-\bar{p}|=O_P\left(\sqrt{\bar{p}}/n\right)=O_P\left(\sqrt{\rho}/n\right)$, which implies
$$|\hat{p}^3-\bar{p}^3|=O_P\left(\frac{\rho^{5/2}}{n}+\frac{\rho^2}{n^2}+\frac{\rho^{3/2}}{n^3}\right).$$
Now we are going to analyze $|\hat{q}-\bar{q}|$. A critical decomposition formula we need is
\begin{eqnarray}
\nonumber && A_{12}A_{13}A_{23}-p_{12}p_{13}p_{23} \\
\nonumber &=& p_{13}p_{23}(A_{12}-p_{12}) + p_{12}p_{23}(A_{13}-p_{13}) + p_{12}p_{23}(A_{23}-p_{23}) \\
\nonumber && + p_{23}(A_{12}-p_{12})(A_{13}-p_{13}) + p_{13}(A_{12}-p_{12})(A_{23}-p_{23}) + p_{12}(A_{23}-p_{23})(A_{13}-p_{13}) \\
\label{eq:critical} && + (A_{12}-p_{12})(A_{13}-p_{13})(A_{23}-p_{23}).
\end{eqnarray}
Given a triple $\lambda=(i,j,k)$, define
\begin{eqnarray*}
\wt{Y}_{\lambda} &=& (A_{ij}-p_{ij})(A_{ik}-p_{ik})(A_{jk}-p_{jk}), \\
\wt{W}_{\lambda} &=& p_{ik}(A_{ij}-p_{ij})(A_{jk}-p_{jk})+p_{ij}(A_{ik}-p_{ik})(A_{jk}-p_{jk})+p_{jk}(A_{ij}-p_{ij})(A_{ik}-p_{ik}).
\end{eqnarray*}
By the definitions, it is easy to see that
$$\mathbb{E}\wt{Y}_{\lambda}\wt{Y}_{\lambda'}=\begin{cases}
p_{ij}(1-p_{ij})p_{ik}(1-p_{ik})p_{jk}(1-p_{jk}), & \lambda=\lambda',\\
0, & \lambda\neq \lambda',
\end{cases}$$
and
$$\mathbb{E}\wt{W}_{\lambda}\wt{W}_{\lambda'}=\begin{cases}
p_{ij}p_{jk}p_{ik}\left[p_{ik}(1-p_{ij})(1-p_{ik})+p_{ij}(1-p_{jk})(1-p_{ik})+p_{jk}(1-p_{ij})(1-p_{jk})\right], & \lambda=\lambda',\\
0, & \lambda\neq \lambda'.
\end{cases}$$
The decomposition (\ref{eq:critical}) implies
\begin{eqnarray*}
\hat{q}-\bar{q} &=& \sum_{1\leq i<j\leq n}\left[\frac{1}{{n\choose 3}}\sum_{\{k\in[n]:k\neq i,j\}} p_{ik}p_{jk}\right](A_{ij}-p_{ij}) \\
&& + \frac{1}{{n\choose 3}}\sum_{1\leq i<j<k\leq n}\wt{Y}_{\lambda} + \frac{1}{{n\choose 3}}\sum_{1\leq i<j<k\leq n}\wt{W}_{\lambda}.
\end{eqnarray*}
We will bound the variance of each of the three terms. For the first term,
\begin{eqnarray*}
&& \Var\left[\sum_{1\leq i<j\leq n}\left[\frac{1}{{n\choose 3}}\sum_{\{k\in[n]:k\neq i,j\}} p_{ik}p_{jk}\right](A_{ij}-p_{ij})\right] \\
&=& \sum_{1\leq i<j\leq n}\left[\frac{1}{{n\choose 3}}\sum_{\{k\in[n]:k\neq i,j\}} p_{ik}p_{jk}\right]^2\Var(A_{ij}-p_{ij}) \\
&=& O\left(\frac{\rho^5}{n^2}\right).
\end{eqnarray*}
For the second term,
\begin{eqnarray*}
\Var\left[\frac{1}{{n\choose 3}}\sum_{1\leq i<j<k\leq n}\wt{Y}_{\lambda}\right] &=& \frac{1}{{n\choose 3}^2}\sum_{1\leq i<j<k\leq n}\mathbb{E}\wt{Y}_{\lambda}^2 \\
&=& O\left(\frac{\rho^3}{n^3}\right).
\end{eqnarray*}
For the third term,
\begin{eqnarray*}
\Var\left[\frac{1}{{n\choose 3}}\sum_{1\leq i<j<k\leq n}\wt{W}_{\lambda}\right] &=& \frac{1}{{n\choose 3}^2}\sum_{1\leq i<j<k\leq n}\mathbb{E}\wt{W}_{\lambda}^2 \\
&=& O\left(\frac{\rho^4}{n^3}\right).
\end{eqnarray*}
Combining the variance bounds for the three terms, we obtain
$$|\hat{q}-\bar{q}|=O_P\left(\frac{\rho^{5/2}}{n}+\frac{\rho^{3/2}}{n^{3/2}}\right).$$
Thus, the proof is complete.
\end{proof}

\begin{proof}[Proof of Lemma \ref{lem:var-p-q-latent}]
We first study $|\hat{p}-\tilde{p}|$. The variance has decomposition
\begin{equation}
\mathbb{E}(\hat{p}-\tilde{p})^2 = \mathbb{E}\Var(\hat{p}|\xi) + \Var(\mathbb{E}(\hat{p}|\xi)).\label{eq:var-d}
\end{equation}
The first term $\mathbb{E}\Var(\hat{p}|\xi)$ has bound
$$\mathbb{E}\Var(\hat{p}|\xi)\leq \frac{1}{{n\choose 2}^2}\sum_{1\leq i<j\leq n}\mathbb{E}(g(\xi_i)^Tg(\xi_j))=O\left(\frac{\|g\|_{\infty}^2}{n^2}\right).$$
We study the second term $\Var(\mathbb{E}(\hat{p}|\xi))$. Note that the conditional expectation is
$$\mathbb{E}(\hat{p}|\xi)=\frac{1}{{n\choose 2}}\sum_{1\leq i<j\leq n}g(\xi_i)^Tg(\xi_j)=\left\|\frac{1}{n}\sum_{i=1}^ng(\xi_i)\right\|^2-\frac{1}{n^2}\sum_{i=1}^n\|g(\xi_i)\|^2.$$
Since $\mathbb{E}\hat{p}=\|\mathbb{E}g(\xi)\|^2$, we have
\begin{eqnarray*}
\Var(\mathbb{E}(\hat{p}|\xi)) &=& \mathbb{E}\left(\left\|\frac{1}{n}\sum_{i=1}^ng(\xi_i)\right\|^2-\frac{1}{n^2}\sum_{i=1}^n\|g(\xi_i)\|^2-\|\mathbb{E}g(\xi)\|^2\right)^2 \\
&\leq& 2\mathbb{E}\left(\left\|\frac{1}{n}\sum_{i=1}^ng(\xi_i)\right\|^2-\|\mathbb{E}g(\xi)\|^2\right)^2 + 2\left(\frac{1}{n}\|g\|_{\infty}^2\right)^2 \\
&=& 2\mathbb{E}\left(\left(\frac{1}{n}\sum_{i=1}^ng(\xi_i)-\mathbb{E}g(\xi)\right)^T\left(\frac{1}{n}\sum_{i=1}^ng(\xi_i)+\mathbb{E}g(\xi)\right)\right)^2 + 2\left(\frac{1}{n}\|g\|_{\infty}^2\right)^2 \\
&\leq& 8\|g\|_{\infty}^2\mathbb{E}\left\|\frac{1}{n}\sum_{i=1}^ng(\xi_i)-\mathbb{E}g(\xi)\right\|^2 + 2\left(\frac{1}{n}\|g\|_{\infty}^2\right)^2 \\
&\leq& 8\frac{\|g\|_{\infty}^2}{n}{\sum_{l=1}^r\|g_l\|_{\infty}^2}  + 2\left(\frac{1}{n}\|g\|_{\infty}^2\right)^2 \\
&\leq& \frac{8\|g\|_{\infty}^4}{n} + \frac{2\|g\|_{\infty}^4}{n^2}.
\end{eqnarray*}
Then, the variance decomposition (\ref{eq:var-d}) implies
$$\mathbb{E}(\hat{p}-\tilde{p})^2=O\left(\frac{\|g\|_{\infty}^4}{n} + \frac{\|g\|_{\infty}^2}{n^2}\right),$$
which leads to
$$|\hat{p}^3-\tilde{p}^3|=O_P\left(\frac{\|g\|_{\infty}^3}{n^3}+\frac{\|g\|_{\infty}^4}{n^2}+\frac{\|g\|_{\infty}^5}{n}+\frac{\|g\|_{\infty}^6}{\sqrt{n}}\right).$$

Now we study $|\hat{q}-\tilde{q}|$. The variance has decomposition
\begin{equation}
\mathbb{E}(\hat{q}-\tilde{q})^2=\mathbb{E}\Var(\hat{q}|\xi)+\Var(\mathbb{E}(\hat{q}|\xi)).\label{eq:var-q-d}
\end{equation}
Note that
\begin{eqnarray*}
\mathbb{E}(\hat{q}|\xi) &=& \frac{1}{{n\choose 3}}\sum_{1\leq i<j<k\leq n}g(\xi_i)^Tg(\xi_j)g(\xi_j)^Tg(\xi_k)g(\xi_k)^Tg(\xi_i) \\
&=& \Tr\left(\frac{1}{{n\choose 3}}\sum_{1\leq i<j<k\leq n}g(\xi_i)g(\xi_i)^Tg(\xi_j)g(\xi_j)^Tg(\xi_k)g(\xi_k)^T\right) \\
&=& \Tr\left(\frac{1}{{n\choose 3}}\sum_{1\leq i<j<k\leq n}X_iX_jX_k\right),
\end{eqnarray*}
where we have used the notation $X_i=g(\xi_i)g(\xi_i)^T$. Therefore,
\begin{eqnarray*}
&& \left|\mathbb{E}(\hat{q}|\xi)-\Tr\left(\left(\frac{1}{n}\sum_{i=1}^nX_i\right)^3\right)\right| \\
&=&\left|\Tr\left(\frac{1}{{n\choose 3}}\sum_{i<j<k}X_iX_jX_k\right)-\Tr\left(\left(\frac{1}{n}\sum_{i=1}^nX_i\right)^3\right)\right| \\
&\leq& \frac{C}{n}\max_{1\leq i\leq n}\fnorm{X_i}^{3} \\
&\leq& \frac{C\|g\|_{\infty}^6}{n}.
\end{eqnarray*}
Hence,
$$
\Var(\mathbb{E}(\hat{q}|\xi)) \leq 2\mathbb{E}\left|\Tr\left(\left(\frac{1}{n}\sum_{i=1}^nX_i\right)^3\right)-\Tr\left(\left(\frac{1}{n}\sum_{i=1}^n\mathbb{E}X_i\right)^3\right)\right|^2 + \frac{C'\|g\|_{\infty}^{12}}{n^2}.
$$
Consider the matrix
$$U=\frac{1}{n}\sum_{i=1}^nX_i,\quad V=\frac{1}{n}\sum_{i=1}^n\mathbb{E}X_i.$$
Then, it is easy to check that
\begin{eqnarray*}
\left|\Tr(U^3)-\Tr(V^3)\right| &=& \left|\Tr((U-V)^3) + 3\Tr((U-V)V^2) + 3\Tr((U-V)^2V)\right| \\
&\leq& \fnorm{U-V}\left(\fnorm{U-V}^2+\fnorm{V}^2+\fnorm{U-V}\fnorm{V}\right).
\end{eqnarray*}
Thus
\begin{eqnarray*}
&& \left|\Tr\left(\left(\frac{1}{n}\sum_{i=1}^nX_i\right)^3\right)-\Tr\left(\left(\frac{1}{n}\sum_{i=1}^n\mathbb{E}X_i\right)^3\right)\right| \\
&\leq& \left|\Tr\left(\frac{1}{n}\sum_{i=1}^n(X_i-\mathbb{E}X_i)\right)\right|\|g\|_{\infty}^4,
\end{eqnarray*}
which implies
\begin{eqnarray*}
\Var(\mathbb{E}(\hat{q}|\xi)) &\leq& 2\mathbb{E}\left|\Tr\left(\frac{1}{n}\sum_{i=1}^n(X_i-\mathbb{E}X_i)\right)\right|^2\|g\|_{\infty}^8 + 2\frac{C^2\|g\|_{\infty}^{12}}{n^2} \\
&\leq& \frac{2\|g\|_{\infty}^8\mathbb{E}|\Tr(X_1)|^2}{n} + 2\frac{C^2\|g\|_{\infty}^{12}}{n^2} \\
&=& O\left(\frac{\|g\|_{\infty}^{12}}{n}\right).
\end{eqnarray*}
Similar to the previous analysis in the proof of Lemma \ref{lem:var-p-q}, we also have
$$\mathbb{E}\Var(\hat{q}|\xi)=O\left(\frac{\|g\|_{\infty}^{10}}{n^2}+\frac{\|g\|_{\infty}^{6}}{n^3}\right).$$
Hence, due to (\ref{eq:var-q-d}), we have
$$|\hat{q}-\tilde{q}|=O_P\left(\frac{\|g\|_{\infty}^{6}}{\sqrt{n}}+\frac{\|g\|_{\infty}^5}{n}+\frac{\|g\|_{\infty}^3}{n^{3/2}}\right).$$
The proof is complete.
\end{proof}

\begin{proof}[Proof of Lemma \ref{lem:var-p-q-M}]
We first study $|\hat{p}^{\mathcal{M}}-\bar{p}|$. It is not hard to see that $\mathbb{E}\hat{p}^{\mathcal{M}}=\bar{p}$. The variance has decomposition
$$\mathbb{E}(\hat{p}^{\mathcal{M}}-\bar{p})^2=\mathbb{E}\Var(\hat{p}^{\mathcal{M}}|A)+\Var(\mathbb{E}(\hat{p}^{\mathcal{M}}|A)).$$
Note that $\mathbb{E}(\hat{p}^{\mathcal{M}}|A)=\hat{p}$. Thus, $\Var(\mathbb{E}(\hat{p}^{\mathcal{M}}|A))=O(n^{-2}\rho)$. To study $\hat{p}^{\mathcal{M}}$, note that it is the average of $m$ uniform draws without replacement from
$$\left\{\frac{1}{{n-1\choose 2}}\sum_{\{1\leq j<k\leq n: j,k\neq i\}}\frac{A_{ij}+A_{ik}+A_{jk}}{3}\right\}_{i\in[n]}.$$
The above set is denoted as $\{H_i\}_{i\in[n]}$. Then,
\begin{eqnarray*}
\mathbb{E}\Var(\hat{p}^{\mathcal{M}}|A) &=& \frac{n-m}{n(m-1)}\mathbb{E}\left[\frac{1}{n}\sum_{i=1}^nH_i^2-\left(\frac{1}{n}\sum_{i=1}^nH_i\right)^2\right] \\
&\leq& \frac{n-m}{n(m-1)}\frac{1}{n}\sum_{i=1}^n\mathbb{E}H_i^2 \\
&=&  \frac{n-m}{n(m-1)}\frac{1}{n}\sum_{i=1}^n\left(\Var(H_i)+(\mathbb{E}H_i)^2\right) \\
&=& O\left(\frac{n-m}{n(m-1)}\left(\rho^2+\frac{\rho}{n}\right)\right).
\end{eqnarray*}
Hence,
$$(\hat{p}^{\mathcal{M}}-\bar{p})^2=O_P\left(\frac{n-m}{n(m-1)}\left(\rho^2+\frac{\rho}{n}\right) + \frac{\rho}{n^2}\right)=O_P\left(\frac{\rho^2}{m}+\frac{\rho}{mn}\right).$$
This implies that
$$|(\hat{p}^{\mathcal{M}})^3-\bar{p}^3|=O_P\left(\frac{\rho^{3/2}}{(mn)^{3/2}}+\frac{\rho^2}{mn}+\frac{\rho^{5/2}}{(mn)^{1/2}}+\frac{\rho^3}{m^{1/2}}\right).$$

We continue to study $|\hat{q}^{\mathcal{M}}-\bar{q}|$. Again, we have $\mathbb{E}\hat{q}^{\mathcal{M}}=\bar{q}$ and
$$\mathbb{E}(\hat{q}^{\mathcal{M}}-\bar{q})^2=\mathbb{E}\Var(\hat{q}^{\mathcal{M}}|A)+\Var(\mathbb{E}(\hat{q}^{\mathcal{M}}|A)).$$
Note that $\mathbb{E}(\hat{q}^{\mathcal{M}}|A)=\hat{q}$. Thus, $\Var(\mathbb{E}(\hat{q}^{\mathcal{M}}|A))=O\left(\frac{\rho^5}{n^2}+\frac{\rho^3}{n^3}\right)$, which was derived in the proof of Lemma \ref{lem:var-p-q}. To study $\hat{q}^{\Delta}$, note that it is the average of $m$ uniform draws without replacement from
$$\left\{\frac{1}{{n-1\choose 2}}\sum_{\{1\leq j<k\leq n: j,k\neq i\}}A_{ij}A_{ik}A_{jk}\right\}_{i\in[n]}.$$
The above set is denoted as $\{G_i\}_{i\in[n]}$. Then,
\begin{eqnarray*}
\mathbb{E}\Var(\hat{p}^{\mathcal{M}}|A) &=& \frac{n-m}{n(m-1)}\mathbb{E}\left[\frac{1}{n}\sum_{i=1}^nG_i^2-\left(\frac{1}{n}\sum_{i=1}^nG_i\right)^2\right] \\
&\leq& \frac{n-m}{n(m-1)}\frac{1}{n}\sum_{i=1}^n\mathbb{E}G_i^2 \\
&=&  \frac{n-m}{n(m-1)}\frac{1}{n}\sum_{i=1}^n\left(\Var(G_i)+(\mathbb{E}G_i)^2\right) \\
&=& O\left(\frac{n-m}{n(m-1)}\left(\rho^6+\frac{\rho^5}{n}+\frac{\rho^3}{n^2}\right)\right) \\
&=& O\left(\frac{\rho^6}{m}+\frac{\rho^5}{mn}+\frac{\rho^3}{mn^2}\right),
\end{eqnarray*}
where the calculation of $\Var(G_i)$ is with the help of the decomposition (\ref{eq:critical}). Thus, the proof is complete
\end{proof}

\begin{proof}[Proof of Lemma \ref{lem:var-p-q-latent-M}]
Note the variance decomposition
$$\mathbb{E}(\hat{p}^{\mathcal{M}}-\tilde{p})^2=\mathbb{E}\Var(\hat{p}^{\mathcal{M}}|A)+\Var(\mathbb{E}(\hat{p}^{\mathcal{M}}|A)),$$
and
$$\mathbb{E}(\hat{q}^{\mathcal{M}}-\tilde{q})^2=\mathbb{E}\Var(\hat{q}^{\mathcal{M}}|A)+\Var(\mathbb{E}(\hat{q}^{\mathcal{M}}|A)).$$
The first terms of the two equations have been studied in the proof of Lemma \ref{lem:var-p-q-M}, where $\rho$ is replaced by $\|g\|_{\infty}^2$. The second terms of the two equations have been studied in the proof of Lemma \ref{lem:var-p-q-latent}. We, therefore, obtain the desired conclusion.
\end{proof}

\begin{proof}[Proof of Lemma \ref{lem:var-p-q-Delta}]
We first study $|\hat{p}^{\Delta}-\bar{p}|$. It is not hard to see that $\mathbb{E}\hat{p}^{\Delta}=\bar{p}$. The variance has decomposition
$$\mathbb{E}(\hat{p}^{\Delta}-\bar{p})^2=\mathbb{E}\Var(\hat{p}^{\Delta}|A)+\Var(\mathbb{E}(\hat{p}^{\Delta}|A)).$$
Note that $\mathbb{E}(\hat{p}^{\Delta}|A)=\hat{p}$. Thus, $\Var(\mathbb{E}(\hat{p}^{\Delta}|A))=O(n^{-2}\rho)$. To study $\hat{p}^{\Delta}$, note that it is the average of $|\Delta|$ uniform draws with replacement from
$$\left\{\frac{A_{ij}+A_{ik}+A_{jk}}{3}\right\}_{1\leq i<j<k\leq n}.$$
Then,
\begin{eqnarray*}
\mathbb{E}\Var(\hat{p}^{\Delta}|A) &\leq& \frac{1}{|\Delta|}\frac{1}{{n\choose 3}}\sum_{1\leq i<j<k\leq n}\mathbb{E}\left(\frac{A_{ij}+A_{ik}+A_{jk}}{3}\right)^2 \\
&=& O\left(\frac{\rho}{|\Delta|}\right).
\end{eqnarray*}
Hence,
$$(\hat{p}^{\Delta}-\bar{p})^2=O_P\left(\frac{\rho}{n^2}+\frac{\rho}{|\Delta|}\right),$$
which implies
$$|\hat{p}^{\Delta}-\bar{p}|^3=O_P\left(\frac{\rho^{5/2}}{n}+\frac{\rho^2}{n^2}+\frac{\rho^{3/2}}{n^3}+\frac{\rho^{5/2}}{|\Delta|^{1/2}}+\frac{\rho^2}{|\Delta|}+\frac{\rho^{3/2}}{|\Delta|^{3/2}}\right).$$

We continue to study $|\hat{q}^{\Delta}-\bar{q}|$. Again, we have $\mathbb{E}\hat{q}^{\Delta}=\bar{q}$ and
$$\mathbb{E}(\hat{q}^{\Delta}-\bar{q})^2=\mathbb{E}\Var(\hat{q}^{\Delta}|A)+\Var(\mathbb{E}(\hat{q}^{\Delta}|A)).$$
Note that $\mathbb{E}(\hat{q}^{\Delta}|A)=\hat{q}$. Thus, $\Var(\mathbb{E}(\hat{q}^{\Delta}|A))=O\left(\frac{\rho^5}{n^2}+\frac{\rho^3}{n^3}\right)$, which was derived in the proof of Lemma \ref{lem:var-p-q}. To study $\hat{q}^{\Delta}$, note that it is the average of $|\Delta|$ uniform draws with replacement from
$$\left\{A_{ij}A_{ik}A_{jk}\right\}_{1\leq i<j<k\leq n}.$$
Then,
\begin{eqnarray*}
\mathbb{E}\Var(\hat{p}^{\Delta}|A) &\leq& \frac{1}{|\Delta|}\frac{1}{{n\choose 3}}\sum_{1\leq i<j<k\leq n}\mathbb{E}\left(A_{ij}A_{ik}A_{jk}\right)^2 \\
&=& O\left(\frac{\rho^3}{|\Delta|}\right).
\end{eqnarray*}
Hence,
$$|\hat{q}^{\Delta}-\bar{q}|=O_P\left(\frac{\rho^{5/2}}{n}+\frac{\rho^{3/2}}{n^{3/2}}+\frac{\rho^{3/2}}{|\Delta|^{1/2}}\right).$$
Thus, the proof is complete.
\end{proof}

\begin{proof}[Proof of Lemma \ref{lem:var-p-q-latent-Delta}]
The proof is the same as that of Lemma \ref{lem:var-p-q-latent-M}.
\end{proof}

\section*{Acknowledgements}

The authors thank Fengnan Gao for suggesting the martingale central limit theorem in \cite{hall2014martingale}, and Scarlett Li for help with the simulations.

\setlength{\bibsep}{8pt}
\bibliographystyle{apalike}
\bibliography{ref}

\end{document}